\documentclass[12pt,a4paper,oneside]{book}

\usepackage[utf8]{inputenc}
\usepackage[english]{babel}
\usepackage{verse}
\usepackage{appendix}
\usepackage{titlesec}
\usepackage{pgfornament}
\usepackage{fancyhdr}
\usepackage{amsmath,amsthm,amssymb,bbm,enumerate}
\usepackage{fancyref}
\usepackage{lastpage}
\usepackage{cite}
\usepackage{pdfpages}
\usepackage{graphicx}
\usepackage{booktabs}
\usepackage{float}
\usepackage{tikz,pgffor}
\usepackage{hyperref}

\usetikzlibrary{decorations.pathmorphing, decorations.markings, patterns, intersections, decorations.text, calc}

\restylefloat{figure}

\newcommand{\ie}{\textit{i.e. }} 
\newcommand{\eg}{\textit{e.g. }} 

\newcommand{\ket}[1]{\left |~#1~\right\rangle}
\newcommand{\bra}[1]{\left\langle~#1~\right |}
\newcommand{\ked}[1]{\left |~#1~\right )}
\newcommand{\brd}[1]{\left (~#1~\right |}

\newcommand{\braked}[1]{\left\langle #1 \right)}
\newcommand{\brdked}[1]{\left( #1 \right)}
\newcommand{\brdket}[1]{\left( #1 \right\rangle}
\newcommand{\braket}[1]{\left\langle #1 \right\rangle}
\newcommand{\abs}[1]{\left|#1\right|}

\newcommand{\one}{\mathbbm{1}}

\newcommand{\UCNOT}{U_{\text{CNOT}}}
\newcommand{\UQT}{U_{\text{QT}}}

\newcommand{\fin}{f^{\text{in}}}
\newcommand{\fins}{f^{\text{in},*}}
\newcommand{\fout}{f^{\text{out}}}
\newcommand{\fouts}{f^{\text{out},*}}

\newcommand{\uin}{u_{\text{in}}}
\newcommand{\uout}{u_{\text{out}}}
\renewcommand{\vin}{v_{\text{in}}}
\newcommand{\vout}{v_{\text{out}}}

\newcommand{\ain}{a^{\text{in}}}
\newcommand{\aind}{a^{\text{in},\dagger}}
\newcommand{\aout}{a^{\text{out}}}
\newcommand{\aoutd}{a^{\text{out},\dagger}}

\newcommand{\aintd}{a^{\text{int},\dagger}}

\newcommand{\invac}{\ket{\text{in}}}
\newcommand{\outvac}{\ket{\text{out}}}

\newcommand{\pd}[2]{\frac{\partial #1}{\partial #2}}

\usepackage[split]{splitidx}
\makeindex
\newindex[Index of Concepts]{concepts}
\newindex[Index of People]{people}
\newcommand{\ic}[1]{\sindex[concepts]{#1}}
\newcommand{\ip}[1]{\sindex[people]{#1}}
\newcommand{\co}[1]{#1\sindex[concepts]{#1}}

\newcommand{\reals}{\mathbb{R}}
\newcommand{\comps}{\mathbb{C}}

\DeclareMathOperator{\tr}{tr}
\DeclareMathOperator{\diag}{diag}
\DeclareMathOperator{\sgn}{sgn}

\newtheorem{thm}{Theorem}

\newcommand{\mc}[1]{\mathcal #1}
\newcommand{\ve}[0]{\varepsilon}

\author{Furkan Sem\.{ı}h Dündar}
\title{{\Huge The Firewall Paradox}\\{\small Second Edition}}
\date{}

\begin{document}
\titleformat{\chapter}[hang]{\Huge}{\centering}{0pt}{\Huge\centering}
\frontmatter

\includepdfmerge{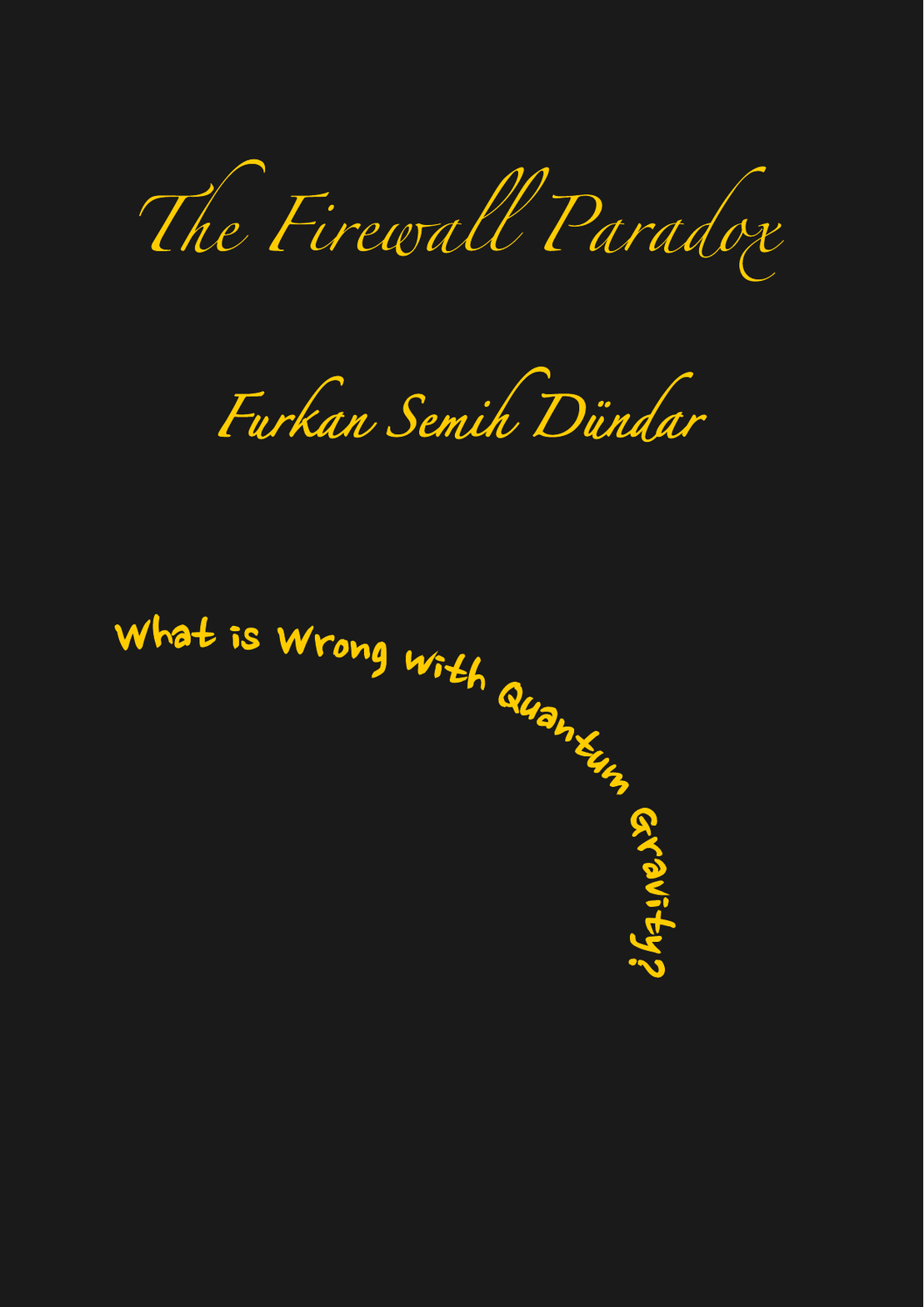}
\includepdfmerge{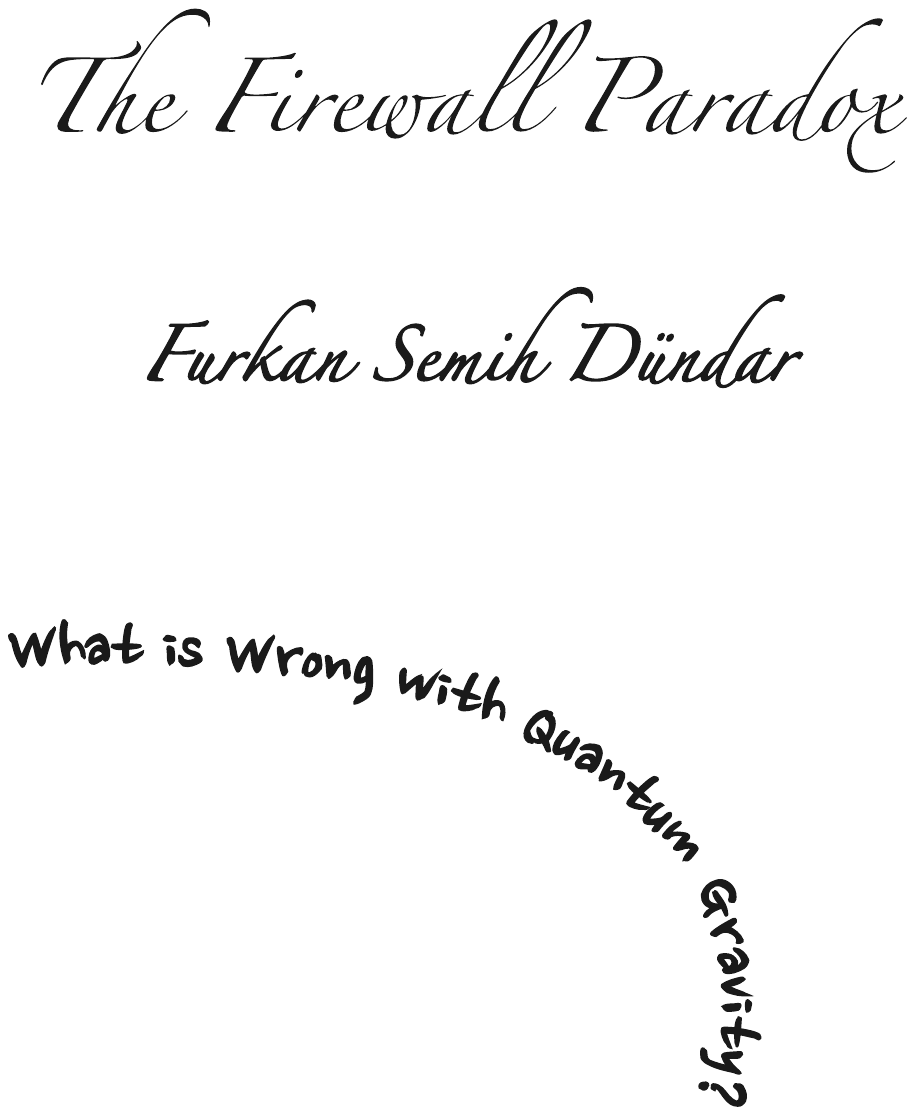}

\cleardoublepage

\hyphenation{Sch-warz-sch-ild}

\cleardoublepage
\phantomsection
\addcontentsline{toc}{chapter}{Abstract}

\vspace*{\fill}
\begin{center}
  \textbf{Abstract}
\end{center}

In this MSc. thesis, we have attempted to give an overview of the firewall paradox and various approaches towards its resolution. After an introductory chapter on some basic concepts in quantum field theory in curved spacetimes such as Hawking radiation, we introduce the paradox. It arises out of application of principles each of which is thought or assumed to be correct: 1) unitary black hole evaporation, 2) validity of quantum field theory in curved spacetime, 3) a measure of the number of black hole quantum states, 4) Einstein's equivalence principle. Then, we present various approaches that exist in the literature towards the resolution of the paradox.

\vspace*{\fill}

\begin{center}
  \textbf{Özet}
\end{center}

Bu yüksek lisans tezinde, ateşten set paradoksu ve onun çözümüne yönelik çeşitli yaklaşımları anlatma çabası içinde bulunuldu. Eğri uzay-zamanlarda kuantum alanlar kuramındaki, Hawking ışınımı gibi, bazı temel kavramlar üzerine olan bir giriş kısmından sonra ateşten set paradoksu anlatılıyor. Paradoks ayrı ayrı doğru olduğu düşünülen ya da sanılan ilkelerin hep beraberce uygulanması sonucunda doğuyor: 1) kara deliklerin kuantum mekaniği ile uyumlu biçimde buharlaşmaları, 2) eğri uzay-zamanlarda kuantum alanlar teorisinin geçerliliği, 3) kara deliklerin kuantum hallerini saymak için bir ölçü, 4) Einstein'ın eşdeğerlilik ilkesi. Sonrasında paradoksun çözümüne yönelik literatürde yer alan çeşitli yaklaşımlara yer veriliyor.

\vspace*{\fill}

%
\cleardoublepage
\phantomsection
\addcontentsline{toc}{chapter}{Dedication}

\vspace*{\fill}
  \textit{\hspace{3em}to Fuzûlî\ip{Fuzûlî}, Nietzsche\ip{Nietzsche, Friedrich
      Wilhelm}, Şule Gürbüz\ip{Gürbüz, Şule}, Palahniuk\ip{Palahniuk,
      Chuck}, Hawking\ip{Hawking, Stephen W.}, and Lady Gaga\ip{Lady
      Gaga} for showing different depths of the ocean called life. Or
    is it just a droplet?}
\vspace*{\fill}

\cleardoublepage
\phantomsection
\addcontentsline{toc}{chapter}{A Poem from the Era of Love}


\vspace*{\fill}

\begin{center}
  \begin{tikzpicture}[every node/.style={inner sep=0pt}]
    \node [align=center] (siir) {
      Tûti-i mûcize-gûyem ne desem lâf değil\\
      Çerh ile söyleşemem âyinesi sâf değil\\
      \\
      Ehl-i dildir diyemem sînesi sâf olmayana\\
      Ehl-i dil birbirini bilmemek insâf değil\\
      \\
      Yine endîşe bilir kadr-i dür-i güftârım\\
      Rûzgâr ise denî dehr ise sarrâf değil\\
      \\
      Girdi miftâh-ı der-i genc-i maânî elime\\
      Âleme bezl-i güher eylesem itlâf değil\\
      \\
      Levh-i Mahfûz-ı sühandır dil-i pâk-i Nef'î\\
      Tab'-ı yârân gibi dükkânçe-i sahhâf değil\\
    };

    \node [shift={(-1.5em,1.5em)}] (nw) at (siir.north west) {\pgfornament[width = 1.5cm]{39}};
    \node [shift={(1.5em,1.5em)}] (ne) at (siir.north east) {\pgfornament[width = 1.5cm, symmetry=v]{39}};
    \node [shift={(-1.5em,0em)}] (sw) at (siir.south west) {\pgfornament[width = 1.5cm, symmetry=h]{39}};
    \node [shift={(1.5em,0em)}] (se) at (siir.south east) {\pgfornament[width = 1.5cm, symmetry=c]{39}};
    
    \pgfornamentvline{nw}{sw}{west}{88}
    \pgfornamentvline{ne}{se}{east}{88}
    \pgfornamenthline{sw}{se}{south}{88}
    \pgfornamenthline{nw}{ne}{north}{88}
  \end{tikzpicture}
\end{center}






\begin{flushright}
  --Nef'î \cite{ipala}
\end{flushright}
\vspace*{\fill}

%
%

\cleardoublepage
\phantomsection
\addcontentsline{toc}{chapter}{Acknowledgements}

\vspace*{\fill}

I would like to thank my brother E. Burak Dündar\ip{Dündar, Enes Burak} for valuable discussions that led to the inclusion of the index as well as for artistic advice in the color choice for the cover page; Tahsin Çağrı Şişman\ip{Sisman@Şişman, Tahsin Çağrı} for lively debates; Özgür Kelekçi\ip{Kelekçi, Özgür} for unbounded discussions ranging from life to philosophy and science; Sadi Turgut\ip{Turgut, Sadi} and Özenç Güngör\ip{Güngör, Özenç} for illuminating discussions on quantum information; Ümit Alkuş\ip{Alkuş, Ümit} for sharing the lecture notes he took in quantum information classes; to Çağatay Menekay\ip{Menekay, Çağatay} for valuable discussions especially in regard to positive frequency of mode solutions; Bahtiyar Özgür Sarıoğlu\ip{Sarıoğlu, Bahtiyar Özgür} for valuable discussions; Sabine Hossenfelder\ip{Hossenfelder, Sabine} for useful comments that lead to my better understanding of entanglement of Hawking radiation from the perspective of an infalling observer; Daniel Harlow\ip{Harlow, Daniel} for a comment on the Harlow-Hayden conjecture; Dejan Stojkovic\ip{Stojkovic, Dejan} for valuable discussions on icezones; Samuel L. Braunstein\ip{Braunstein, Samuel L.} for pointing out his precursory idea of firewalls; members of METU-GR-HEP email group for providing a lively atmosphere for scientific discussions. Lastly, I would like to thank my advisor Bayram Tekin\ip{Tekin, Bayram} for the boldness and courage he presented in the choice of this thesis topic as well as bearing the trouble of having a student like myself. I mentioned him at the end, because the ones who come last and are present everywhere however unseen, constitute the basis of everything.

I would also like to thank TÜBİTAK for their graduate scholarship that I benefited during my studies in the last two years. Without this support, I would have to work in extra jobs and hence would have much less time to devote to my studies. I am grateful.

\vspace*{\fill}

\chapter*{Abbreviations}
\addcontentsline{toc}{chapter}{Abbreviations}

\begin{itemize}
\item [$\eta_{\mu \nu}$] Metric tensor (flat)
\item [$g_{\mu \nu}$] Metric tensor (general)
\item [$g$] Metric determinant ($g \equiv \det(g_{\mu \nu})$)
\item [$\partial_\mu$] Partial derivative
\item [$\nabla_\mu$] Covariant derivative (metric compatible)
\item [$\mathcal{L}$] Lagrangian density (usually referred to as only ``Lagrangian'')
\item [$S$] Action ($S \equiv \int d^n x |g|^{1/2} \mathcal{L}$)
\item [$\mathbb Z$] Set of integers
\item [$\mathbb Z^{\geq 0}$] Set of non-negative integers
\item [$\mathbb R$] Set of real numbers
\item [$\mathbb R^*$] Set of non-zero real numbers
\end{itemize}

We use Einstein's summation convention throughout the text unless otherwise
indicated. For example, $u^\mu v_\mu$ stands for $\sum_\mu u^\mu v_\mu$. Unless explicitly specified, we use the units in which $m_{\text{Planck}} = \hbar = c = 1$ holds.

\cleardoublepage
\tableofcontents

\mainmatter

\fancyhead[C]{\leftmark}


\chapter{Introduction to the Firewall Paradox}
\label{chap:introtoparadox}


Science has various ways of improving itself. The urge or need to explain observations of new phenomena is the most direct example that comes to mind. However, in the absence of access to regimes where new phenomena can be observed, paradoxes found in gedanken experiments are quite valuable. They force scientists to re-consider the basics on which they have depended so far. Upon this reconsideration, science can become capable of yielding better explanations of nature.


The \emph{firewall paradox}\ic{firewall paradox} has been introduced in the article \cite{Almheiri2012} by Ahmed Almheiri\ip{Almheiri, Ahmed}, Donald Marolf\ip{Marolf, Donald}, Joseph Polchinski\ip{Polchinski, Joseph} and James Sully\ip{Sully, James} (AMPS) that appeared on the arXiv on the 13th of July, 2012. This was the introduction of firewalls\footnote{However, the idea of a firewall has become \emph{imaginable} three years before in 2009: the ``energetic curtain'' in Samuel Braunstein's\ip{Braunstein, Samuel L.} \cite{Braunstein2009} is a precursor to firewalls. With the addition of Stefano Pirandola\ip{Pirandola, Stefano} and Karol Życzkowski\ip{Zyczkowski@Życzkowski, Karol} as authors, Braunstein's\ip{Braunstein, Samuel L.} work was published \cite{Braunstein2013}.}.

In the most basic terms, the firewall paradox is as follows. If black holes evaporate unitarily as expected by quantum mechanics, they become almost maximally entangled with the radiation they have emitted so far \cite{PageBHInfo}. Hence, newly emitted Hawking particles are almost maximally entangled with early radiation. However, these newly emitted quanta cannot \emph{also} be entangled with interior modes, which would otherwise violate basic principles of quantum mechanics. Because of the lack of entanglement in the latter case, quantum state around the horizon cannot be vacuum. On the other hand, equivalence principle requires that the place of event horizon cannot be determined locally: the region of spacetime around the horizon that an infalling observer passes through is not locally different from any other region of spacetime and is in Minkowski vacuum state. Therefore, equivalence principle together with the accepted wisdom\footnote{We mean \emph{black hole complementarity}\ic{black hole complementarity}. We discuss this idea in Section~\ref{sec:compl-black-hole}.} about black hole evaporation are not consistent. Infalling observers detect particles of high energy at the horizon \cite{Almheiri2012}, hence the name firewall.


In order to better understand the paradox in quantitative terms, we need some basic results and ideas, such as Hawking radiation and black hole complementarity, from quantum gravity. We deal with these concepts in Chapter~\ref{chap:intro}. Having acquired the basics, we focus our attention on the paradox in Chapter~\ref{chap:the-paradox}. Various proposals towards the resolution of the firewall paradox are included in Chapter~\ref{chap:mult-approaches}. We finalize the thesis in Chapter~\ref{chap:conclusion} by giving a conclusion.

\chapter{Basics}
\label{chap:intro}

Our aim in this chapter is to give a minimalistic account of various concepts we will later refer to and use in the text.

\section{General Relativity}
\label{sec:gr}

The general theory of relativity is the standard theory of gravitation we use today. It is the first of two pillars of modern physics, the second being the quantum theory. Gravitation is an interesting phenomenon, and perhaps it would not be rather shallow to claim, as it might have been expressed by others, that although it is the first force discovered in nature it is the least understood.

In general relativity, gravitation is seen as a manifestation of curvature of the geometry of spacetime. The geometry is encoded in the \emph{metric tensor}\ic{metric tensor} ($g_{\mu\nu}$). Also, there is another tensor named \emph{Riemann tensor}\ic{Riemann tensor} that carries all the information about curvature and is a function of the metric tensor\footnote{This is not strictly true, since it is indeed a function of a \emph{connection}\ic{connection} which is usually denoted as $\Gamma^\rho_{\;\;\mu \nu}$. In the present case, we will use a connection called the \emph{Christoffel connection}\ic{Christoffel symbol} which is the unique metric compatible connection in the case of zero torsion. It is symmetric under the exchange of $\mu \leftrightarrow \nu$.}. It is usually denoted as $R^\rho_{\;\;\mu \sigma \nu}$. We can contract various indices and obtain two other tensors that are used frequently. They are the \emph{Ricci tensor}\ic{Ricci tensor} ($R_{\mu \nu}$) and the \emph{scalar curvature}\ic{Scalar curvature} ($R$, sometimes referred to as the Ricci scalar\ic{Ricci scalar}). The definitions are:

\begin{align}
  R_{\mu \nu} \equiv R^\rho_{\;\; \mu \rho \nu} && R \equiv R^{\mu}_{\;\; \mu}.
\end{align}

In terms of these, quantitatively, the Einstein equation reads:

\begin{align}
  R_{\mu \nu} - \frac 1 2 g_{\mu \nu} R = 8\pi T_{\mu \nu},
\end{align}

where $g_{\mu \nu}$ is the metric tensor and $T_{\mu \nu}$ is the energy-momentum tensor, which depends only on the energy-matter content of the universe. The terms\footnote{The sum of these terms is called the \emph{Einstein tensor}\ic{Einstein tensor} ($G_{\mu \nu}$) which we mention just to note.} on the left hand side, on the other hand, are of purely geometric origin. Hence this equation relates energy-momentum to the geometry of the spacetime. Of course it can be read in reverse as well, in the direction that what kind of an energy-momentum tensor would yield the geometry at hand.

\subsection{Black Holes}\ic{black hole|textbf}
\label{sec:intro-gr-bh}

Black holes are a genuine prediction of general relativity. Classically, they are regions of spacetime from whose inside there cannot be any causal effect on the rest of the universe. However the idea of a ``black hole'' dates back to earlier times, and we would like to touch upon the history of the concept in a few paragraphs.

\ic{black hole!Mitchell-Laplace}John Michell (1783) and Pierre-Simon Laplace\footnote{An English translation of the text can be found in the appendix~A of \cite{hawking-ellis}.} (1799) independently thought of the idea of a \emph{dark star} whose gravitational pull on light particles emitted from its surface is so high that they never reach infinity, and hence are destined to turn back \cite{origin-of-bh-concept}. It seems that they both assume that light particles have some mass. The main idea behind the concept of dark star is that there may exist some stars on whose surface the escape velocity exceeds the speed of light.

As noted in \cite{origin-of-bh-concept-false-bh} ``the Newtonian dark body of Michell-Laplace is \emph{\textbf{not}} a black hole!'' (emphasis in the original). This is mainly because they think of an escape velocity, so when the body will be dark for observers at infinity, it will not be so for nearby observers. It is not really quite important whether what they describe is not a true black hole or not. What is important is that the concept of a dark star has become \emph{imaginable}.

Karl Schwarzschild\ip{Schwarzschild, Karl} found \cite{schwarzschild1916} (please see \cite{schwarzschild1916EngTrans} for an English translation) the first black hole solution of general relativity in 1916. This solution, known as the Schwarzschild black hole\ic{black hole!Schwarzschild}, describes a black hole that does not rotate and has zero electrical charge. In the coordinates that bears his name, it is described by the following metric\ic{metric!Schwarzschild}:

\begin{equation}
  \label{eq:1}
  ds^2 = -\left(1 - \frac{2M}{r} \right) dt^2 + \left(1 - \frac{2M}{r} \right)^{-1} dr^2 + r^2 d\Omega^2,
\end{equation}

where $d\Omega^2$ is the metric on unit 2-sphere. The singularity at $r=2M$ is fictitious and can be resolved after a suitable change of coordinates. However, there is a true singularity located at $r=0$.

\begin{figure}[t]
  \centering
  \includegraphics[width=10cm]{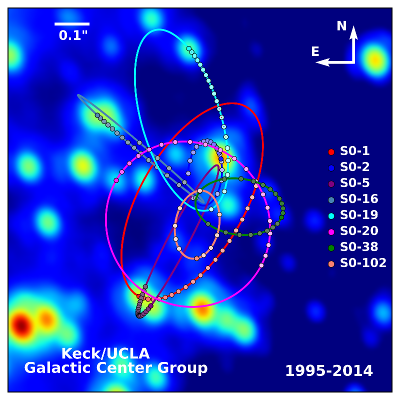}
  \caption{Orbits of various stars in the central arcsec of the Milkyway galaxy. In the region that is enclosed by every orbit lies a supermassive black hole of mass about four million solar masses \cite{bhAstroGhez}. (This image was created by Prof. Andrea Ghez and her research team at UCLA and are from data sets obtained with the W. M. Keck Telescopes.)}
  \label{fig:bh-astro-ghez}
\end{figure}

Kruskal coordinates\ic{Kruskal coordinates} on the other hand cover the entire manifold and is nonsingular on the event horizon. The metric reads\footnote{For more discussion one may see a standard book on the subject such as Sean Carroll's\ip{Carroll, Sean} \cite{carroll-gr-book}.}:

\begin{equation}
  \label{eq:21}
  ds^2 = \frac{32 M^3}{r} e^{-r/2M} (-dT^2 + dR^2) + r^2 d\Omega^2,
\end{equation}

where $T^2-R^2 = (1-r/2M) e^{r/2M}$. In the region covered by the Schwarzschild coordinates\ic{Schwarzschild coordinates} $T,R$ coordinates are defined as follow:

\begin{align}
  T &= \left( \frac{r}{2M} - 1\right)^{1/2} e^{r/4M} \sinh\left(\frac{t}{4M}\right),\\
  R &= \left( \frac{r}{2M} - 1\right)^{1/2} e^{r/4M} \cosh\left(\frac{t}{4M}\right).
\end{align}

The event horizon lies at $r=2M$ and its location in $T,R$ coordinates satisfy $T^2 = R^2$, or equivalently $T = \pm \abs R$. The singularity, on the other hand, is located at $r=0$. Its location is given by the relation $T^2-R^2 = 1$, or in other words by $T = \pm \sqrt{1+R^2}$. These are two hyperbolae. Figure~\ref{fig:kruskal-coor} includes representations of these features.

\begin{figure}
  \centering
  \begin{tikzpicture}[scale = 2, thick, domain=-2:2] 
  \draw[decorate, decoration=zigzag] plot (\x, { sqrt(1+\x*\x)});
  \draw[decorate, decoration=zigzag] plot (\x, {-sqrt(1+\x*\x)});
  \draw[dashed] (-2,-2) -- (2, 2);
  \draw[dashed] (-2, 2) -- (2,-2);
  \draw[->] (-2,0) -- (2,0) node [anchor=west] {$R$};
  \draw[->] (0,-2) -- (0,2) node [anchor=south] {$T$};
\end{tikzpicture}\vspace{1em}
  \caption[Description of the whole spacetime manifold in Kruskal coordinates. This is known as the \emph{Kruskal diagram}. Zigzag lines correspond to singularity, whereas dashed lines indicate the event horizon. There are indeed two singularities and event horizons. One of a white hole (at the bottom) and one of a black hole (at the top). The region on the right side extending to higher $R$ values is the spacetime outside the event horizon: the exterior Schwarzschild region. Its symmetric partner on the left side is the same but a causally disconnected region. They are connected through a worm hole. The worm hole is present in constant $T$ slices that does not intersect with any singularity. It is, however, \emph{nontraversable}: any observer who enters the worm hole cannot exit through the other mouth of the hole.]{Description of the whole spacetime manifold in Kruskal coordinates\ic{Kruskal coordinates}. This is known as the \emph{Kruskal diagram}\ic{Kruskal diagram}. Zigzag lines correspond to singularity, whereas dashed lines indicate the event horizon. There are indeed two singularities and event horizons. One of a white hole\ic{white hole} (at the bottom) and one of a black hole (at the top). The region on the right side extending to higher $R$ values is the spacetime outside the event horizon: the exterior Schwarzschild region\ic{black hole!Schwarzschild!exterior region}. Its symmetric partner on the left side is the same but a causally disconnected region. They are connected through a worm hole\ic{worm hole}. The worm hole is present in constant $T$ slices that does not intersect with any singularity. It is, however, \emph{nontraversable}: any observer who enters the worm hole cannot exit through the other mouth of the hole.}
  \label{fig:kruskal-coor}
\end{figure}
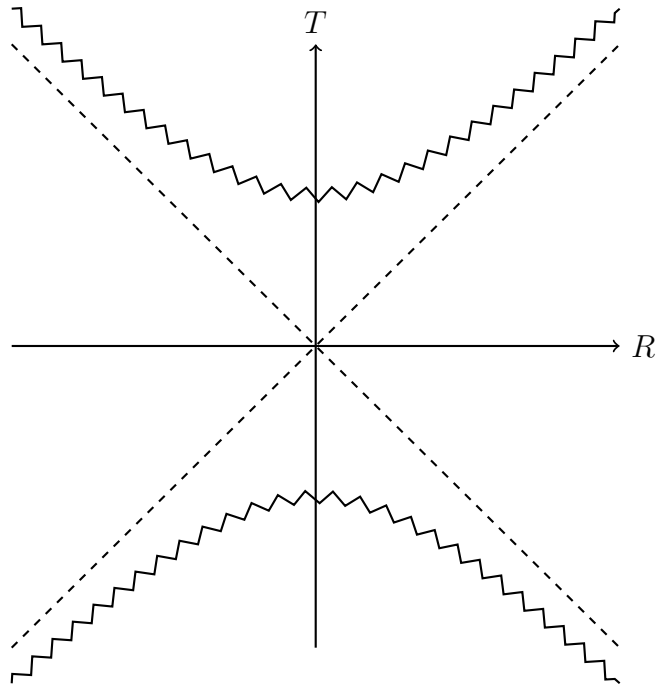

Another merit of this coordinate system, apart from covering the whole manifold, is that radial light rays follow paths that are \emph{straight lines} tilted $\pm 45^\circ$ from the $T$ axis; as in the Minkowski spacetime\ic{Minkowski spacetime}.

 The fact that radial light rays follow straight lines in Kruskal coordinates is quite useful in understanding the causal structure\ic{causal structure} of the spacetime, \ie understanding which events can effect which other events. However, the representation of the manifold such as the one in Figure~\ref{fig:kruskal-coor} extends indefinitely. To remedy the problem, \emph{Penrose diagrams}\ic{Penrose diagram} are quite useful. By a conformal transformation, the whole spacetime manifold is represented in a finite amount of paper space. In a \co{conformal transformation}, the light cone structure --hence the \co{causal structure} of spacetime-- is preserved. \co{Penrose diagram} for the Schwarzschild spacetime is drawn in Figure~\ref{fig:penrose-schwarzschild}.

\begin{figure}
  \centering
  \begin{tikzpicture}[scale = 2, thick] 
  \draw (0,0) -- (1,1) node [midway,above] {$\mathcal{I}^+$};
  \draw[decorate, decoration=zigzag] (1,1) -- (3,1) node[midway,above] {Singularity};
  \draw (3,1) -- (4,0) node[midway,above,anchor=south west] {$\mathcal{I}^+$};
  \draw (4,0) node[anchor=west] {$i^0$};
  \draw (0,0) node[anchor=east] {$i^0$};
  \draw (3,1) node[above] {$i^+$};
  \draw (1,1) node[above] {$i^+$};
  \draw (4,0) -- (3,-1) node[below] {$i^-$} node[midway,below,anchor=north west] {$\mathcal{I}^-$};
  \draw[decorate, decoration=zigzag] (3,-1) -- (1,-1) node[midway,below] {Singularity};
  \draw (1,-1) node[below] {$i^-$} -- (0,0) node [midway,below] {$\mathcal{I}^-$};
  \draw[dashed,thin] (1,1) -- (3,-1);
  \draw[dashed,thin] (1,-1) -- (3,1);
\end{tikzpicture}\vspace{1em}
  \caption{Penrose diagram for the extended Schwarzschild solution.}
  \label{fig:penrose-schwarzschild}
\end{figure}
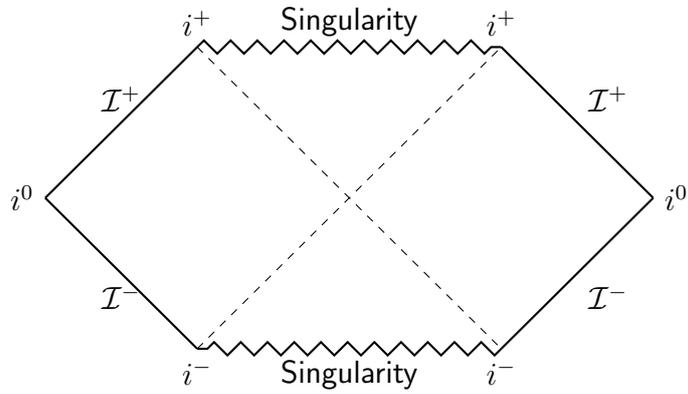

 There are three types of infinities seen in the figure: timelike, spacelike and lightlike. It would be useful to list them in a table:

 \begin{table}[h]
   \caption[Types of infinities present in a Penrose diagram.]{Types of infinities present in a \co{Penrose diagram}.}
   \vspace{1em}
   \centering
   \begin{tabular}{cll}\toprule
     Symbol    & Name               & Property                          \\ \midrule
     $i^0$     & Spacelike infinity & spacelike geodesics begin and end \\ \midrule
     $i^-$     & Timelike past      & timelike geodesics begin          \\ \midrule
     $i^+$     & Timelike future    & timelike geodesics end            \\ \midrule
     $\mc I^-$ & Lightlike past     & lightlike geodesics begin         \\ \midrule
     $\mc I^+$ & Lightlike future   & lightlike geodesics end           \\ \bottomrule
   \end{tabular}
   \label{tab:inf}
 \end{table}

The use of the word ``geodesic'' in Table~\ref{tab:inf} is important. For example, there are timelike paths that are asymptotically null and do not end in $i^+$. On the other hand, there are spacelike paths that may not end or begin at $i^0$. Figure~\ref{fig:pen-sch-asymp-null} illustrates a few examples.

\begin{figure}
  \centering
  \begin{tikzpicture}[scale = 2, thick] 
  \draw (0,0) -- (1,1) node [midway,above] {$\mathcal{I}^+$};
  \draw[decorate, decoration=zigzag] (1,1) -- (3,1) node[midway,above] {Singularity};
  \draw (3,1) -- (4,0) node[midway,above,anchor=south west] {$\mathcal{I}^+$};
  \draw (4,0) node[anchor=west] {$i^0$};
  \draw (0,0) node[anchor=east] {$i^0$};
  \draw (3,1) node[above] {$i^+$};
  \draw (1,1) node[above] {$i^+$};
  \draw (4,0) -- (3,-1) node[below] {$i^-$} node[midway,below,anchor=north west] {$\mathcal{I}^-$};
  \draw[decorate, decoration=zigzag] (3,-1) -- (1,-1) node[midway,below] {Singularity};
  \draw (1,-1) node[below] {$i^-$} -- (0,0) node [midway,below] {$\mathcal{I}^-$};
  \draw[dashed,thin] (1,1) -- (3,-1);
  \draw[dashed,thin] (1,-1) -- (3,1);
  \draw (3,-1) .. controls (3,0) .. (3.5,0.5);
  \draw (0.6,0.6) .. controls (1.2,0) and (2.8,0) .. (3.4,0.6);
\end{tikzpicture}\vspace{1em}
  \caption{The timelike path that begins at $i^-$, however does not end at $i^+$. It is asymptotically null and reaches $\mc I^+$ instead. The spacelike path begins at $\mc I^+$ of one side and ends at $\mc I^+$ of another region.}
  \label{fig:pen-sch-asymp-null}
\end{figure}
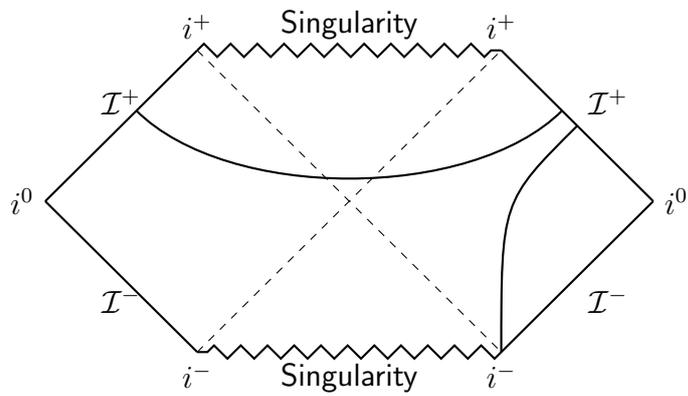

The geometry of spacetime, hence its causal structure, is determined by general relativity, which we discussed so far. The geometric quantity $R_{\mu\nu} - \frac{1}{2} R g_{\mu\nu}$ is on one side of the Einstein equation. The other side consists of the energy momentum tensor. Matter, on the other hand, is governed by quantum laws. In lack of a proper theory of quantum gravity, one way to approximately reconcile the geometry of spacetime with the quantum nature of matter is to generalize the quantum field theory to curved spacetimes. This is our next subject.

\section{Quantum Field Theory in Curved Spacetimes}
\label{sec:qft-in-curved-st}

In flat spacetime, predictions of quantum field theory (QFT) are in excellent agreement with observations. Building on this experience, one may want to generalize QFT to arbitrary, curved, spacetimes.

One way of achieving this goal, is to replace the flat metric ($\eta_{\mu \nu}$) with an appropriate metric ($g_{\mu \nu}$) and partial derivatives ($\partial_\mu$) with covariant derivatives ($\nabla_\mu$).

For example, in Minkowski spacetime, the Lagrangian for a massive real scalar field is as follows (in Cartesian coordinates):

\begin{equation}
  \label{eq:3}
  \mc L = \frac 1 2 \left( \partial^\mu \phi \partial_\mu \phi + m^2 \phi^2 \right).
\end{equation}

Please note that because we use the mostly positive metric signature, we have $m^2$ instead of $-m^2$ which is usually used in books on QFT because they adopt the mostly negative metric signature convention.

In order to carry this Lagrangian into curved spacetimes, we map the metric and derivative operations accordingly. Moreover we may add \cite{parker-toms} a coupling with the scalar curvature of the form $\xi R \phi^2$:

\begin{equation}
  \label{eq:4}
  \mc L = \frac 1 2 \left( \nabla^\mu \phi \nabla_\mu \phi + m^2 \phi^2 + \xi R \phi^2 \right).
\end{equation}

In $n$ spacetime dimensions, $\xi=0$ and $\xi = (n-2)/4(n-1)$ correspond to minimal and conformal couplings \cite{parker-toms} respectively. The Euler-Lagrange equation for $\phi$ is easy to derive, it reads:

\begin{equation}
  \label{eq:5}
  \left( \nabla^2 - m^2 - \xi R \right) \phi = 0.
\end{equation}

In this section, for the sake of simplicity, we will be interested in the minimally coupled massless real scalar fields. Hence the Lagrangian and Euler-Lagrange equation we are interested in are:

\begin{align}
  \label{eq:6}
  \mc L = \frac 1 2 \nabla^\mu \phi \nabla_\mu \phi , && \nabla^2 \phi = 0.
\end{align}

In Minkowski spacetime\ic{Minkowski spacetime}, one can define a QFT vacuum that all the inertial observers agree upon. On the other hand, the covariant formalism of QFT in curved spacetimes can be applied to Minkowski spacetime described by non-Cartesian coordinates. One may use a set of coordinates that are appropriate to, for example, accelerated observers. It is then seen that the QFT vacuum that inertial observers agree upon is not the appropriate vacuum state and accelerated observers detect particles. In the next section we discuss quantitative basis of similar phenomena.

\subsection{Particle Creation}
\label{sec:particle-creation}

Let us begin with an example. For example, in the case of gravitational collapse\ic{gravitational collapse} we may suppose that the initial matter density is so low that the spacetime is almost flat. After the implosion, the spacetime is that of Schwarzschild. Both of these are static spacetimes in themselves, whereas the whole process of gravitational collapse does not describe a static solution. We consider the portion of spacetime before the creation of a black hole as the ``in-region'' and the resulting portion of spacetime after the presence of the black hole as the ``out-region.''

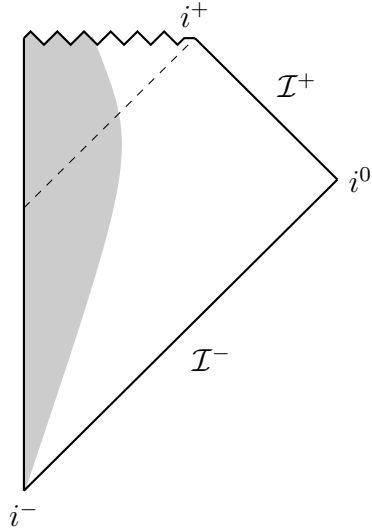
\begin{figure}
  \centering
  \begin{tikzpicture}[scale = 3, thick]
  \draw [decorate, decoration={lineto, pre=zigzag,pre length=0.9cm}, fill=black!20, color=black!20] (0,2) -- (0.3,2) -- (0,0);
  \draw [fill=black!20, color=black!20] (0.3,2) .. controls (0.5,1.5) .. (0,0);
  \draw (0,0) -- (0,2);
  \draw[decorate, decoration=zigzag] (0,2) -- (0.75,2) node [anchor=south] {$i^+$};
  \draw (0.75,2) -- (1.375,1.375) node [midway, anchor=south west] {$\mathcal{I}^+$} node [anchor=west] {$i^0$};
  \draw (1.375,1.375) -- (0,0) node [midway, anchor=north west] {$\mathcal{I}^-$} node [anchor=north] {$i^-$};
  \draw[thin, dashed] (0,1.25) -- (0.75,2);
\end{tikzpicture}\vspace{1em}
  \caption[Penrose diagram describing a gravitational collapse. Event horizon is indicated with a dashed line, whereas the collapsing matter is found in the grey region.]{\co{Penrose diagram} describing a \co{gravitational collapse}. Event horizon is indicated with a dashed line, whereas the collapsing matter is found in the grey region.}
  \label{fig:penrose-collapse}
\end{figure}

Of course the reason behind these names is the same as in QFT: we think of an initial stationary region and some interactions that occur afterwards. Later on, we obtain a final stationary region of spacetime. Not all spacetimes are of this form, however we will be interested in this type of spacetimes because of the presence of particle interpretation of the theory in stationary regions. There will be one exception, when discussing the Hawking radiation, the final region --the Cauchy surface-- we choose is the union of event horizon and lightlike future. There is a timelike Killing vector in $\mc I^+$; however no such vector exists on the event horizon. Therefore the particle interpretation is ambiguous on the event horizon. However this state of affairs will not be an obstacle, because we are mainly interested in the radiation emitted by the black hole that will reach the lightlike future.

If $\xi$ is a timelike Killing vector, meaning that $\xi \cdot \xi < 0$ and $\nabla_{\mu}\xi_\nu + \nabla_\nu \xi_\mu = 0$ hold, a solution of the field equation that satisfies $\mc L_{\xi} f = -i\omega f$, where $\omega > 0$ and $\mc L_{\xi}$ indicates the Lie derivative\ic{Lie derivative} with respect to $\xi$, is called a mode of positive frequency\ic{positive frequency!definition}. If we use $\xi$ as a timelike coordinate, then this condition is equivalent to $\partial_\xi f = -i\omega f$. For more on the quality of being a positive frequency solution, readers are referred to \cite{Fabbri2005, waldQFT-BH}.

Suppose we have two sets of positive frequency solutions: $\{\fout_k\}_k$ in the in-region and $\{\fin_k\}_k$ in the out-region\footnote{We find the notation adopted in \cite{Fabbri2005} convenient.}. Moreover we require these solutions to be normalized according to the $U(1)$-inner product\ic{U(1)@$U(1)$-inner product} to be described shortly as follows: $(f_k,f_l) = \delta(k-l)$, $(f_k,f_l^*) = 0$. We deal with the case where $k$ is a continuous variable. When it takes discrete values, the situation is similar: Dirac delta distributions are mapped to Kronecker delta symbols, integrals below are replaced with sums and so on.

We expand the field solution $\phi$ in each region, in and out, respectively.

\begin{align}
  \phi &= \int dk \left( \fout_k \aout_k + \fouts_k \aoutd_k \right), \label{eq:phi-out}\\
       &= \int dk \left( \fin_k \ain_k + \fins_k \aind_k \right). \label{eq:phi-in}
\end{align}

The recipe of \co{second quantization} is that the coefficients $\ain,\aout$ of $\fin,\fout$ are to be regarded as operators.

Because the sets $\{\fout_k\}_k \cup \{\fouts_k\}_k$ and $\{\fin_k\}_k \cup \{\fins_k\}_k$ are complete we can express an element of each set in terms a superposition of functions in the other set. The completeness property allows one to express $\aout$ in terms of $\ain$ and $\aind$, similarly $\ain$ in terms of $\aout$ and $\aoutd$.

The definition of the $U(1)$-inner product\ic{U(1)@$U(1)$-inner product} is quite important in this regard. It is an inner product between two field solutions and is defined as follows\footnote{For a detailed derivation, the reader may consult to Appendix~\ref{cha:derivation-u1-inner}.}:

\begin{equation}
  \label{eq:13}
  (\phi,\psi) = i \int_\Sigma d^nx \abs g^{1/2} n_\mu (\psi \nabla^\mu \phi^* - \phi^* \nabla^\mu \psi).
\end{equation}

For example, by looking at the expansion of the field $\phi$ in terms of out-modes \eqref{eq:phi-out} we see that $\aout_k = (\fout_k,\phi)$ is satisfied. However, we can also use the in-mode expansion \eqref{eq:phi-in} of the field. This will yield:

\begin{equation}
  \label{eq:22}
  \aout_k  = \int dl \left[ (\fout_k,\fin_l) \ain_l + (\fout_k,\fins_l) \aind_l \right].
\end{equation}

If we know the results of inner products between in and out mode functions, we would know the expression of $\aout$ in terms of $\ain$ and $\aind$. This information would give us, for example, the number of particles seen by an out-observer that the in-vacuum had. The results of the following inner products are called \emph{Bogoliubov coefficients}\ic{Bogoliubov coefficients}, $\alpha_{kl}$ and $\beta_{kl}$ defined as follows:

\begin{align}
  \alpha_{kl} = (\fin_l,\fout_k), && \beta_{kl} = -(\fins_l,\fout_k).
\end{align}

For instance, we would like to expand $\fout$ interms of $\fin$ and $\fins$. This is linear algebra. The result would be:

\begin{align}
  \label{eq:23}
  \fout_k &= \int dl \left[ (\fin_l,\fout_k) \fin_l - (\fins_l,\fout_k) \fins_l \right],\\
    &= \int dl (\alpha_{kl} \fin_l + \beta_{kl} \fins_l).
\end{align}

By taking the inner product of both sides with $\fin_{l'}$ or $\fins_{l'}$ one may verify the expansion. By expressing the inner products $(\fout_k,\fout_l) = \delta(k-l)$ and $(\fout_k,\fouts_l) = 0$ in the in-basis, one may obtain the following identities:

\begin{align}
  \int dq (\alpha^*_{kq} \alpha_{lq} - \beta^*_{kq}\beta_{lq}) &= \delta(k-l), \label{eq:bog-delta}\\
  \int dq (\alpha_{kq}\beta_{lq} - \beta_{kq}\alpha_{lq}) &= 0.\label{eq:bog-0}
\end{align}

Let us return to our aim of expressing $\aout$ in terms of in-operators. Using the above-defined Bogoliubov coefficients, we can rewrite \eqref{eq:22} in a more concise notation:

\begin{equation}
  \label{eq:24}
  \aout_k = \int dl ( \alpha_{kl}^* \ain_l - \beta_{kl}^* \aind_l ).
\end{equation}

On the other hand, with two sets of annihilation/creation operators, we define two vacua, \emph{in-vacuum}\ic{vacuum!in!definition} and \emph{out-vacuum}\ic{vacuum!out!definition} as follows:

\begin{align}
  \forall k, \quad \ain_k \invac = 0; && \forall k, \quad \aout_k \outvac = 0.
\end{align}

We are now in a position to ask how many particles there are in any given out-mode in the \emph{in-vacuum}\ic{vacuum!in}. The answer depends on the Bogoliubov coefficient $\beta_{kl}$. Expectation value of an out-mode number operator in the in-vacuum needs to be calculated. By using expression (\ref{eq:24}) for $\aout_k$, one gets $\aout_k \invac = - \int dl\; \beta_{kl}^* \aind_l \invac$. Hence one finds:

\begin{equation}
  \bra{\text{in}} \aoutd_k \aout_k \invac = \int dl\; \abs{\beta_{kl}}^2,
\end{equation}

where the commutation relation $[\ain_k,\aind_l] = \delta(k-l)$ has been used. If $\forall k,l; \beta_{kl} = 0$, by (\ref{eq:bog-delta}), we see that $\alpha_{kl}$ is a unitary transformation: positive frequency solutions $\{\fout_k\}_k$ and $\{\fin_k\}_k$ are related by a unitary transformation. On the contrary, if this is the case, naturally $\forall k,l; \beta_{kl} = 0$.

We shall give two examples where particle creation occurs. The first one is called the Unruh radiation \cite{UnruhRadiationNotesOnBH}. It concerns uniformly accelerated observers on flat background geometry. The second example is the celebrated Hawking radiation \cite{hawking1975}. It is about the radiation emitted by a black hole.

\subsection{Unruh Radiation}\ic{Unruh Radiation|textbf}
\label{sec:unruh-radiation}

Here we suppose the spacetime is $1+1$ dimensional, because it makes the illustration of the concept much more convenient. The field equation in the massless case is $\nabla^2 \phi = 0$. If we write this explicitly we obtain $g^{\mu\nu} \partial_\mu \partial_\nu \phi - g^{\mu\nu} \Gamma^\lambda_{\mu\nu} \partial_\lambda \phi = 0$ where $\Gamma^\lambda_{\mu\nu}$ is the Christoffel symbol\ic{Christoffel symbol}\ic{Christoffel symbol}. We solve the field equation first in Cartesian coordinates and then solve it in, what is called, the \emph{Rindler coordinates}\ic{Rindler coordinates}.

In Cartesian coordinate system, all the Christoffel symbols $\Gamma_{\mu\nu}^\lambda$ vanish. Moreover since the metric tensor is $g_{\mu\nu} = \diag(-1,+1)$, the inverse metric turns out to be numerically equal to the metric: $g^{\mu\nu} = \diag(-1,+1)$. All in all, the field equation is found as follows:

\begin{equation}
  \label{eq:12}
  (-\partial_t^2 + \partial_x^2) \phi = 0.
\end{equation}

Positive frequency\ic{positive frequency} solutions are $\exp(-i\omega t + ikx)$ where $\omega = \abs k$ and $k \in \reals^*$. We regard these solutions as $\fin_k$, after normalizing them. Normalized $\fin_k$ solutions are as follows:

\begin{equation}
  \label{eq:25}
  \fin_k = (4\pi \omega)^{-1/2} e^{-i\omega t + ikx}.
\end{equation}

The portion of spacetime covered by the Rindler coordinates\ic{Rindler coordinates} lies in the whole Minkowski spacetime that is covered by the usual Cartesian coordinates\ic{Cartesian coordinates}, therefore the use of in/out terminology is not as it is in the sense of describing two different regions of spacetime that are stationary. It is, however, in terms of a \emph{re}-interpretation of the quantum field at hand by different observers. The Rindler coordinates\ic{Rindler coordinates!definition}\footnote{The coordinates for which we use the name \emph{Rindler}, $\eta,\xi$, are related to the original coordinates introduced \cite{Rindler1966} by Wolfgang Rindler\ip{Rindler, Wolfgang}, $T,X$ through $T = a\eta$ and $X = \exp(a\xi)/a$.} are defined as follows:

\begin{align}
  t = \frac 1 a e^{a\xi} \sinh(a\eta), && x = \frac 1 a e^{a\xi} \cosh(a\eta),
\end{align}

where both $\eta$ and $\xi$ range from $-\infty$ to $\infty$ and $a$ is a positive parameter. These coordinates, however, only cover a \emph{quadrant} of the Minkowski spacetime. We could equally write the above definition with $t,x$ replaced by $-t,-x$. This would correspond to another quadrant of the spacetime. If we denote the former coordinates as $t_R,x_R$ and the latter ones as $t_L,x_L$, the proper definition that handles both cases will be as follows:

\begin{align}
  t_L &= -\frac{1}{a} e^{a\xi} \sinh(a\eta) & x_L &= -\frac{1}{a} e^{a\xi} \cosh(a\eta),\\
  t_R &= \frac 1a e^{a\xi} \sinh(a\eta) & x_R &= \frac 1a e^{a\xi} \cosh(a\eta).
\end{align}

Curves that are parameterized by $\eta$, on which $\xi$ is constant, correspond to worldlines, which are hyperbolae, that describe \emph{uniformly accelerated}\ic{uniform acceleration} motion: $a_\mu a^\mu$ is constant on each of these worldlines where $a^\mu$ is the acceleration four-vector. Requiring the magnitude of acceleration vector to be constant is the correct restriction to describe uniform acceleration\ic{uniform acceleration!definition}, it reproduces the correct trajectory in the Newtonian limit.

On the other hand, constant $\eta$ surfaces on which $\xi$ varies, are lines that are described by $t/x = \tanh(a\eta)$. Figure~\ref{fig:rindler} illustrates constant $\eta$ or $\xi$ surfaces.

\begin{figure}
  \centering
  \begin{tikzpicture}[scale = 2] 
  \draw[dashed] (-2,-2) -- (2, 2);
  \draw[dashed] (-2, 2) -- (2,-2);
  \draw[->] (-2,0) -- (2,0) node [anchor=west] {$x$};
  \draw[->] (0,-2) -- (0,2) node [anchor=south] {$t$};
  \foreach \a in {-0.8,-0.6,...,0.8}
      \draw[thin, domain=-2:2] plot (\x, {\a*\x});
  
  \draw[thin, domain=-2.99322284613:2.99322284613] plot ({0.2*cosh(\x)}, {0.2*sinh(\x)});
  \draw[thin, domain=-2.29243166956:2.29243166956] plot ({0.4*cosh(\x)}, {0.4*sinh(\x)});
  \draw[thin, domain=-1.87382024253:1.87382024253] plot ({0.6*cosh(\x)}, {0.6*sinh(\x)});
  \draw[thin, domain=-1.56679923697:1.56679923697] plot ({0.8*cosh(\x)}, {0.8*sinh(\x)});
  \draw[thin, domain=-1.31695789692:1.31695789692] plot ({1.0*cosh(\x)}, {1.0*sinh(\x)});
  \draw[thin, domain=-1.09861228867:1.09861228867] plot ({1.2*cosh(\x)}, {1.2*sinh(\x)});
  \draw[thin, domain=-0.89558809952:0.89558809952] plot ({1.4*cosh(\x)}, {1.4*sinh(\x)});
  \draw[thin, domain=-0.69314718056:0.69314718056] plot ({1.6*cosh(\x)}, {1.6*sinh(\x)});
  \draw[thin, domain=-0.46714530810:0.46714530810] plot ({1.8*cosh(\x)}, {1.8*sinh(\x)});

  \draw[thin, domain=-2.99322284613:2.99322284613] plot ({-0.2*cosh(\x)}, {-0.2*sinh(\x)});
  \draw[thin, domain=-2.29243166956:2.29243166956] plot ({-0.4*cosh(\x)}, {-0.4*sinh(\x)});
  \draw[thin, domain=-1.87382024253:1.87382024253] plot ({-0.6*cosh(\x)}, {-0.6*sinh(\x)});
  \draw[thin, domain=-1.56679923697:1.56679923697] plot ({-0.8*cosh(\x)}, {-0.8*sinh(\x)});
  \draw[thin, domain=-1.31695789692:1.31695789692] plot ({-1.0*cosh(\x)}, {-1.0*sinh(\x)});
  \draw[thin, domain=-1.09861228867:1.09861228867] plot ({-1.2*cosh(\x)}, {-1.2*sinh(\x)});
  \draw[thin, domain=-0.89558809952:0.89558809952] plot ({-1.4*cosh(\x)}, {-1.4*sinh(\x)});
  \draw[thin, domain=-0.69314718056:0.69314718056] plot ({-1.6*cosh(\x)}, {-1.6*sinh(\x)});
  \draw[thin, domain=-0.46714530810:0.46714530810] plot ({-1.8*cosh(\x)}, {-1.8*sinh(\x)});
\end{tikzpicture}\vspace{1em}
  \caption{An illustration of Rindler wedges. Constant $\eta$ surfaces are straight lines passing through the origin, whereas constant $\xi$ surfaces are hyperbolae that correspond to worldlines that describe uniformly accelerated motion.}
  \label{fig:rindler}
\end{figure}
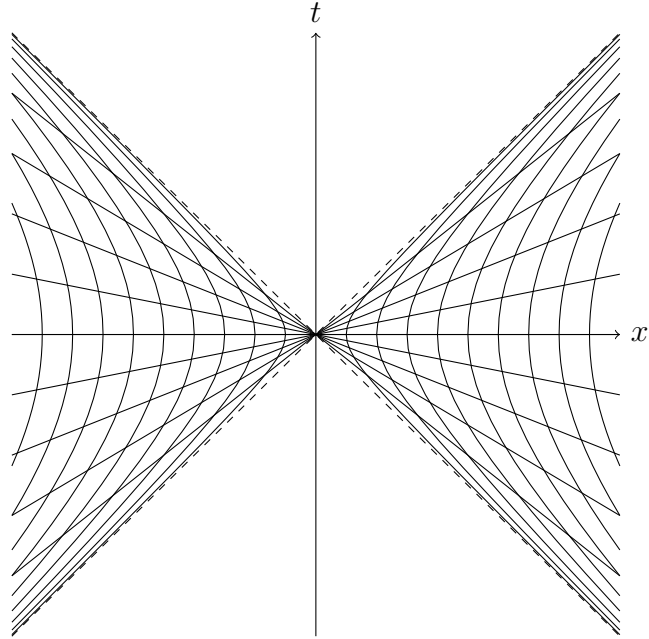

In Rindler coordinates\ic{Rindler coordinates}, the metric is found to be:

\begin{equation}
  \label{eq:26}
  ds^2 = e^{2a\xi} (-d\eta^2 + d\xi^2).
\end{equation}

The non-vanishing Christoffel symbols\ic{Christoffel symbol!in Rindler metric} are given as:

\begin{equation}
  \label{eq:27}
  \Gamma_{\xi\xi}^\xi = \Gamma_{\eta\eta}^\xi = \Gamma_{\eta\xi}^\eta = a.
\end{equation}

These are needed in the expansion of the field equation:

\begin{equation}
  \nabla^2 \phi = g^{\mu\nu} \partial_\mu \partial_\nu \phi - g^{\mu\nu} \Gamma^\lambda_{\mu\nu} \partial_\lambda \phi = 0.  
\end{equation}

In this case the part that contains the Christoffel symbol\ic{Christoffel symbol} vanishes. The equation becomes $e^{-2a\xi}(-\partial_\eta^2+\partial_\xi^2)\phi=0$. Because the exponential factor in front never vanishes, we get:

\begin{equation}
  \label{eq:28}
  (-\partial_\eta^2 + \partial_\xi^2)\phi = 0.
\end{equation}

The positive frequency solutions are easy to find. However, there is a subtlety. In Minkowski spacetime the timelike Killing vector we choose to define positive frequency is $\partial_t$. In the right Rindler wedge, $\partial_\eta$ points in the same direction. So, solutions $e^{-i\omega \eta + ik\xi}$ are positive frequency in this wedge. However, in the left Rindler wedge, $\partial_t$ and $\partial_\eta$ points in \emph{opposite} directions. Therefore, we choose $-\partial_\eta$ as timelike Killing vector in the left Rindler wedge. For that reason, positive frequency solutions in this wedge are $e^{i\omega \eta + ik\xi}$.

The $U(1)$-inner product\ic{U(1)@$U(1)$-inner product!in Rindler coordinates} in left or right Rindler wedges are given as follows:

\begin{equation}
  \label{eq:29}
  (\phi,\psi)_{R,L} = \mp i\int_{\text{const. }\eta} d\xi (\psi \partial_\eta \phi^* - \phi^* \partial_\eta \psi),
\end{equation}

where `$-$' sign is for the right Rindler wedge, whereas the `$+$' is for the left Rindler wedge. When normalized, positive frequency solutions in left or right Rindler wedges are found to be:

\begin{align}
  f^L_k &= (4\pi\omega)^{-1/2} e^{i\omega\eta + ik\xi},\\
  f^R_k &= (4\pi\omega)^{-1/2} e^{-i\omega\eta + ik\xi}.
\end{align}

Our aim is to calculate the number of particles in each mode that an accelerated observer will see in the Minkowski vacuum. We could directly start calculating the Bogoliubov coefficients, however there is a more elegant way due to William Unruh\ip{Unruh, William G.} \cite{UnruhRadiationNotesOnBH}. The approach is to construct solutions out of solutions in each Rindler wedge, that is analytic in the Minkowski spacetime. Reference \cite{BirrellDavies} gives the following two positive frequency solutions, which we normalize and define as $f_k^{(1)}, f_k^{(2)}$:

\begin{align}
  f_k^{(1)} &= \frac{e^{\pi \omega / 2a} f_k^R + e^{-\pi\omega / 2a} f_{-k}^{L,*}}{[2 \sinh(\pi\omega/a)]^{1/2}}, & f_k^{(2)} &= \frac{e^{-\pi\omega/2a} f_{-k}^{R,*} + e^{\pi\omega/2a} f_k^L}{[2 \sinh(\pi\omega/a)]^{1/2}}.\label{eq:40}
\end{align}

We can expand the field in these modes as follows:

\begin{equation}
  \phi = \int dk \; (f_k^{(1)} a_k^{(1)} + f_k^{(1),*} a_k^{(1),\dagger} + f_k^{(2)} a_k^{(2)} + f_k^{(2),*} a_k^{(2),\dagger}),\label{eq:38}
\end{equation}

and operators $a_k^{(1)},a_k^{(2)}$ annihilate Minkowski vacuum state, $\ket M$.

On the other hand, the field can be expanded in mode functions in left and right Rindler wedges:

\begin{equation}\label{eq:39}
  \phi = \int dk \; (f_k^L a_k^{L} + f_k^{L,*} a_k^{L,\dagger} + f_k^R a_k^{R} + f_k^{R,*} a_k^{R,\dagger}).
\end{equation}

The merit and elegance of this approach is that, by expanding and reordering the terms in the integrand appearing in (\ref{eq:38}), one can easily find expression for $a_k^L,a_k^R$ in terms of $a_k^{(1)},a_k^{(2)},a_k^{(1),\dagger},a_k^{(2),\dagger}$. After all, this is what we are after.

Using the definitions (\ref{eq:40}) in (\ref{eq:38}) and comparing the result with equation (\ref{eq:39}) we find:

\begin{align}
  a_k^L &= \frac{e^{-\pi\omega/2a} a_{-k}^{(1),\dagger} + e^{\pi\omega/2a} a_k^{(2)}}{[2 \sinh(\pi\omega/a)]^{1/2}}, & a_k^R &= \frac{e^{\pi\omega/2a} a_{k}^{(1)} + e^{-\pi\omega/2a} a_{-k}^{(2),\dagger}}{[2 \sinh(\pi\omega/a)]^{1/2}}.\label{eq:41}
\end{align}

In order to determine the average occupation number of the mode $f_k^L$ that an accelerated observer sees in the Minkowski vacuum, we need to calculate the expectation value of the corresponding number operator: $\bra M a_k^{L,\dagger} a_k^L \ket M$. Using (\ref{eq:41}) we find $a_k^L \ket M = e^{-\pi\omega/2a} a_{-k}^{(1),\dagger} \ket M / [2 \sinh(\pi\omega/a)]^{1/2}$. Hence the expectation value of number operator becomes:

\begin{equation}
  \bra M a_k^{L,\dagger} a_k^L \ket M = \frac{\delta(0)}{e^{2\pi\omega/a}-1}.
\end{equation}

The result of $\bra M a_k^{R,\dagger} a_k^R \ket M$ is the same. The appearance of the delta function follows from the continuum normalization of mode functions. If we constructed wave packets out of $f_k^L, f_k^R$, we would have obtained $\delta_{-k,-k}$, instead of $\delta(k-k)$, which equals one. This is in the end a Planckian distribution\ic{Planck distribution} with temperature\footnote{Note that if we used units for which $\hbar,G$ and $c$ are not unity, we would get $T = \hbar a / 2\pi c$.} $T=a/2\pi$. It is called the \emph{Unruh temperature}\ic{Unruh temperature}. We have seen that Minkowski vacuum is a thermal state, and is not empty.

Although we will allocate more space to the construction of wave packets in the next section while discussing the Hawking radiation, let us briefly mention the method and see that we will get a nondivergent expectation value for the number operator.

We already know the relation between $f^L,f^R$ and $f^{(1)},f^{(2)}$. This information will be quite useful. By superposing modes of similar wave vector, we create wave packets as follows\footnote{As Hawking did in \cite{hawking1975}.}:

\begin{equation}
  g^{L,R}_{jn} \equiv \ve^{-1/2} \int_{j\ve}^{(j+1)\ve} dk \; e^{-i2\pi n k / \ve} f_k^{L,R};\quad j,n \in \mathbb Z,
\end{equation}

for some $\ve > 0$. Let us call the associated annihilation operators with these wave packets as $b_{jn}^{L,R}$. We can expand the field as follows:

\begin{equation}
  \phi = \sum_{jn} \left( g_{jn}^L b_{jn}^L + g_{jn}^R b_{jn}^R + g_{jn}^{L,*} b_{jn}^{L,\dagger} + g_{jn}^{R,*} b_{jn}^{R,\dagger} \right).
\end{equation}

On the other hand, these wave packets satisfy the discrete versions of normalization conditions, \ie $(g^{L,R}_{jn},g^{L,R}_{j'n'}) = \delta_{jj'}\delta_{nn'}$ and $(g^{L,R}_{jn},g^{L,R;*}_{j'n'}) = 0$ with any inner product between right and left modes vanishing. This property allows us to write:

\begin{equation}
  b_{jn}^{L,R} = (g^{L,R}_{jn}, \phi).
\end{equation}

Expressing the field as an integral as in (\ref{eq:38}) allows us to calculate the expectation value of the number operator $b_{jn}^{L,R;\dagger}b_{jn}^{L,R}$ in the in-vaccum easily. We provide the final result:

\begin{equation}
  \bra{\text{in}} b_{jn}^{L,R;\dagger}b_{jn}^{L,R} \invac = \ve^{-1} \int_{j\ve}^{(j+1)\ve} dk \; \frac{1}{e^{2\pi\omega/a} - 1}.
\end{equation}

We could take the integral exactly, but we are interested in the regime when $\ve \ll 1$ is satisfied. In any case, this integral equals $\ve$ times the integrand evaluated at some point in the interval $(j\ve,j\ve + \ve)$. Because $\ve \ll 1$ is satisfied, in the zeroth order we can choose this point as $j\ve$ and hence associate a frequency $\omega$ to the wave packet as $\omega = j\ve$. Hence we find, in the limit $\ve \to 0^+$:

\begin{equation}
  \bra{\text{in}} b_{jn}^{L,R;\dagger}b_{jn}^{L,R} \invac = \frac{1}{e^{2\pi\omega/a} - 1}.
\end{equation}

Here we have a result that is finite and has no delta function singularity.

\subsection{Hawking Radiation}\ic{Hawking radiation|textbf}
\label{sec:hawking-radiation}

The discussion of Hawking radiation in the book by Alessandro Fabbri\ip{Fabbri, Alessandro} and José Navarro-Salas\ip{Navarro-Salas, José} \cite{Fabbri2005} is quite good and lucid, we will mainly follow it in regard to its approach and notation.

We begin by considering the spacetime described by the Vaidya metric\ic{metric!Vaidya}:

\begin{equation}
  \label{eq:2}
  ds^2 = -\left( 1 - \frac{2 M(v)}{r}\right) dv^2 + 2 dv dr + r^2 d\Omega^2.
\end{equation}

When $M(v)$ is constant we see that this is the Schwarzschild metric in Eddington-Finkelstein coordinates. Consequently then, if $M=0$ is satisfied we recover the Minkowski metric when we make a change of coordinate from $v$ to $t$ through $v=t+r$.

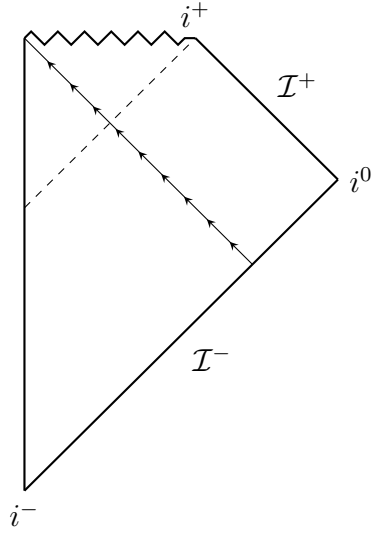
\begin{figure}
  \centering
  \begin{tikzpicture}[scale = 3, thick]
  \draw (0,0) -- (0,2);
  \draw[decorate, decoration=zigzag] (0,2) -- (0.75,2) node [anchor=south] {$i^+$};
  \draw (0.75,2) -- (1.375,1.375)  node [midway, anchor=south west] {$\mathcal{I}^+$} node [anchor=west] {$i^0$};
  \draw (1.375,1.375) -- (0,0) node [midway, anchor=north west] {$\mathcal{I}^-$} node [anchor=north] {$i^-$};
  \draw[thin, dashed] (0,1.25) -- (0.75,2);
  \draw[thin,postaction={decorate, decoration={markings, mark=between positions 0.1 and 1 step 0.1 with {\arrow{stealth}}}}] (1,1) -- (0,2);
\end{tikzpicture}\vspace{1em}
  \caption{Penrose diagram that describes the collapse of a null shell. The event horizon is shown as a dashed line.}
  \label{fig:penrose-vaidya}
\end{figure}

We consider the implosion of a spherically symmetric null shell as described in Figure~\ref{fig:penrose-vaidya}. The points in spacetime in whose future the implosion event lies constitute the flat portion of the whole spacetime, whereas the ones that are in the future of the implosion consitute the curved (Schwazrschild) portion.

Both in flat and curved portions of the spacetime we use spherical coordinates. In each region the metric reads:

\begin{equation}
  ds^2 = 
  \begin{cases}
     -dt^2 + dr^2 + r^2 d\Omega^2, & \mbox{(flat)}\\
     -(1-\frac{2M}{r}) dt^2 + (1-\frac{2M}{r})^{-1} dr^2 + r^2d\Omega^2 .& \mbox{(curved)}
  \end{cases}
\end{equation}

Because of the spherical symmetry present in each region, we can decompose the field solutions into spherical harmonics:

\begin{equation}
  \frac{h_l(t,r)}{r} Y_{l,m}(\theta,\phi).
\end{equation}

This combination satisfies the field equation $\nabla^2 \phi = 0$. In the in-region, we name $h_l$'s as $h^{\text{in}}_l$ and in the out-region as $h^{\text{out}}_l$. These two functions satisfy respectively the following equations:

\begin{align}
  \left[ -\partial_t^2 + \partial_r^2 - \frac{l(l+1)}{r^2}\right] h^{\text{in}}_l &= 0,\label{eq:49}\\
  \left[ -\partial_t^2 + \partial_{r^*}^2 - V_l(r)\right] h^{\text{out}}_l &= 0,\label{eq:50}
\end{align}

where $r^*$ is known as the \emph{tortoise coordinate}\ic{tortoise coordinates} that is defined through $r^* \equiv r + 2M \ln (r/2M-1)$ and the potential term is:

\begin{equation}
  V_l(r) = \left( 1 - \frac{2M}{r} \right) \left[ \frac{l(l+1)}{r^2} + \frac{2M}{r^3} \right].
\end{equation}

Using the two given metrics, it is an easy though somewhat lengthy process to derive these.

Since $\mc I^-$ is the lightlike past, there can only be incoming waves present in this region. Similarly, there can only be outgoing wave solutions that reach $\mc I^+$. Because these regions are approached as $r \to \infty$, we can easily solve the equations asymptotically and obtain positive frequency solutions in $\mc I^-$ and $\mc I^+$. As done in \cite{Fabbri2005}, we will first neglect the centrifugal and potential terms in the equations above. The importance of these terms is to yield the graybody factors. After all, there is only one spacetime and one field equation on it. When we trace back the evolution of a positive frequency outgoing mode in $\mc I^+$ backwards, the solution will be effected by the presence of these terms and as it reaches $\mc I^-$ it will be a superposition of negative and positive frequency incoming modes. In the end, these terms determine the graybody factors by determining which modes at $\mc I^-$ are present in which amount. We will handle the graybody factors later on.

Outgoing positive frequency solution near $\mc I^+$ is $e^{-i\omega \uout} Y_{l,m}(\theta,\phi) / r$, and the incoming one near $\mc I^-$ is $e^{-i\omega \vin} Y_{l,m}(\theta,\phi) / r$. For notational consistency we defined $\uout = t-r^*, \vout = t+r^*$ and $\uin = t-r, \vin = t+r$. The coordinates $t,r$ in each case refer to the ones in the corresponding metric: there is no reason that they are the same, but we will match the two metrics along the null infalling shell. This last step is required in determining the form of $\fout$ around $\mc I^-$.

The above mentioned solutions are unnormalized. In order to normalize them, we use the $U(1)$-inner product\ic{U(1)@$U(1)$-inner product} which are as follows for each region:

\begin{align}
  (\phi,\psi)_{\mc I^+} &= -i \int d\Omega du\;  r^2 (\psi \partial_u \phi^* - \phi^* \partial_u \psi), \\
  (\phi,\psi)_{\mc I^-} &= -i \int d\Omega dv\; r^2 (\psi \partial_v \phi^* - \phi^* \partial_v \psi).\label{eq:42}
\end{align}

The integrand in the first equation also involves a factor of $\left(1-\frac{2M}{r}\right)^{-1}$ when it is derived for constant $\vout$ hypersurfaces, however because $\mc I^+$ is reached as $v\to \infty$, this factor approaches $1$. We also divided each inner product by two, because it is unnecessary to take into account these constant factors.

Using these inner products, one can see that the normalized positive frequency solutions are:

\begin{align}
  \fin_\omega &= (4\pi \omega)^{-1/2} \frac{e^{-i\omega \vin}}{r} Y_{lm}(\theta,\varphi), & \fout_\omega &= (4\pi\omega)^{-1/2} \frac{e^{-i\omega \uout}}{r} Y_{lm} (\theta,\varphi).
\end{align}

Because whether the wave is ingoing or outgoing is already determined by the use of $\vin$ or $\uout$, we have chosen $\omega$ as a subscript instead of the usual $k$.

In the end our aim is to calculate the Bogoliubov coefficients, which includes evaluation of inner products. We choose the hypersurface to evaluate these quantities as lightlike past $\mc I^-$. However, we know the form of $\fout_k$ only near $\mc I^+$; we need to be able to evolve it backward in time to know its form near $\mc I^-$. This process, on the other hand, requires passing from out-region to in-region. For that purpose, we should match the two metrics along the imploding null shell.

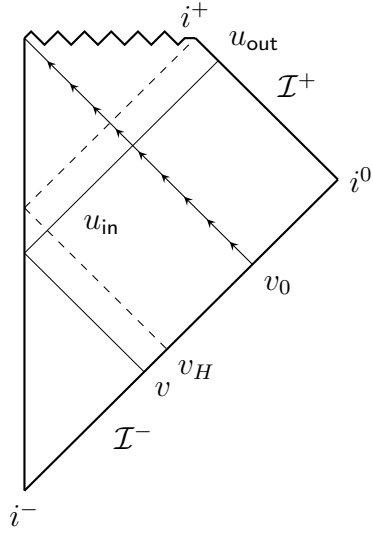
\begin{figure}
  \centering
  \begin{tikzpicture}[scale = 3, thick]
  \draw (0,0) -- (0,2);
  \draw[decorate, decoration=zigzag] (0,2) -- (0.75,2) node [anchor=south] {$i^+$};
  \draw (0.75,2) -- (1.375,1.375)  node [midway, anchor=south west] {$\mathcal{I}^+$} node [anchor=west] {$i^0$};
  \draw (1.375,1.375) -- (0,0) node [near end, anchor=north west] {$\mathcal{I}^-$} node [anchor=north] {$i^-$};
  \draw[thin, dashed] (0,1.25) -- (0.75,2);
  \draw[thin,postaction={decorate, decoration={markings, mark=between positions 0.1 and 1 step 0.1 with {\arrow{stealth}}}}] (1,1) -- (0,2);
  \draw (1,1) node [anchor=north west] {$v_0$};
  \draw[thin,dashed] (0,1.25) -- (0.625,0.625) node [anchor=north west] {$v_H$};
  \draw[very thin] (0.85,1.9) node [anchor=south west] {$\uout$} -- node [near end, below, anchor=north west] {$\uin$} (0,1.05) -- (0.525,0.525) node [anchor=north west] {$v$};;
\end{tikzpicture}\vspace{1em}
  \caption{Gravitational collapse scenario via radially imploding null shell.}
  \label{fig:vaidya-more}
\end{figure}

We basically require $r$ coordinate to be the same on $v = v_0$, that is $r_{\text{in}}(v_0,\uin) = r_{\text{out}}(v_0,\uout)$. Please see Figure~\ref{fig:vaidya-more}. In the in-region we have $r_{\text{in}} = (\vin-\uin)/2$ and in the out-region $r^* = r_{\text{out}} + 2M \ln(r_{\text{out}}/2M -1) =(\vout-\uout)/2$. Evaluating the last equality at $v = v_0$ and using $r_{\text{in}}(v_0,\uin)$ in place of $r_{\text{out}}(v_0,\uout)$ we obtain the following relation:

\begin{equation}\label{eq:uout-intermsof-uin}
  \uout = \uin - 4M \ln\left( \frac{v_0-\uin}{4M} - 1 \right).
\end{equation}

We would like to evolve an outgoing mode near $\mc I^+$ backward in time. It is $e^{-i\omega \uout}/r$ to begin with\footnote{We suppress the spherical harmonic part and the normalization constant.}. When crossing $v=v_0$ we should use (\ref{eq:uout-intermsof-uin}) in order to find the value of $\uin$ that matches to the given $\uout$. Now the wave has the form $\exp(-i\omega \uout(\uin)) / r$. We can trace this mode back to $r=0$ where it turns into an ingoing mode coming from $\mc I^-$. Because $\uin = t+r$ and $\vin = t-r$, at the origin these two coordinates have equal values. We should interchange $\uin$ in $\uout(\uin)$ with $v$ that has the same value as $\uin$. Now near $\mc I^-$, we have the solution $\exp(-i\omega \uout(v))/r$. There is one subtlety to consider. For $v > v_H$, as can be seen from Figure~\ref{fig:vaidya-more}, incoming modes do not reach $\mc I^+$ but instead are trapped behind the horizon where in the end they hit the singularity. For that purpose we need to multiply what we have found by an appropriate Heaviside step function\ic{Heaviside step function} in order to indicate that the support of these modes at $\mc I^-$ is $v \in (-\infty,v_H)$. Because there should be no confusion, we use the letter `$\theta$' to denote this function. The factor we need to multiply the solution with is $\theta(v_H-v)$. Finally, we see that when traced backwards, an outgoing mode $\exp(-i\omega\uout)/r$ happens to be $\exp(-i\omega\uout(v))\theta(v_H-v)/r$.

There is another subtlety, on the other hand, and is not vitally important to find the form of the mode near $\mc I^-$. It is about the relative phase between the incoming and the resulting outgoing mode that occurs after the wave hits $r=0$ in the flat region. We trace back the evolution of the following outgoing mode:

\begin{equation}
  \fout_\omega = (4\pi\omega)^{-1/2} \frac{e^{-i\omega\uout}}{r} Y_{lm}.
\end{equation}

When passing the imploding null shell, we expect the wave to be continuous, for that purpose the matching wave in the in-region is:

\begin{equation}
  (4\pi\omega)^{-1/2} \frac{e^{-i\omega\uout(\uin)}}{r} Y_{lm}.\label{eq:10}
\end{equation}

We know on the other hand, when traced back to $\mc I^-$, what is obtained should be proportional to the following:

\begin{equation}
  (4\pi\omega)^{-1/2} \frac{e^{-i\omega\uout(v)}}{r} Y_{lm}.\label{eq:11}
\end{equation}

In the end the whole solution in the in-region should be what is in (\ref{eq:10}) plus a constant times (\ref{eq:11}). Because there is no source at the origin, this constant should be one and we obtain the whole solution in the in-region as follows:

\begin{equation}
  (4\pi\omega)^{-1/2} \left( \frac{e^{-i\omega\uout(\uin)}}{r} - \frac{e^{-i\omega\uout(v)}}{r} \theta(v_H-v) \right) Y_{lm}.
\end{equation}

Therefore, we conclude that when traced backwards the outgoing mode $\fout_\omega$ is found to be the following near $\mc I^-$:

\begin{equation}
  -(4\pi\omega)^{-1/2} \frac{e^{-i\omega\uout(v)}}{r} \theta(v_H-v) Y_{lm}.
\end{equation}

After this point, we need to calculate the Bogoliubov coefficients between the in-modes and the out-modes. Since the angular integration of the inner product (\ref{eq:42}) is over $S^2$ with no weight and the field modes are proportional to spherical harmonics which are orthonormal, we suppress angular indices such as $l,m$ in Bogoliubov coefficients and simply write $\alpha_{\omega\omega'}, \beta_{\omega\omega'}$. Because nonvanishing Bogoliubov coefficients must have the same angular indices on both sides. For a derivation of the coefficients, the reader is referred to Appendix~\ref{chap:bog-coef-Hawking}. The result is:

\begin{align}
  \alpha_{\omega\omega'} &= -i\sigma \left( \frac{\omega' M}{\pi} \right)^{1/2} \frac{e^{\sigma 2\pi \omega M}}{\sinh^{1/2}(4\pi \omega M)} \frac{e^{i(\omega' - \omega)v_H}}{(4M)^{i4\omega M} \abs{\omega' - \omega}^{1+i4\omega M}},\\
  \beta_{\omega\omega'} &= i \left( \frac{\omega' M}{\pi} \right)^{1/2} \frac{e^{- 2\pi \omega M}}{\sinh^{1/2}(4\pi \omega M)} \frac{e^{-i(\omega' + \omega)v_H}}{(4M)^{i4\omega M} (\omega' + \omega)^{1+i4\omega M}},
\end{align}

where $\sigma = \sgn(\omega'-\omega)$. If we used these Bogoliubov coefficients, the number of particles emitted from the black hole an observer would detect near $\mc I^+$ would diverge. For this reason, similar to what Hawking\ip{Hawking, Stephen W.} did in his seminal article \cite{hawking1975} we shall construct wave packets out of out-modes. We define a wave packet $g_{jn}$, not to be confused with the metric, as follows:

\begin{equation}
  g_{jn} \equiv \ve^{-1/2} \int_{j\ve}^{(j+1)\ve} d\omega \; e^{-i2\pi n \omega/\ve} \fout_\omega \quad ; j \in \mathbb{Z}^{\geq 0}, n \in \mathbb{Z}.
\end{equation}

It is an easy exercise to show that $(g_{jn},g_{j'n'}) = \delta_{jj'}\delta_{nn'}$ and $(g_{jn},g^*_{j'n'}) = 0$. We can expand $g_{jn}$ in terms of in-modes and read out the new semi-discrete Bogoliubov coefficients:

\begin{equation}
  g_{jn} = \int d\omega \; (\alpha_{jn\omega} \fin_\omega + \beta_{jn\omega} \fins_\omega).
\end{equation}

Similar to what we did before, we can express the annihilation operator $b_{jn}$ related to these wave packets in terms of $\ain_\omega$ and $\aind_\omega$:

\begin{equation}
  b_{jn} = \int d\omega \; (\alpha_{jn\omega}^* \ain_\omega - \beta_{jn\omega}^* \aind_\omega).
\end{equation}

It then follows that the expectation value of the number operator $b_{jn}^\dagger b_{jn}$ in the in-vacuum is:

\begin{equation}
  \bra{\text{in}} b_{jn}^\dagger b_{jn} \invac = \int d\omega' \; \abs{\beta_{jn\omega'}}^2. \label{eq:44}
\end{equation}

Because $j\ve \sim \omega$ when $\ve \ll 1$, we have chosen the integration variable as $\omega'$.

Expressing the result of the inner-product between these wave packets, one can obtain relations between these coefficients. For example, the analog of (\ref{eq:bog-delta}) reads as:

\begin{equation}
  \int d\omega' \; (\alpha^*_{jn\omega'}\alpha_{j'n'\omega'} - \beta^*_{jn\omega'}\beta_{j'n'\omega'}) = \delta_{jj'} \delta_{nn'}.
\end{equation}

For $j=j', n=n'$ we obtain:

\begin{equation}
  \int d\omega' \; (\abs{\alpha_{jn\omega'}}^2 - \abs{\beta_{jn\omega'}}^2) = 1.\label{eq:45}
\end{equation}

Finding a relation between $\abs{\alpha_{jn\omega'}}$ and $\abs{\beta_{jn\omega'}}$ we can use (\ref{eq:45}) to calculate (\ref{eq:44}). Rest of the Appendix~\ref{chap:bog-coef-Hawking} focuses on semi-discrete Bogoliubov coefficients.

These are found to be as follows:

\begin{align}
  \alpha_{jn\omega'} &= -i\sigma h(\omega';\ve) \int_{j\ve}^{(j+1)\ve} \negmedspace\negmedspace d\omega \;
     \frac{e^{\sigma 2\pi \omega M}}{\sinh^{1/2}(4\pi\omega M)}
     \frac{e^{-i(v_H + 2\pi n /\ve)\omega}}{(4M)^{i4\omega M} \abs{\omega' - \omega}^{1+i4\omega M}},\label{eq:46}\\
   \beta_{jn\omega'} &= i h^*(\omega';\ve) \int_{j\ve}^{(j+1)\ve} \negmedspace\negmedspace d\omega \; 
     \frac{e^{- 2\pi \omega M}}{\sinh^{1/2}(4\pi\omega M)}
     \frac{e^{-i(v_H + 2\pi n /\ve)\omega}}{(4M)^{i4\omega M} (\omega' + \omega)^{1+i4\omega M}},\label{eq:47}
\end{align}

where $h(\omega';\ve) = e^{i\omega' v_H} (\omega' M / \pi \ve)^{1/2}$.

Now, we shall suppose $\ve \ll 1$ and $M \ve \gg 1$ to hold at the same time. This is reasonable. After all we are interested in masses $M$ that is well over the Planck mass\ic{Planck mass} ($m_{\text{Planck}} = (\hbar c / G)^{1/2} = 2.18\times 10^{-8}$~kg) where classical spacetime picture, general relativity, holds. For example, solar mass equals $1.99 \times 10^{30}$~kg \cite{solarMassNasa}. In Planck units\ic{Planck units}, it equals $9.14\times 10^{37}$. So one has pretty much freedom in satisfying $\ve \ll 1$ and $M \ve \gg 1$ at the same time.

Notice that the integration over $\fout_\omega$ to construct wave packets $g_{jn}$ is taken from $\omega = j\ve$ to $\omega = j\ve + \ve$. Since $\ve \ll 1$, we can say that the frequency, $\omega$, of the wave packet $g_{jn}$ is about $j\ve$: so $\omega \approx j\ve$.

We focus on the Hawking radiation observed at late times. Since an outgoing mode at late times when traced back over time is \emph{highly} blue shifted (see equation (\ref{eq:11})) we expect that almost all of the contribution to the integral $\int d\omega' \; \abs{\beta_{jn\omega'}}^2$ comes from the region where $\omega' \gg \omega$.

Now, if $M\ve \gg 1$, the exponentials of real-valued arguments appearing in (\ref{eq:46}) and (\ref{eq:47}) are the dominant terms that determine the result of the integral. On the other hand for $\omega' \gg \omega$, one has $\omega'-\omega \approx \omega'+\omega$. Therefore, in this regime, one can conclude that $\abs{\alpha_{jn\omega'} / \beta_{jn\omega'}} = \exp(4\pi\omega M)$, where we wrote $\omega$ in place of $j\ve$. Also remember that $\sigma = 1$.

On the contrary, if $\omega' \ll \omega$, one may say that $\abs{\alpha_{jn\omega'} / \beta_{jn\omega'}} = 1$. Let us separate the integral (\ref{eq:45}) in three parts:

\begin{equation}
  \left(\int_0^{\ll \omega} + \int_{\sim \omega} + \int_{\gg \omega}^\infty \right) d\omega' \; (\abs{\alpha_{jn\omega'}}^2 - \abs{\beta_{jn\omega'}}^2) = 1.
\end{equation}

Because $\abs{\alpha_{jn\omega'} / \beta_{jn\omega'}} = 1$ in the first part, there is no contribution coming from this part of the integral. We neglect the middle part where $\omega' \sim \omega$ because at late times we expect backtraced waves to be high frequency and the resulting Bogoliubov coefficients between high and low frequency modes to be negligible. It is this part we neglect when considering the late time radiation.

In the last part, because $\abs{\alpha_{jn\omega'} / \beta_{jn\omega'}} = \exp(4\pi\omega M)$ holds we can write:

\begin{equation}
  \int_{\gg \omega}^\infty d\omega' \; (e^{8\pi\omega M} - 1) \abs{\beta_{jn\omega'}}^2 = 1\quad\text{(at late times)}.
\end{equation}

Finally, we obtain the desired result for the celebrated Hawking radiation:

\begin{equation}
  \bra{\text{in}} b_{jn}^\dagger b_{jn} \invac = \frac{1}{e^{8\pi \omega M} - 1}.
\end{equation}

When we compare this equation with the Planckian distribution $1/(e^{\hbar \omega / kT} - 1)$ (in the text we set $\hbar = k = 1$) we see that Hawking radiation is at a temperature of $T = 1/8\pi M$, which is known as the \emph{Hawking temperature}\ic{Hawking temperature} of a black hole.

Finally, we would like to comment on the graybody factors. In order to fully account for the Hawking radiation, we should not have omitted the centrifugal and potential terms in equations (\ref{eq:49}) and (\ref{eq:50}). This would in turn reduce the amplitude of the wave that reaches as Hawking radiation to $\mc I^+$.  This is a scattering problem: some portion of the wave returns to event horizon whereas the remaining portion is observed as black hole vapor. For that purpose, we should map the outgoing wave packets $g_{jn}$ to $g'_{jn} = \text{(phase)} \times \Gamma_{jn}^{1/2} g_{jn}$ for some $\Gamma_{jn} > 0$. The coefficient $\Gamma_{jn}$ should not depend on any other parameter, apart from angular $l,m$ that we have suppressed so far, as Bogoliubov coefficients do. Because this is a scattering problem and in providing the indices `$j,n$' we give all the information about the wave. Therefore, while the backtraced mode near $\mc I^-$ will be normalized, the modes that reach $\mc I^+$ will give the following inner product $(g'_{jn}, g'_{j'n'}) = \Gamma_{jn} \delta_{jj'} \delta_{nn'}$. Therefore the relation (\ref{eq:45}) would be modified as:

\begin{equation}
  \int d\omega' \; (\abs{\alpha'_{jn\omega'}}^2 - \abs{\beta'_{jn\omega'}}^2) = \Gamma_{jn},
\end{equation}

where $\alpha'_{jn\omega'} = (\fin_{\omega'}, g'_{jn})$ and $\beta_{jn\omega'} = -(\fins_{\omega'}, g'_{jn})$ are the new Bogoliubov coefficients. Because the change with respective to earlier Bogoliubov coefficients is only in magnitude, the ratio between them will remain intact. Hence the spectrum we would obtain for particles that reach $\mc I^+$ is:

\begin{equation}
  \bra{\text{in}} b_{jnlm}^\dagger b_{jnlm} \invac = \frac{\Gamma_{jnlm}}{e^{8\pi \omega M} - 1},
\end{equation}

where we have written angular dependency explicitly.

Because the expansion of the field is the same as before, therefore is the relation of $b_{jn}$ to $\ain_{\omega'}$ and $\aind_{\omega'}$.

For more discussion on quantum field theory in curved spacetimes, readers may like to consult \cite{mukhanov, parker-toms, BirrellDavies}.

\section{Quantum Entanglement}\ic{entanglement|textbf}
\label{sec:qit}

Our main \emph{guide} in this section will be the lecture notes of John Preskill\ip{Preskill, John} \cite{preskill-ln}. Therefore we do not refer to this source everytime we use it. When other sources are used, we shall make explicit references.

Entanglement is a phenomenon that is not found in classical physics and is uniquely of quantum nature. When the Hilbert space ($\mathcal{H}$) describing the quantum system of interest is bipartite, \ie $\mathcal{H} = \mathcal{H}_A \otimes \mathcal{H}_B$ for some $\mathcal{H}_A,\mathcal{H}_B$, entangled states are present in $\mathcal{H}$. An entangled state\ic{state!entangled} is one that cannot be written as a product state\ic{state!product}\footnote{A product state $\ket \psi \in \mathcal H$ is a state that can be written as $\ket \psi = \ket a \otimes \ket b$ for some $\ket a \in \mathcal H_A,\ket b \in \mathcal H_B$.}.

As it is the case with almost any concept, there are \emph{degrees} of entanglement. For that purpose, we turn to Schmidt decomposition\ic{Schmidt!decomposition} of quantum states.

We consider a bipartite Hilbert space: $\mc H = \mc H_A \otimes \mc H_B$ and a state $\ket \psi$ in it. Because it is in $\mc H$, we can express it as follows:

\begin{align}
  \ket \psi &= \sum_{ij} a_{ij} \ket i \otimes \ked j ,\\
\intertext{where $\{\ket i\}_i,\{\ked j\}_j$ are two orthonormal bases in $\mc H_A, \mc H_B$. Let us define new vectors in $\mc H_B$ as $\ked{i'} \equiv \sum_j a_{ij}\ked j$, where some $\ked{i'}$ vectors may be zero depending on $a_{ij}$. We write $\ket \psi$ as:}
\ket \psi &= \sum_i \ket i \otimes \ked{i'}.
\end{align}

On the other hand, the system as a whole has the density matrix $\rho = \ket \psi \bra \psi$. By performing a partial trace on part $B$, one can obtain a reduced density matrix\footnote{Explicitly, $\rho_A = \tr_B \rho = \sum_i \brd i \rho \ked i$. Notice that whereas $\rho$ acts on the states in the whole Hilbert space, $\rho_A$ acts only on the ones that are in $\mc H_A$. This is easy to see. One can expand $\rho$ as $\rho = \sum_{ijkl} \rho_{ijkl} \ket i \bra j \otimes \ked k \brd l$ and once the trace over the part $B$ is performed, what is obtained is again a density matrix but this time the one that acts only on $\mc H_A$.} $\rho_A$ that can be used to calculate the expectation values of experiments performed only on part $A$, disregarding the existence of part $B$. Because density matrices are Hermitian, they can be diagonalized and $\rho_A$ is no exception. We now choose $\{\ket i\}_i$ as the basis in which $\rho_A$ is diagonal, that is:

\begin{equation}
  \label{eq:8}
  \rho_A = \sum_i p_i \ket i \bra i .
\end{equation}

What we are going to do is to calculate $\rho_A$ by performing the partial trace of $\rho$ over part $B$ explicitly.

\begin{align}
  \rho_A &= \tr_B \rho,\\
         &= \sum_i \brd i \rho \ked i ,\\
         &= \sum_i \brdket{i | \psi} \braked{\psi | i} ,\\
  \intertext{Using $\ket \psi = \sum_j \ket j \otimes \ked{j'}$, we continue}
         &= \sum_{ijk} \brdked{i | k'} \ket k \bra j \brdked{j' | i} ,\\
         &= \sum_{jk} \brdked{j' | k'} \ket k \bra j .
\end{align}

However we know that $\rho_A$ is diagonal in the basis we have chosen for $\mc H_A$. Therefore $\brdked{j' | k'}$ must be proportional to Kronecker delta and moreover it must yield the correct eigenvalue of $\rho_A$. In other words, $\brdked{j' | k'} = \delta_{jk} p_{j}$ must be satisfied where the summation convention is not used.

We see that the vectors $\ked{i'}$ turn out to be orthogonal after all, albeit not necessarily orthonormal\footnote{If they are orthonormal, that means some $p_{i'}$ is equal to one. However, the condition that the eigenvalues (which are nonnegative) of any density matrix sum up to one requires all other eigenvalues to be vanishing in this case. Hence there happens to be only one \emph{nonzero} $\ked{i'}$ vector.}. We would like to rewrite these vectors as its norm times a unit vector. For that purpose we are in need of notation for these new vectors. It would be nice to use $\ket{i'}$ to denote these new vectors in $\mc H_B$. First it is a clean notation, second it somehow hints at the dependence of $\ket{i'}$ on the particular basis, that diagonalizes $\rho_A$, chosen in $\mc H_A$.

In the end, we have shown that we can write any vector $\ket \psi$ in a bipartite Hilbert space in the following form known as the \emph{Schmidt form}\ic{Schmidt!form}:

\begin{equation}
  \label{eq:9}
  \ket \psi = \sum_i p_i^{1/2} \ket i \otimes \ket {i'}.
\end{equation}

The number of nonzero eigenvalues of $\rho_A$, \ie nonzero $p_i$ values, is called the \emph{Schmidt number}\ic{Schmidt!number} of $\ket \psi$. A product state has the Schmidt number one, of course. Therefore, another way of saying whether the state is entangled or not is to look if its Schmidt number\ic{Schmidt!number} is greater than one. If it is, the state is entangled.

We have classified some states in a bipartite Hilbert state as entangled and mentioned that there are degrees of entanglement. Here we shall mention a particular type of entanglement: \emph{maximal entanglement}\ic{entanglement!maximal}. We quote \cite{Susskind2012} as regards its meaning:

\begin{quote}
  The meaning of maximal entanglement is that for every observable in
  $A$, one can predict the result of measuring it by measuring the
  corresponding observable in $B$.
\end{quote}

A state is called maximally entangled\ic{entanglement!maximal!definition} if and only if the reduced density matrix $\rho_A$ is proportional to identity operator in $\mc H_A$ \cite{Susskind2012}. We provide a mathematically precise interpretation of the quote above, in the following theorem.

\begin{thm}
  A bipartite Hilbert space $\mc H = \mc H_A \otimes \mc H_B$ is considered where $N \equiv \dim \mc H_A = \dim \mc H_B$. Then, a state $\ket \psi \in \mc H$ is maximally entangled if and only if it can be written as $\ket \psi = \sum_i \alpha_i \ket{a_i} \otimes \ket{A_i}$ for some $\alpha_i \in \comps$ where $\ket{a_i}$'s are orthonormal eigenvectors of any chosen Hermitian operator $\mc O$ and $\ket{A_i}$'s are eigenvectors of $U^\dagger \mc O' U$ which is similar to $\mc O$ and $U$ is a unitary matrix.
\end{thm}

\begin{proof}
  $(\Rightarrow)$. First, we suppose that $\ket \psi$ is maximally entangled. In the Schmidt decomposition, we used a basis for $\mc H_A$ that diagonalizes the reduced density matrix $\rho_A$. Because $\rho_A$ is proportional to identity, every orthonormal basis satisfies this condition. For any Hermitian operator $\mc O$, as an orthonormal basis, we can choose its eigenvectors $\ket{a_i}$. If there are repeated eigenvalues, it should be understood that eigenvectors of the same eigenvalue are made orthonormal to each other through the Gram-Schmidt process. After the Schmidt decomposition, we can write $\ket \psi$ as follows:

  \begin{equation}
    \ket \psi = N^{-1/2} \sum_i \ket{a_i} \otimes \ket{A_i}.
  \end{equation}

  The point is that $\ket{A_i}$'s are orthonormal as well. Because of this fact, there exists a unitary operator, which is of course invertible, $U$ such that $\ket{a_i} = U \ket{A_i}$. It then follows that if $\ket{a_i}$ is an eigenvector of $\mc O$ with eigenvalue $\lambda_i$, $\ket{A_i}$ is an eigenvector of $\mc O' = U^\dagger \mc O U$ with the same eigenvalue. First part of the proof is complete.

  $(\Leftarrow)$. Second, we suppose that for every Hermitian operator $\mc O$ we can write $\ket \psi$ in the following form:

  \begin{equation}
    \ket \psi = \sum_i \alpha_i \ket{a_i} \otimes \ket{A_i},\label{eq:37}
  \end{equation}

where $\ket{a_i}$'s are orthonormal eigenvectors of $\mc O$ and $\ket{A_i}$'s are those of a similar operator $\mc O'$. We would like to show that $\forall i, \alpha_i \neq 0$.

Suppose $\exists i, \alpha_i = 0$, and renumerate the indices such that $\alpha_1 = 0, \alpha_2 \neq 0$. Instead of $\mc O$, consider the following operator:

\begin{multline}
  \mc T = \frac \lambda 2 (\ket{a_1} + \ket{a_2})(\bra{a_1} + \bra{a_2}) + \frac{3\lambda}{2} (\ket{a_1} - \ket{a_2})(\bra{a_1} - \bra{a_2})\\ + \sum_{i=3}^N \lambda_i \ket{a_i}\bra{a_i},
\end{multline}

where $\lambda_i$ is the eigenvalue of $\mc O$ with eigenvector $\ket{a_i}$ and $\lambda$ is some nonzero real number. In the expansion (\ref{eq:37}) of $\ket \psi$, $\alpha_1$ is already zero, it does not matter if $\ket{a_1}$ is not an eigenvector of $\mc T$. However, even though all $\ket{a_i}$ for $i>2$ are eigenvectors of $\mc T$; $\ket{a_2}$ is not one of its eigenvectors. We have obtained a contradiction. Therefore, $\forall i, \alpha_i \neq 0$ must be true.

When we show all of $\alpha_i$'s have equal magnitude, the proof will have been completed. Our approach will be similar. We consider an operator $\mc O$ and expand $\ket \psi$ in the same way:

\begin{equation}
  \ket \psi = \sum_i \alpha_i \ket{a_i} \otimes \ket{A_i}, \label{eq:56}
\end{equation}

however this time we know that $\forall i, \alpha_i \neq 0$. Let us choose $\mc O$ that has distinct eigenvalues. We have already required $\ket{a_i}$'s to constitute an orthonormal set, however because $\mc O$ has distinct eigenvalues, $\mc O'$ that is similar to it does have distinct eigenvalues as well. Eigenvalues are independent of the basis chosen. We conclude that $\ket{A_i}$'s are orthonormal to each other.

We suppose there are two $\alpha_i$ values whose moduli are not equal to each other. After renumerating the indices in the expansion of $\ket \psi$, we say these are $\alpha_1$ and $\alpha_2$.

Instead of the operator $\mc O$, we could have chosen $\tilde{\mc O}$ that has eigenvectors $\ket +$, $\ket -$, $\ket{a_3}, \ldots$ . First of these two are given by the following definitions:

\begin{align}
  \ket + &= \frac{\ket{a_1} + \ket{a_2}}{\sqrt 2} & \ket - &= \frac{\ket{a_1} - \ket{a_2}}{\sqrt 2}.
\end{align}

Because $\ket{a_1}$ and $\ket{a_2}$ are unitarily related to $\ket{+}$ and $\ket{-}$, the operator $\tilde{\mc O}$ is similar to $\mc O$. When expanded in the eigenvectors of $\tilde{\mc O}$ and some similar operator which we call $\tilde{\mc O}'$, the state $\ket \psi$ can be written as follows:

\begin{align}
  \ket \psi = \beta_1 \ket + \otimes V \ket{A_1} + \beta_2 \ket - \otimes \ket{A_2} + \sum_{i=3}^N \beta_i \ket{a_i} \otimes V \ket{A_i},\label{eq:57}
\end{align}

where $V$ is a some unitary operator that relates $\{\ket{A_i}\}_i$ to eigenvectors of $\tilde{\mc O}'$. The two expansions, (\ref{eq:56}) and (\ref{eq:57}), must be equal to each other. It is easily seen that $V$ should leave $\ket{A_i}$ for $i>2$ intact. When we write $\ket{a_1}$ and $\ket{a_2}$ in terms of $\ket{+}$ and $\ket{-}$ in equation (\ref{eq:56}), we obtain the following:

\begin{align}
  \beta_1 \ket{+} \otimes V\ket{A_1} &= \left( \frac{\abs{\alpha_1}^2 + \abs{\alpha_2}^2}{2} \right)^{1/2} \ket{+} \otimes \frac{\alpha_1 \ket{A_1} + \alpha_2 \ket{A_2}}{(\abs{\alpha_1}^2 + \abs{\alpha_2}^2)^{1/2}},\\
  \beta_2 \ket{-} \otimes V\ket{A_2} &= \left( \frac{\abs{\alpha_1}^2 + \abs{\alpha_2}^2}{2} \right)^{1/2} \ket{+} \otimes \frac{\alpha_1 \ket{A_1} - \alpha_2 \ket{A_2}}{(\abs{\alpha_1}^2 + \abs{\alpha_2}^2)^{1/2}}.
\end{align}

We define $\ket P, \ket M$ through $\ket P = V \ket{A_1}, \ket M = V \ket{A_2}$. The above equations give us $\ket P, \ket M$ in terms of $\ket{A_1},\ket{A_2}$:

\begin{align}
  \ket P = \frac{\alpha_1 \ket{A_1} + \alpha_2 \ket{A_2}}{(\abs{\alpha_1}^2 + \abs{\alpha_2}^2)^{1/2}}, \quad \ket M = \frac{\alpha_1 \ket{A_1} - \alpha_2 \ket{A_2}}{(\abs{\alpha_1}^2 + \abs{\alpha_2}^2)^{1/2}}.
\end{align}

The operator $V$ can be expressed in Dirac notation as: $V = \ket P \bra{A_1} + \ket M \bra{A_2} + \sum_{i=3}^N \ket{A_i}\bra{A_i}$. Since $V$ is unitary, $V^\dagger V = \one$ must hold. Working in the $\{\ket{A_i}\}_i$ basis, unitarity of $V$ is seen to imply $\abs{\alpha_1}^2 - \abs{\alpha_2}^2 = 0$. We obtain a contradiction, because we assumed $\alpha_1$ and $\alpha_2$ to be of different magnitudes. Therefore, $\forall i,j; \abs{\alpha_i} = \abs{\alpha_j}$. So $\ket \psi$ is a maximally entangled state. This completes the proof.
\end{proof}

In order to get acquainted with the concept let us look at a simple example concerning two qubits\ic{qubit}. We remind the reader that a \emph{qubit}\ic{qubit!definition} is a two level system. It can be considered as the spin degrees of freedom of a spin-1/2 particle, but not necessarily so.

The quantum state defined as $\ket \psi = 2^{-1/2} (\ket + \otimes \ket + + \ket - \otimes \ket -)$ is maximally entangled\ic{entanglement!maximal!example}. We perform a spin-$z$ measurement. Then, if we observe the second particle in $\pm z$ spin state, the first one will be in this state as well. On the other hand, suppose we changed our mind and wanted to make a spin-$x$ measurement instead. We can rewrite the same quantum state in the $x$-basis:

\begin{equation}
  \ket \psi = \frac{\ket{x+} \otimes \ket{x+} + \ket{x-} \otimes \ket{x-}}{\sqrt{2}}.
\end{equation}

Here we find the same $++, --$ correlation between the spin-$x$ measurements of the particles. However, the state $\ket \psi$ we used is somehow special. If we chose $\ket \psi \propto \ket{n+} \otimes \ket + + \ket{n-} \otimes \ket -$, we would know that the result of the first particle's spin measurement would yield, however it would not necessarily point in the same direction as the second particle's spin.

If the quantum state chosen is not maximally entangled, then we lose our ability to predict what result would be obtained for a similar experiment on the first particle. For example, consider the following state:

\begin{equation}
  \ket \phi = \frac{\ket{+} \otimes \ket{+} + \ket{x+} \otimes \ket{-}}{\sqrt{2}}.
\end{equation}

This is not a maximally entangled state. The reduced density matrix $\rho_A$ in the $z$-basis equals:

\begin{equation}
  \rho_A = \frac{1}{4}
  \begin{pmatrix}
    3 & 1\\
    1 & 1
  \end{pmatrix}.
\end{equation}

By performing a spin-$z$ measurement on the second particle, one may infer what state the first particle will have been in. However the possibilities will not be eigenstates of an operator that is similar to $S_z$. For this reason, we cannot be sure what we would obtain for a similar measurement done on the first particle. This is all because the state is not maximally entangled: $\braket{+\mid x+} \neq 0$.

\section{Complementarity}\ic{complementarity|textbf}
\label{sec:complementarity}

The concept of \emph{complementarity} has been introduced by Niels Bohr\ip{Bohr, Niels} in 1927 \cite{Plotnitsky2012}. It would be useful to understand this idea through his response \cite{Bohr1935} to Albert Einstein\ip{Einstein, Albert}, Boris Podolsky\ip{Podolsky, Boris} and Nathan Rosen\ip{Rosen, Nathan}'s (EPR) paper \cite{Einstein1935} that put forward the EPR paradox\ic{EPR paradox}.

We first discuss EPR's article \cite{Einstein1935} to highlight the main points, and then return to Bohr's\ip{Bohr, Niels} reply in order to understand the idea of complementarity.

\subsection{EPR Paradox}\ic{EPR paradox|textbf}
\label{sec:epr-paradox}

In the paper \cite{Einstein1935} a particular two-particle state is considered, which we write in ket notation (apart from a constant factor) as follows:

\begin{align}
  \label{eq:7}
  \ket \Psi &= \int_{-\infty}^\infty dp\; e^{i p x_0} \ket{-p} \otimes \ket p,\\
  \intertext{which can be written as}
  \ket \Psi &= \int_{-\infty}^\infty dx \ket{x+x_0} \otimes \ket{x},
\end{align}

where $x_0$ is some free parameter. It is thought that there is no interaction between the two particles after some time. If one makes a position/momentum measurement on the second particle and finds the states $\ket x / \ket p$, then the state of the first particle must be $\ket{x+x_0} / e^{i p x_0} \ket{-p}$. Of course, because position and momentum are noncommuting observables in quantum theory, one cannot specify their values with absolute precision at some instant of time.

The authors postulated the following criterion of reality\ic{Criterion of reality} \cite{Einstein1935}:

\begin{quote}
  If, without in any way disturbing a system, we can predict with
  certainty (i.e., with probability equal to unity) the value of a
  physical quantity, then there exists an element of physical reality
  corresponding to this physical quantity.
\end{quote}

In the measurements that are done only on the second particle, the first one is never disturbed even though the state in which it is in can be known with certainty. It is thought that if wave functions give a complete description of reality, then since one does not disturb the first particle while doing experiments on the second one, its position and momentum can be known simultaneously, hence they ``have simultaneous reality.'' \cite{Einstein1935} (See below).

The main proposition the article depends on is ``either (1) the description of reality given by the wave function in quantum mechanics is not complete or (2) these two quantities [corresponding to noncommuting observables] cannot have simultaneous reality.'' \cite{Einstein1935}.

In logical terms: (1) or (2) is true. However just in the previous paragraph, it has been shown that if (1) is false, then (2) is false as well. Hence, the proposition that (1) or (2) is true turns out false. It is then concluded that (1) must be true: ``the quantum-mechanical description of physical reality given by wave functions is not complete.'' \cite{Einstein1935}.

\subsection{Use of Complementarity by Bohr}
\label{sec:bohrs-response-epr}

The title of the paper \cite{Einstein1935} that put forward the EPR paradox was ``Can Quantum-Mechanical Description of Physical Reality Be Considered Complete?'' and Bohr wrote an article \cite{Bohr1935} with the same title a few months later, in response. Of course, this was clearly a challenge. The essence of Bohr's argument is that the criterion of reality\ic{Criterion of reality} used by EPR is ambiguous and quantum mechanics describes what is there to be explained physically \cite{Bohr1935}.

Bohr considers a single slit experiment, which may be a part of a bigger experimental setup, to begin with. Because the slit can be thought of as performing a position measurement on the particle, even though the particle had some definite momentum before the interaction, its momentum would become uncertain according to the Heisenberg's uncertainty principle\ic{uncertainty principle}\footnote{A proposed analogue experiment for understanding this relation that is remote from daily experience is included in Appendix~\ref{cha:an-anal-heis}.}:

\begin{align}
  \Delta x \Delta p \geq \frac \hbar 2.
\end{align}

Because ``the uncertainty $\Delta p$ is inseparably connected with the possibility of an exchange of momentum between the particle and the diaphragm'' he then asks the question of ``to what extent the momentum thus exchanged can be taken into account in the description of the phenomenon to be studied by the experimental arrangement concerned, of which the passing of the particle through the slit may be considered as the initial stage.'' \cite{Bohr1935}. Afterwards, he notes that since the diaphragm is held fixed in space, the ability to take into account the exchanged momentum is lost.

In order to take into account the momentum exchange, Bohr conceives a similar but distinct experimental setup. In this case the diaphragm is free to move in direction perpendicular to slit, \ie sideways. However, the diaphragm itself becomes an object of experimentation like the electron \cite{Bohr1935}.

All in all, what Bohr\ip{Bohr, Niels} says is that experimental setups that are suited for measurement of position or momentum variables are in essence different. As one learns more about one of these conjugate variables, one's knowledge of the other becomes more uncertain. It is in this way that there is a \emph{complementarity}\ic{complementarity}: one cannot reach complete information about both of the conjugate variables.

Lastly, we would like to comment on where in the EPR's criterion of reality\ic{Criterion of reality} that the claimed ``essential ambiguity'' \cite{Bohr1935} lies, according to Bohr\ip{Bohr, Niels}. It is the part ``without in any way disturbing a system'' \cite{Bohr1935}. In essence, EPR thinks that because one in principle can ``predict \emph{either one or the other} of the two conjugate measurable quantities \ldots quantum objects independently possess \emph{both} of these quantities, even though we can never predict both of them \emph{simultaneously}'' \cite{Plotnitsky2012} (emphasis in the original). As Bohr\ip{Bohr, Niels} has illustrated, one cannot use the same apparatus to measure any one of conjugate variables, that is, the cases when the diaphragm is fixed or not are \emph{different}. The very choice of the instrument as regards which one of the conjugate variables are measured \emph{is} a disturbance to the system, which consists of both the quantum object and the measurement instrument. Then according to Bohr, the argumentation of EPR in the direction that quantum mechanics is incomplete is invalidated because in quantum mechanics one has to consider the measurement apparatus as well. EPR's criterion is ambiguous because it does not take into account the measurement apparatus. 

Readers who look for a more careful analysis of the EPR paradox\ic{EPR paradox} might find reference \cite{Plotnitsky2012} quite useful.

\subsection{Complementarity in Black Hole Physics}
\label{sec:compl-black-hole}

General relativity, as its name alludes to, is \emph{relativistic}. It puts emphasis on the observer. This fact might sometimes be overlooked because relativistic or quantum phenomena are quite far from the realm of conditions under which human intuition has evolved.

Let us imagine a stone that is thrown into a black hole. It will certainly fall and in the end hit the singularity, in \emph{its own perspective}. An observer that falls down with the stone will confirm this series of events, however the confirmation will never reach any observer that remains outside the black hole anyway.

On the other hand, for example, an observer that is far away from the black hole will see that the stone will get closer and closer to the horizon however will \emph{never} cross it. This is \emph{relativity} after all.

There are a few ways to see that the stone never crosses the event horizon according to the distant observer. For concreteness we shall think of a Schwarzshild black hole and a radially infalling stone. First is to solve the geodesic equation for the stone using asymptotic time ($t$) to parametrize the path. As $t \to \infty$ one sees that the position of the stone approaches $r=2M^+$. The second is to use the spacelike slices ($\Sigma_t$) on which asymptotic time is constant. The point where the stone lies in time $t$, is found by taking the intersection of $\Sigma_t$ with the stone's trajectory in spacetime. We would like to show this last observation in more pictorial terms.

For that purpose, we use the Kruskal coordinates\ic{Kruskal coordinates} for the Schwarzschild black hole\ic{black hole!Schwarzschild}. The reader may want to check out section~\ref{sec:intro-gr-bh} on black holes. In these coordinates, a constant $t$ hypersurface is described by the following relation: $T/R = \tanh(t/4M)$. Since $t \in \reals$, the result of $\tanh(t/4M)$ takes all values in the interval $(-1,1)$. So constant $t$ hypersurfaces are represented as straight lines passing through the origin of the coordinate system that has slope ranging between $-1$ and $1$. The slope increases with $t$. We draw the \co{Kruskal diagram} with some sample constant $t$ hypersurfaces on top, in Figure~\ref{fig:kruskal-coor-const-t-slices}.

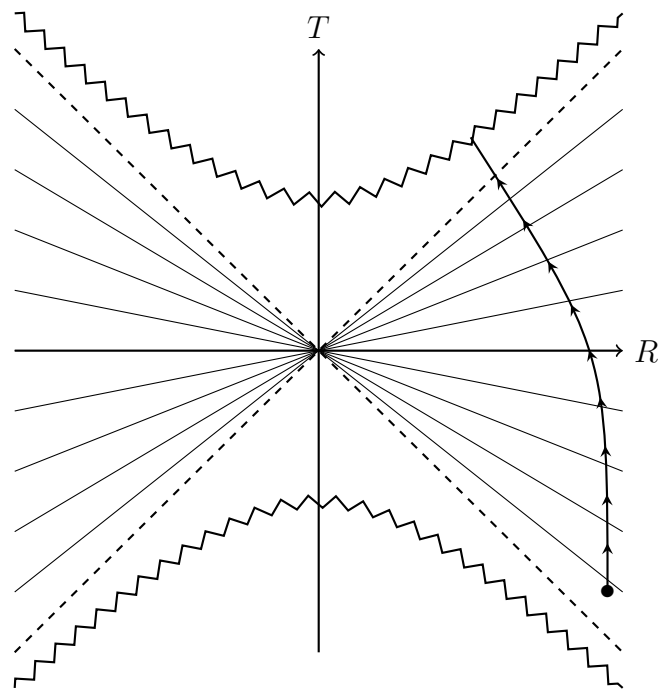
\begin{figure}
  \centering
  \begin{tikzpicture}[scale = 2, thick, domain=-2:2] 
  \draw[decorate, decoration=zigzag] plot (\x, { sqrt(1+\x*\x)});
  \draw[decorate, decoration=zigzag] plot (\x, {-sqrt(1+\x*\x)});
  \draw[dashed] (-2,-2) -- (2, 2);
  \draw[dashed] (-2, 2) -- (2,-2);
  \draw[->] (-2,0) -- (2,0) node [anchor=west] {$R$};
  \draw[->] (0,-2) -- (0,2) node [anchor=south] {$T$};
  \foreach \a in {-0.8,-0.6,...,0.8}
    \draw[thin] plot (\x, {\a*\x});
  \draw[thick,postaction={decorate, decoration={markings, mark=between positions 0.1 and 1 step 0.1 with {\arrow{stealth}}}}] (1.9,-1.6) .. controls (1.9,0) .. (1,1.4142);
  \node at (1.9,-1.6) {$\bullet$};
\end{tikzpicture}\vspace{1em}
  \caption{Kruskal diagram representing some sample constant $t$ slices. Lines of higher slope correspond to higher $t$ values. The distant observer has access to only half of the slices, \ie the parts that are found in the right hand side.}
  \label{fig:kruskal-coor-const-t-slices}
\end{figure}

As is clear from Figure~\ref{fig:kruskal-coor-const-t-slices}, as $t$ approaches $\infty$ the intersection point of stone's trajectory and $\Sigma_t$ approaches the event horizon. In the limit $t \to \infty$, the stone is seen as a point on the event horizon.


We have touched upon the general relativity side. We would now like to mention a result in quantum mechanics known as the \emph{no-cloning theorem}\ic{No-cloning theorem}. It simply means that quantum mechanics does not allow arbitrary states to be copied. The word arbitrary is important here. Otherwise there are processes that may copy some particular states. Various proofs of the no-cloning theorem are presented in Appendix~\ref{chap:proof-no-cloning}.


On the other hand, if the black hole evaporation conserves information it seems that the no-cloning theorem\ic{No-cloning theorem} can be violated, leaving us with a paradox\footnote{Readers may like to see, for example, reference \cite{Susskind2005} for discussions related to the no-cloning paradox.}. One assumes that a qubit is thrown into a black hole, and because the information that is carried by the qubit will be present in the black hole vapor, cloning of arbitrary states can be realized hence in violation of the no-cloning theorem\ic{No-cloning theorem}. All that is to be done is to consider a spacelike slice that contains both the infalling qubit and the part of vapor that carries its information. Figure~\ref{fig:cloning-paradox} illustrates such a spacelike hypersurface in the gravitational collapse scenario.

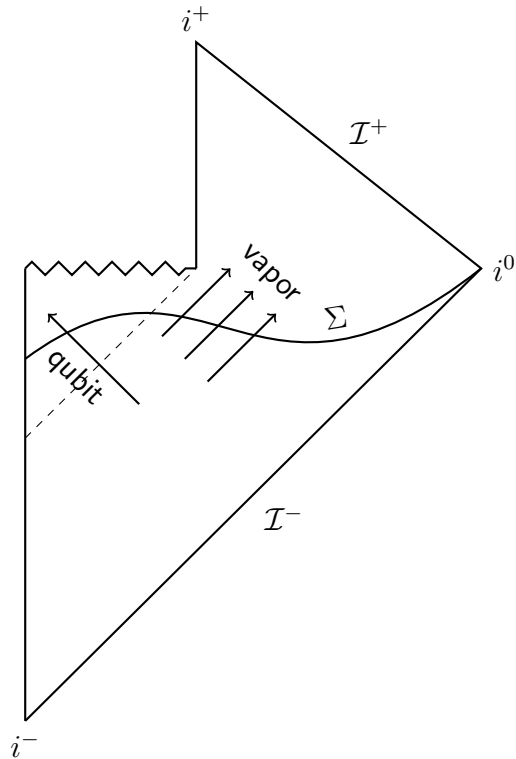
\begin{figure}
  \centering
  \begin{tikzpicture}[scale = 3, thick]
  \draw (0,0) -- (0,2);
  \draw[decorate, decoration=zigzag] (0,2) -- (0.75,2);
  \draw (0.75,2) -- (0.75,3) node [anchor=south] {$i^+$} -- node [midway,anchor=south west] {$\mathcal{I}^+$} (2,2) node [anchor=west] {$i^0$} -- node [midway,anchor=north west] {$\mathcal{I}^-$} (0,0) node [anchor=north] {$i^-$};
  \draw[thin, dashed] (0,1.25) -- (0.75,2);
  \draw [->] (0.5,1.4) -- (0.1,1.8) node [midway, below, sloped] {qubit};
  \draw [->] (0.6, 1.7) -- (0.9, 2.0);
  \draw [->] (0.7, 1.6) -- (1.0, 1.9);
  \draw [->] (0.8, 1.5) -- (1.1, 1.8);
  \draw (1.0, 1.9) node [anchor=south, rotate=-45] {vapor};
  \draw (0,1.6) .. controls (0.8,2.2) and (1,1.2) .. (2,2) node [near end, above, sloped] {$\Sigma$};
\end{tikzpicture}\vspace{1em}
  \caption{A spacelike slice $\Sigma$ that contains both the infalling qubit and its information present in some part of the black hole vapor.}
  \label{fig:cloning-paradox}
\end{figure}

Violation of the no-cloning theorem\ic{No-cloning theorem} is resolved by black hole complementarity\ic{black hole complementarity|textbf} \cite{bousso-comp-not-enough}.


The idea is simple. The outside observer has access to the qubit in the black hole vapor. The worst case scenario is that after he obtains the qubit in vapor he may decide to jump into the black hole in order to compare what is at hand with the infalling qubit. If this scenario can never be realized \cite{complementarityGedanken}, then no violation of the no-cloning theorem\ic{No-cloning theorem} is observed.

According to black hole complementarity, information about the infalling qubit is either inside or outside of the black hole: it depends on the observer \cite{complementarityGedanken}.


The \emph{stretched horizon}\ic{stretched horizon!definition} is a timelike surface just above the event horizon that contains all the degrees of freedom that can be associated with a black hole \cite{Susskind1993}. It is a phenomenological construct. According to the distant observer, infalling objects do not fall behind the event horizon, rather approach it asymptotically. Hence they fall \emph{onto} the stretched horizon. From the perspective of a distant observer, stretched horizon encodes all the necessary information that describes the black hole. Therefore from the outside observer's perspective, one need not talk about the black hole's \emph{interior}.

Stretched horizon has various properties \cite{Susskind1993}:

\begin{quote}
  If provided with an electrical multimeter, our observer will
  discover that the membrane has a surface resistivity of 377~ohms. If
  disturbed, the stretched horizon will respond like a viscous fluid,
  albeit with negative bulk viscosity.
\end{quote}

Finally, these are the three postulates of black hole complementarity\ic{black hole complementarity!postulates} \cite{Susskind1993}:

\begin{enumerate}
\item[P1] The process of formation and evaporation of a black hole, as viewed by a distant observer, can be described entirely within the context of standard quantum theory. In particular, there exists a unitary $S$-matrix which describes the evolution from infalling matter to outgoing Hawking-like radiation.
\item[P2] Outside the stretched horizon of a massive black hole, physics can be described to good approximation by a set of semiclassical field equations.
\item[P3] To a distant observer, a black hole appears to be a quantum system with discrete energy levels. The dimension of the subspace of states describing a black hole of mass $M$ is the exponential of the Bekenstein entropy $S(M)$.
\end{enumerate}

The firewall paradox\ic{firewall paradox} claims that P1, P2 and the equivalence principle are inconsistent with each other. This is the subject of the next chapter.


\chapter{The Firewall Paradox}\ic{firewall paradox|textbf}
\label{chap:the-paradox}

The epic of \emph{Leyla ile Mecnun}\ic{Leyla ile Mecnun, the epic of} describes an idea of love. However, one can find neither a \emph{Leyla} nor a \emph{Mecnun} in this world. There are approximate personalities, but never true characters. Perhaps the idea of black hole complementarity\ic{black hole complementarity} was a good dream but is untrue nonetheless: AMPS' paper claims that it is inconsistent in itself. Only the efforts of intellectuals and time will tell the answer.

We have already cited three postulates that are explicitly present in \cite{Susskind1993}. AMPS \cite{Almheiri2012} adds Einstein's equivalence principle, which is implicitly assumed in \cite{Susskind1993} as a certainty, as the fourth postulate of black hole complementarity\ic{black hole complementarity!postulates!the fourth}:

\begin{enumerate}
\item[P4] A freely falling observer experiences nothing out of the ordinary when crossing the horizon.
\end{enumerate}

The main argument of the AMPS' article is that P1, P2, P3 and P4 are inconsistent. It would be beneficial to summarize their basic argument in one paragraph.

\ic{firewall paradox!summary}The postulate P1 implies that a sufficiently old black hole is almost maximally entangled with the Hawking radiation it has radiated so far (early radiation). P4 requires that the infalling observer sees a quantum vacuum at the event horizon. However the vacuum state that an infalling observer encounters consists of entangled field modes from either side of the horizon: inside and outside the black hole. Because any Hawking quantum that will be radiated by an old black hole will be entangled with the early radiation, it cannot be entangled with the modes behind the horizon in order to result in a vacuum state. This is contrary to the subaddivity property of entropy. Then by using P2, the Hawking quantum radiated by the old black hole is evolved backwards in time, through which it is blue shifted by huge amounts. Infalling observer encounters this high energy particle at the horizon: Horizon is a special place, equivalence principle is broken. Mode by mode, it is seen that there is a firewall located at the horizon.

Now, we shall go over the article.

A black hole that has formed through collapse of a pure state is considered \cite{Almheiri2012}. Because of P1, the evaporation process is unitary: the black hole vapor (the one that would be produced after complete evaporation) is in the following pure state \cite{Almheiri2012}:

\begin{equation}
  \label{eq:amps-psi}
  \ket \Psi = \sum_i \ket{\psi_i}_E \otimes \ket{i}_L .
\end{equation}

Here the subindices $E$ and $L$ are used to denote the states that are in early or late part of the radiation. The distinction between what constitutes early or late is based on the \emph{Page time}\ic{Page time!definition}: the time that that black hole will have lost half of its entropy \cite{Almheiri2012}. Accordingly, when a black hole's age exceeds the Page time, it is called \emph{old}; otherwise \emph{young}\ic{black hole!young or old} \cite{Almheiri2012}.

In equation~(\ref{eq:amps-psi}) $\{\ket{i}_L\}$ is some orthonormal basis for the late part of the radiation. The kets $\ket{\psi_i}_E$ have no special property apart from being almost orthonormal (or almost orthogonal with almost equal norms) because $\ket \Psi$ is almost maximally entangled.


In (\ref{eq:amps-psi}) it is implicit that $\ket \Psi$ is a state that contains entanglement between early and late parts of the radiation, otherwise it would not be a sum. The entangled nature of the quantum state is expected by virtue of Don Page's\ip{Page, Don N.} calculation \cite{PageBHInfo} that indicates late time radiation is almost maximally entangled with the early radiation. Although he is very well aware that the amount of information to be found in black hole vapor may not be stated without a full theory of quantum gravity, he claims that this amount should be close to the amount found on average \cite{PageBHInfo}. Lowe\ip{Lowe, David A.} and Thorlacius\ip{Thorlacius, Larus} comments on the applicability of Page's result \cite{LoweThorBHInfoProb}:

 \begin{quote}
   If we assume the Hawking radiation and black hole are in a random pure state, as should be true if interactions effectively thermalize the black hole degrees of freedom, then we may apply the results of Page \ldots to compute the typical amount of information contained in the Hawking radiation.
 \end{quote}

Sabine Hossenfelder\ip{Hossenfelder, Sabine} \cite{hossenfelderCommentOnFirewall} raises the question that this may not be so, but agrees that it is a natural assumption:

\begin{quote}
  It basically says that all the partners of the early Hawking
  particles, those that would normally fall into the singularity and
  get lost, should come out as late as possible because that is when
  new quantum gravitational effects most plausibly occur.
\end{quote}


Because the energy carried away by the Hawking radiation is finite, authors use this observation to assert that the Hilbert space that describes the Hawking radiation should be finite dimensional. This assertion in the end allows the use of before mentioned calculation of Page.

So far, we have given reasons why an old black hole is almost maximally entangled with the early Hawking radiation. Next, we consider an infalling observer who ``knows the initial state of the black hole and also the black hole S-matrix'' \cite{Almheiri2012}. As he crosses the horizon, by equivalence principle, he sees a vacuum state. Please see Figure~\ref{fig:fw0}. This point will be important in the following discussion.

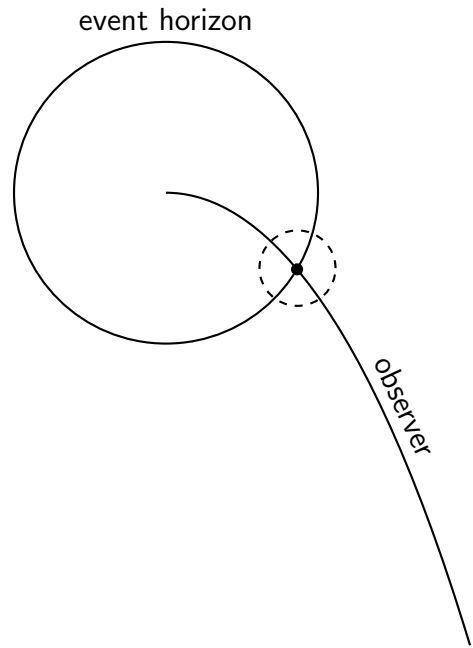
\begin{figure}
  \centering
  \begin{tikzpicture}[scale = 1, thick]
  \draw[name path=BH] (0,0) circle (2);
  \draw (0,2) node [anchor=south] {event horizon};
  \draw[dashed] (1.73,-1) circle (0.5);
  \draw[name path=OBS] (4,-6) .. controls (2.5,-1) and (1,0) .. (0,0) node [near start, above, sloped] {observer};
  \fill[black, name intersections={of=BH and OBS}]
    (intersection-1) circle (0.08);
\end{tikzpicture} \vspace{1em}
  \caption{Trajectory of the infalling observer. He must observe quantum vacuum state whilst passing through the event horizon by virtue of Einstein's equivalence principle.}
  \label{fig:fw0}
\end{figure}

On the other hand, a Hawking mode of width of order $r_s$, where $r_s$ is the Schwarzschild radius of the black hole, is considered. This is essentially a wave packet; one should think of choosing an out-basis constructed out of such wave packets\footnote{They are localized only in the radial direction: they have definite angular momentum.}. It has an associated annihilation operator, named $b$ by AMPS \cite{Almheiri2012}. Far from the black hole, it is a well defined particle.

Since the early and late radiations are almost maximally entangled, by measuring suitable observables in early radiation one can infer results of corresponding measurements that would have been done on the late radiation\footnote{Please see introductory chapter on quantum entanglement.}. If there was exact maximal entanglement, the correspondance would be perfect. However when there is almost maximal entanglement, the relative error is quite low for black holes that are well above the Planck mass\ic{Planck mass}. AMPS calculate the average relative error as $\bar{\varepsilon} = L/E$ where $L,E$ are degrees of freedom present in late and early radiation \cite{Almheiri2012}. This number is approximately equal to $\exp(-2\pi M^2)$ which is extremely close to zero, where $M$ is the initial mass of the black hole. For a solar mass black hole, we calculated before that $M = 9.14 \times 10^{37}$ in Planck units. This should give an idea of how vanishingly small this error is for a solar mass black hole.

Then the expectation value of the number operator, $b^\dagger b$, for the before mentioned mode can be found by making measurements in the early radiation. ``In other words, an observer making measurements on the early radiation can know the number of photons that will be present in a given mode of the late radiation'' \cite{Almheiri2012}. This information allows us to infer that such a Hawking quantum \emph{will have been emitted} after some time.

Because the mode~$b$ will have been emitted when the black hole is old, it must be almost maximally entangled with the early radiation. Therefore, when it is traced back towards the horizon it will not be in an entangled state with an interior mode to yield the vacuum state\footnote{Appendix~\ref{chap:entangledvacuum} illustrates why and how the vacuum is entangled.} that an infalling observer is expected to see. Formation of this entanglement is disallowed by the ``monogamy of entanglement''\ic{entanglement!monogamy of}. Equivalence principle is violated where it is least expected.

We will make clear what is meant by the monogamy of entanglement through strong subadditivity of entropy when we discuss high angular momentum modes. More mathematically, the discussion of low angular momentum modes is as follows.


Since the black hole vapor is dominated by low angular momentum modes \cite{Page1976}, AMPS use simplified gray body factors where they are equal to unity for low angular momentum states and are equal to zero for the rest. This makes it easy to directly express $b$ in terms of operators that an infalling observer would use near the horizon. In their own notation \cite{Almheiri2012}:

\begin{equation}
  b = \int_0^\infty d\omega \; (B(\omega) a_\omega + C(\omega) a_\omega^\dagger),
\end{equation}

here $\{a_\omega\}_\omega$ are the operators that an infalling observer would use. Of course $C(\omega)$ does not vanish, because $a$-vacuum and $b$-vacuum differ, which is in the end related to the difference of the time coordinate appropriate for stationary and infalling observers. Although an infalling observer may detect particles, their occurrence is suppressed exponentially in the energy of particles \cite{BirrellDavies}.

If the gray body factors were not chosen in this simplified form, AMPS would need to take into account the angular momentum barrier which would introduce more complication since there would be scattering as the mode is evolved backwards in time. As we will mention in the discussion of black hole mining experiments, their argument is fully general: for all angular modes, low and high.

All in all, the vacuum around the horizon that an infalling observer would see is not the vacuum that $b$ annihilates: $b$~mode will be detected by the infalling observer. When they are detected, they will have been \emph{highly} blue shifted. Suppose it has energy $E$ around $\mc I^+$. Its 4-momentum is \cite{carroll-gr-book}: 

\begin{equation}
  k^\mu = \left( E \left(1 - \frac{2M}{r} \right)^{-1}, \left[ E^2 - \frac{L}{r^2} \left(1 - \frac{2M}{r} \right) \right]^{1/2}, \pm \frac{L}{r^2},0\right).
\end{equation}

We orient the coordinate system such that $d\theta/d\lambda = 0$ where $\lambda$ is an affine parameter. On the other hand, the 4-velocity of a radially ingoing timelike geodesic that has zero speed near $i^-$ is given by $u^\mu = ((1-2M/r)^{-1},-(2M/r)^{1/2},0,0)$ \cite{hartleGR}. Therefore, energy of the photon detected in the rest frame of the infalling observer, for $r \approx 2M$, is found to be:

\begin{equation}
  -u\cdot k \approx 2 E \left(1 - \frac{2M}{r} \right)^{-1}. \label{eq:53}
\end{equation}

The stretched horizon for a Schwarzschild black hole has area that is one Planck unit larger than the area of the event horizon \cite{complementarityGedanken}.  Therefore $r= 2M + 1$ (remember we use Planck units\ic{Planck units}) is rather outside the stretched horizon, however it will suffice\footnote{If one calculates $r$ for which the difference between the areas of stretched horizon and event horizon is 1, it turns to be about $r=2M + 1/16\pi M$. Then the blue shift factor in (\ref{eq:53}) will be proportional to $M^2$ and blue shifted Hawking particle would have energy of the order $M$. Because this is well above the Planck energy, one may argue that one may evolve photons backwards in time until Planckian energies are obtained. However the energy observed by a \emph{stationary} observer located at fixed $r$ is $E (1-2M/2)^{-1/2}$. For $r = 2M + 1/16\pi M$ and $E = 1/8\pi M$ this is approximately $(2\pi)^{-1/2}$. Therefore, photons can be evolved backwards in time until they reach the stretched horizon.} to show high energy of detected photons. The factor $(1 - 2M/r)^{-1}$ to a good approximation is $2M$. Therefore the energy of the photon is blue shifted by order $4M$ which for a solar mass is about $4 \times 10^{38}$. If the outgoing photon has a typical energy, \ie of the order of Hawking temperature $T = 1/8\pi M$, then observed energy of photon is $1/2\pi$, in Planck units. An infalling observer encounters Planckian radiation\ic{firewall paradox!Planckian radiation} near the horizon.


The discussion up to this point concerns only the low angular momentum modes. Modes of higher angular momentum are mostly reflected back towards the black hole because of the centrifugal potential that increases as $\sim l^2$ as it is present in (\ref{eq:49}) and (\ref{eq:50}). However it is well known \cite{UnruhWaldPRD_25_942, UnruhWaldPRD_27_2271, UnruhWaldEssay} that these modes can be detected by lowering a particle detector behind the centrifugal barrier.

If one solves the geodesic equation for null rays in Schwarzschild geometry, for details one may refer to a classical textbook such as Carroll's \cite{carroll-gr-book}, one sees that the maximum of the centrifugal barrier lies at $r=3M$ although its height is proportional to $l^2$. So there is no barrier for an $l=0$ mode. On the other hand, here we are concerned with light \emph{waves}, not light \emph{rays}. In equation (\ref{eq:50}) the maximum of potential function may be calculated however one should consider that $l$ is now quantized: $l \in \mathbb N$. For $l=0$ it is at $r=8M/3$ whereas for $l>0$ it resides at $r=[(9l^4 + 18l^3 + 23l^2 + 14l + 9)^{1/2}-3]M/2l(l+1) + 3M/2$ which is always greater than $17M/6$ and asymptotically approaches $r=3M$. It increases monotonically. Our main point in this paragraph has been to show the place where a lowered particle detector must pass in order to detect high angular momentum modes efficiently. As a side note, this region where high angular momentum modes are abundant, a shell of proper length of the order $r_s$ from the horizon, is sometimes \cite{harlow-pirsa} called black hole's ``atmosphere''\ic{black hole!atmosphere}\footnote{Calling this region as the atmosphere of a black hole seems quite apt. It is argued \cite{bousso-comp-not-enough} that ``the minable modes in the zone \ldots \emph{must} be considered as a subsystem of the black hole.'' (emphasis in the original). Hence a black hole's atmosphere, like that of any other stellar object, is a part of it.} and sometimes \cite{bousso-comp-not-enough} as ``the zone.''\ic{black hole!the zone@`the zone'}

As for modes of high angular momentum, AMPS argue that the situtation is the same. When a particle is sensed by a lowered particle detector, this particle should have been entangled with the early radiation \emph{before} it was detected \cite{Almheiri2012}. Entanglement of these modes with the early radiation arises ``since they can be mined and thus form a subsystem of the late radiation'' \cite{bousso-comp-not-enough}.  Hence the earlier argument fully applies. AMPS then conclude that ``the infalling observer encounters a Planck density of Planck scale radiation and burns up'' \cite{Almheiri2012}.


A way of expressing the paradox in direct mathematical terms is to show that it causes a violation of the \emph{strong subadditivity}\ic{entropy!strong subadditivity} property of entropy. A three-partite Hilbert space is considered: $\mc H = \mc H_A \otimes \mc H_B \otimes \mc H_C$ and a density matrix defined on $\mc H$. As usual, reduced density matrices in various subspaces of $\mc H$ are obtained through taking partial traces. We shall denote by, for example, $S_{AB}$ the entropy associated with a reduced density matrix on $\mc H_A \otimes \mc H_B$. The strong subadditivity of entropy states that following holds \cite{ssaLiebRuskai}:

\begin{equation}
  S_{AB} + S_{BC} \geq S_{B} + S_{ABC}. \label{eq:ssa}
\end{equation}

Back to the Hawking radiation, AMPS' argument concerning the violation of strong subadditivty is as follows. One observes three Hilbert spaces, in their notation: A) early Hawking radiation, B) a Hawking particle emitted when the black hole is old, C) its interior partner.

B and C must be in a pure entangled state in order to give rise to vacuum observed by an infalling observer: $S_{BC} = 0$. Therefore $S_{ABC} = S_A$. When these results are applied to (\ref{eq:ssa}), we obtain $S_{AB} \geq S_A + S_B$. On the other hand, an old black hole is the one that begins purifying early radiation\footnote{Sometimes, the time when the entropy of the so far emitted Hawking radiation begins to decrease is defined to be the Page time\ic{Page time!alternate definition}.}, which means that the entropy of A and B is less than the entropy of A: $S_{AB} < S_A$. As a result one obtains the following:\ic{entropy!strong subadditivity!violation in firewall paradox}

\begin{equation}
  0 > S_B\quad\text{(contradiction)} \label{eq:48}
\end{equation}

Remember that, in the beginning of this section we menioned the idea of constructing observables acting on the early radiation that will give results of corresponding observations that will have been done on the late radiation. For that purpose, AMPS used the randomness of the Hawking radiation to calculate an average relative error, $\bar{\varepsilon} \approx \exp(-2\pi M^2)$, which is extremely close to zero for semiclassical black holes. Here, as regards the violation of strong subadditivity, AMPS observes \cite{Almheiri2012} that:

\begin{quote}
   [the randomness of the Hawking radiation] is not needed; it is sufficient that the entropy of the black hole be decreasing. From another point of view, one need not be able to predict the state with perfect fidelity: rather, any information about the state of the $b$ mode precludes the state being annihilated by $a$.
\end{quote}

Before the AMPS' article \cite{Almheiri2012} it was supposed that an infalling observer would see a vacuum while passing through the event horizon and the Hawking radiation would be pure as required by black hole complementarity\ic{black hole complementarity}. What AMPS did, is to put emphasis on the consequences of unitary black hole evaporation from the perspective of \emph{infalling} observers. They used the almost maximal entanglement of small subsystems with the rest of the system \cite{pageAvgEntropy, PageBHInfo} quite elegantly and have presented a neat paradox.

This is the firewall paradox.

\begin{figure}[b]
  \centering
  \includegraphics[width=5cm]{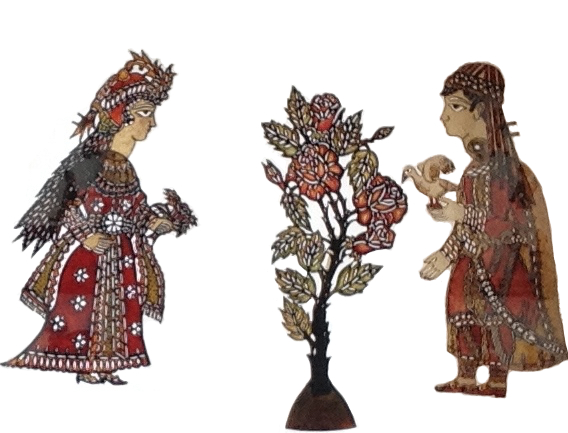}
  \caption{Figures of \emph{Leyla} and \emph{Mecnun} in Turkish shadow theatre. Adapted from a photograph taken in the Museum of Karagöz, Bursa.}
  \label{fig:leyla-and-mecnun}
\end{figure}

\chapter{Multitude of Approaches}
\label{chap:mult-approaches}

In this chapter, we include various approaches that attempt to resolve the firewall paradox.

\section{Harlow-Hayden conjecture}\ic{Harlow-Hayden conjecture|textbf}
\label{sec:hh-conjecture}

This conjecture is about the non-observance of contradiction between the postulates of black hole complementarity from an operational point of view. It is an attempt to save the original complementarity.

Daniel Harlow\ip{Harlow, Daniel} and Patrick Hayden\ip{Hayden, Patrick} wrote an article \cite{harlow-hayden} from the perspective of quantum computation conjecturing that ``the decoding time for Hawking radiation will in general be exponential in the entropy of the remaining black hole'' \cite{harlow-hayden}. This has become to known as \emph{the Harlow-Hayden conjecture}\ic{Harlow-Hayden conjecture}.

The idea is to show the impossibility of a distant observer doing measurements to verify early-radiation/late-Hawking-quantum entanglement and jumping into the black hole to test the interior-mode/late-Hawking-quantum entanglement. If this is operationally untestable, then there is no paradox from the perspective of complementarity.

The importance of distillating early radiation comes into picture because the hypothetical observer (who aims to show that complementarity is inconsistent) should first verify the entanglement between late and early radiation. For that purpose, it is important that the part of early radiation that is entangled with a late Hawking quantum must be distilled into an easily accessible qubit. Otherwise, there is a gigantic cloud of radiation and the chances to test the entanglement are bleak without further quantum computing on the early radiation. According to the Harlow-Hayden conjecture, quantum computation that is necessary to perform this job requires a time of the order of $e^{M^2}$ where $M$ is the black hole mass \cite{harlow-hayden}.

On the other hand the black hole would already have evaporated after $t \sim M^3$, however, it takes $t \sim e^{M^2}$ amount of time to isolate the entangled subsystem in the black hole vapor \cite{harlow-hayden}. Therefore, when the hypothetical observer confirms the entanglement with the black hole vapor, the black hole would have already evaporated. Hence it is impossible to jump into the black hole and test the entanglement between the two sides of the horizon.

Evidence for the correctness of this conjecture seems strong \cite{HarlowPrivate} and hence if no single observer can see a violation of black hole complementarity: the firewall paradox\ic{firewall paradox} is resolved.

\section{Strong Complementarity Principle}\ic{strong complementarity|textbf}
\label{sec:strong-complementarity}

This approach is a stronger version of complementarity, in essence expressing that every observer has his own description of nature which must be in agreement with those of others when the results of observations can be communicated \cite{bousso-comp-not-enough}. It would be beneficial to remember that a strong condition is more restrictive than a weak condition. It is in this sense that this approach is called \emph{strong} complementarity principle. The names \emph{strong complementarity} and \emph{causal patch complementarity}\ic{causal patch complementarity|see {strong complementarity}} are used interchangeably.

The Harlow-Hayden conjecture\ic{Harlow-Hayden conjecture} is sometimes \cite{Almheiri2013} cited among strong complementarity approaches, however it should better be viewed as an idea to save the original complementarity. Because the idea there is that no observer can see a violation of complementarity, hence there appears no problem in supposing that an infalling observer finds quantum vacuum while passing through the event horizon. In particular, even though Harlow\ip{Harlow, Daniel} and Hayden\ip{Hayden, Patrick} discuss strong complementarity at some point in their paper \cite{harlow-hayden}, the conjecture itself is independent of this discussion.

An interesting work \cite{ilgin-yang-inside-old} in this approach is that of İrfan Ilgın\ip{Ilgın, İrfan} and I-Sheng Yang\ip{Yang, I-Sheng}. They focus their attention on the causal patch of an infalling observer and find that there is no spacelike hypersurface on which both the interior mode and early Hawking radiation are both low energy. In the paper, the firewall paradox is postponed to Planck scale physics and possible modifications of low energy physics are found unnecessary. In particular ``through unknown UV physics'' \cite{ilgin-yang-inside-old} at the boundary of the infalling causal patch, need for firewalls is eliminated. On the other hand Ilgın and Yang remark that although distillation of early radiation by an infalling observer may have its own problems, if it is possible then there really is a firewall paradox \cite{ilgin-yang-inside-old}.

\section{ER = EPR}\ic{ER = EPR|textbf}
\label{sec:er-epr}


The abbreviations `ER' and `EPR' stand for Einstein-Rosen bridges (worm holes) and Einstein-Podolsky-Rosen pairs (entangled quantum states) respectively. Juan Maldacena\ip{Maldacena, Juan} and Leonard Susskind\ip{Susskind, Leonard} ``take the radical position that in a theory of quantum gravity they are inseparably linked, even for systems consisting of no more than a pair of entangled particles'' \cite{ErEpr}.

In order to give a justification to their conjecture, they mention \cite{ErEpr} similarities between ER and EPR: 1) One cannot use one of these to make superluminal communication, 2) Entanglement between two systems cannot be created by LOCC\footnote{LOCC: Local Operations and Classical Communication.\ic{LOCC}}, it should be done either causing the systems to interact or merging them with parts of already entangled pairs. In case of separate black holes that are not connected through an EPR bridge, one may only create such bridges by merging the black holes with ones that have already a worm hole between them.

Maldacena and Susskind consider an eternal AdS black hole\ic{black hole!in AdS space} which is described by the following quantum state:

\begin{equation}
  \ket \Psi = \sum_n e^{-\beta E_n / 2} \ket{n} \otimes \ket{n}  ,\label{eq:51}
\end{equation}

where $\ket n$ is the microstate of each black hole with energy $E_n$ \cite{ErEpr}. Here are two black holes that are connected through an Einstein-Rosen bridge. Black holes are causally disconnected.

Their ingenious idea is to interpret the second black hole as the early radiation in the AMPS argument. The black hole is entangled with a second system, and the horizon is smooth: there is no firewall. This is one of their arguments against firewalls. However, as the careful reader has already noticed, the reduced density matrix obtained from the state $\ket \Psi$ in (\ref{eq:51}), $\rho = \sum_n e^{-\beta E_n} \ket n \bra n$, is not maximally entangled. They emphasise this point \emph{en passe} in a footnote \cite{ErEpr}.

In the case of eternal Schwarzschild solution, their reasoning is similar. However this time, they put the two casually disconnected black holes in casual contact, as two one sided Schwarzschild black holes that are very far away from each other. Of course this solution is only approximate. However since the black holes are assumed to be very far away from each other, this is not a crude approximation. Black holes should be in a specific entangled state for an Einstein-Rosen bridge between them to exist, and this is ``a third interpretation of the eternal Schwarzschild black hole'' \cite{ErEpr}.

On the contrary, whether firewalls exist or not depends on what is done with the radiation \cite{ErEpr}. If the black hole vapor is collected and then condensed into another black hole, by a quantum computer that is acting on the new black hole the system can be arranged such that there is a smooth ER bridge between the black holes \cite{ErEpr}. However, by an appropriate quantum computation action on the second black hole a firewall can be created behind the event horizon of the first horizon \cite{ErEpr}.

When nothing is done on the radiation, first guess would be that each one of Hawking quanta is connected to the black hole through ER bridges of quantum nature \cite{ErEpr}. These bridges better be of unknown quantum nature, because they are expected to connect particles to black hole and the concept of particle in QFT in curved spacetimes is ambiguous, in general.

The question of the smoothness of black hole horizons is related to the connection structure of these quantum ER bridges to the black hole and apart from equivalence principle --which is already challenged by the firewall paradox\ic{firewall paradox}-- there is currently no independent argument for its smoothness \cite{ErEpr}.

As far as the AMPS argument is concerned ``[Maldacena and Susskind] have given enough reasons for not believing the AMPS argument that there \emph{had} to be firewalls'' \cite{ErEpr} (emphasis in the original).


For further progress on this approach, readers may like to read \cite{susskind-new-old-bh, susskind-butterfly, Susskind2014, susskind-addendum, jensen-karch, sonner}. As for ideas about the relation between entanglement and geometry, one may see Brian Swingle's\ip{Swingle, Brian} article \cite{swingle} and Mark van Raamsdonk's\ip{van Raamsdonk, Mark} essay \cite{RammsdonkEssay}.

\section{Fuzzball Complementarity}\ic{fuzzball complementarity|textbf}
\label{sec:fuzzball-comp}


Under certain assumptions, conditions for unitary black hole evaporation are found \cite{mathur-info}: one of the following items must be true:

\begin{itemize}
\item Physics needs to be modified to yield unitary evaporation, \eg inclusion of nonlocal interactions.
\item Black holes have hair, \ie degrees of freedom located around the horizon.
\item Assumptions used to reach these options are wrong or insufficient.
\end{itemize}

Otherwise, information in the universe would not be preserved and black hole evaporation would not be unitary as expected through quantum mechanics. On the other hand, it would be useful to remember as noted by \cite{mathur-info} that unitarity does not imply information preservation.

A \emph{fuzzball}\ic{fuzzball!definition} is a string theoretic quantum state of black hole that exhibits structure almost at the event horizon and has no interior geometry \cite{mathur-info}. It is the second view mentioned above towards the resolution of the information paradox. In Mathur and Turton's words \cite{MathurTurtonFlaw}: ``Hawking radiation arises from the surface of the fuzzball just like radiation from a piece of coal. This resolves the information paradox.''


Introducing fuzzballs, we would like to mention the idea of \emph{fuzzball complementarity}\ic{fuzzball complementarity}. Its explicit postulates can be found in reference \cite{MathurSSaboutBH} which we summarize as below:

\begin{itemize}
\item [F1] Black hole microstates have no interior. They end compactly before the horizon is reached. An actual state of a black hole is a superposition of these microstates.
\item [F2] The fields around the horizon are not in vacuum state.
\item [F3] Processes involving infalling quanta of energy $E \gg T$ can be described to a good accuracy by general relativistic black hole solutions that have interior geometry.
\end{itemize}

When a particle of energy that is well higher than the temperature ($E \gg T$) hits a fuzzball, it excites modes on the fuzzball surface and most of these modes are not entangled with the early radiation as a black hole older than the Page time\ic{Page time} would have been \cite{MathurTurtonFlaw}. Moreover the dynamics of these excitations have a complementary description from the perspective of an infalling observer falling in a smooth geometry \cite{MathurTurtonFlaw}. Therefore the firewall paradox is resolved in the fuzzball picture.

As a disclaimer it would be appropriate to point out that fuzzball complementarity, however, does not claim that there is a complementary description for quanta that have energy of the order of black hole temperature \cite{MathurTurtonFlaw}.

On the other hand, one already observes that fuzzballs are quite different from the physics we are used to. For example, there is no black hole interior and the spacetime region around the horizon is not in vacuum \cite{MathurTurtonFlaw}. Black hole complementarity we presented in Chapter~\ref{chap:intro} is called ``traditional complementarity'' \cite{MathurTurtonFlaw} in presence of fuzzball complementarity. With this distinction in mind, it is expressed in \cite{MathurSSaboutBH} that the firewall paradox is put forward to ``exclude traditional complementarity.''

\section{Holographic Principle and Black Hole Interior}
\label{sec:paparaju}


The holographic principle\ic{holographic principle} states that the number of degrees of freedom in a region of spacetime scales with the area of its boundary \cite{BoussoHolPrinc, tHooftDimReduction}. There is strong support for the holographic principle, however it is still a conjecture \cite{BoussoHolPrinc}. Reader may like to see \cite{BoussoHolPrinc} for a review on holographic principle\ic{holographic principle}.

On the other hand, AdS/CFT correspondence\ic{AdS/CFT correspondence} conjectures that what happens in the anti-de~Sitter spacetime\ic{anti-de~Sitter spacetime} is related to what happens in conformal field theory in its boundary. It is seen as a specific realization of the holographic principle\ic{holographic principle}. For an introduction to AdS/CFT correspondence references  \cite{RamalloAdSCFT, NastaseAdSCFT, PetersenAdSCFT} may be of interest.


\ic{Papadodimas-Raju proposal|textbf}In joint works of Kyriakos Papadodimas\ip{Papadodimas, Kyriakos} and Suvrat Raju\ip{Raju, Suvrat} \cite{PapaRajuInFall, PapaRajuPRL, PapaRajuPRD} it is argued that the interior of a black hole in AdS can be described by CFT in the spacetime boundary.

Their construction of the black hole interior depends on the state of CFT at the boundary, and this state dependence is what allows black holes to evaporate unitarily while maintaining smooth horizons \cite{PapaRajuPRL}. They especially dismiss the fuzzball conjecture, and the way they evade the conclusions of \cite{mathur-info} is via the claim that the modes in the black hole interior are not independent from those in early Hawking radiation \cite{PapaRajuPRL}. For more on this proposal, readers may like to see Daniel Harlow's\ip{Harlow, Daniel} paper \cite{HarlowOnPapaRaju}.

\section{Extreme Cosmic Censorship Conjecture}\ic{extreme cosmic censorship conjecture|textbf}
\label{sec:ext-cos-cens}

This viewpoint is put forward by Don Page\ip{Page, Don N.} \cite{page-ecc}. He cites two reasons that cause the firewall paradox. The first one is the assumption of validity of semiclassical approximation: ``Effective field theory is local, whereas the constraint equations of gravity are nonlocal, so the assumption of effective field theory is almost certainly incorrect'' \cite{page-ecc}. The second one is claimed to be the overcounting of black hole microstates \cite{page-ecc}. Page expresses that in arguments behind the firewall proposal, it is implicitly assumed that one counts states that will exhibit singular structure when they are evolved backwards in time \cite{page-ecc}. In order to exclude these states, he proposes the \emph{extreme cosmic censorship} conjecture \cite{page-ecc}:

\begin{quote}
  The universe is entirely nonsingular (except for singularities deep inside black holes and/or white holes which do not persist to the infinite future or past, with these singularities coming near the surface only when the holes have masses near the Planck mass that normally happens only close to the ends and/or beginnings of their lifetimes).
\end{quote}

Page considers the following categories of quantum states \cite{page-ecc}:

\begin{itemize}
\item \textbf{Unconstrained kinematic states:} These are most general states one considers in a theory. They do not have to obey conditions such as gauge conditions. They do not have to be realized.

\item \textbf{Constrained physical states:} These are most general states that satisfy various constraint equations in the theory.

\item \textbf{Nonsingular realistic states:} These are constrained physical states that obey the extreme cosmic censorship criterion.

\item \textbf{The actual state:} This is the realized state of the universe.

\end{itemize}

According to Page, the actual state of the universe is a nonsingular realistic state. So extreme cosmic censorship is valid. Nonsingular realistic states do not have any singular structure, such as a firewall, at the event horizon. This is true by definition. Therefore in all such states, an infalling observer sees a quantum vacuum at the horizon: the degrees of freedom inside and outside the black hole are entangled.

Page's important observation is that this entanglement is illusory: the main reason is that, according to the extreme cosmic censorship, only realizable states of the universe are nonsingular realistic, all of which contain entanglement between inside and outside of the black hole in a form that would result in vacuum state for an infalling observer.

Of course, if the semiclassical theory outside the stretched horizon is valid, then an outgoing mode must be entangled with early radiation if black hole evaporation is a unitary process. Page is well aware of this fact, and he proposes that effective field theory should be abandoned \cite{page-ecc}. In order to evade the firewall paradox, in this respect, it is thought that there must be a process that will transfer the entanglement between late Hawking quanta with interior modes to entanglement between late Hawking quanta with early radiation \cite{page-ecc}. However the existence of such a mechanism is not shown yet.


\section{Icezones: the anti-thesis of Firewalls}\ic{icezone|textbf}

The article \cite{Stojkovic2013} written by John Hutchinson\ip{Hutchinson, John} and Dejan Stojkovic\ip{Stojkovic, Dejan} explores two ideas: 1) low energy physics may be enough to preserve unitartiy black hole evaporation while not violating the principle of equivalence, 2) observer dependence of entanglement makes it harder to obtain a practical paradox.

We begin with the second one. In the Section~\ref{sec:hh-conjecture} on the Harlow-Hayden conjecture\ic{Harlow-Hayden conjecture}, it was discussed that if no observer can detect a violation of the monogamy of entanglement then the black hole complementarity will be secure. Reference \cite{Stojkovic2013} mentions the role of observer in measuring entanglement. In particular ``[t]he entanglement of two nearby modes separated by an event horizon is not just dependent on proximity but also on the acceleration of the observer'' \cite{Stojkovic2013}. The source of entanglement degradation is pointed out \cite{Stojkovic2013} as the Unruh radiation\ic{Unruh radiation}. This argument only makes it harder for an infalling observer to observe a violation in the laws of quantum mechanics. Therefore it is in support of impracticality arguments, such as the Harlow-Hayden conjecture, against the firewalls.

The main argument is \emph{icezones}. Reference \cite{Stojkovic2013} accepts, in accordance with the AMPS \cite{Almheiri2012}, that Rindler vacuum is no longer a good approximation around the horizon for old black holes. However the main difference is that there are only low energy quanta present around the black hole \cite{Stojkovic2013} hence the name icezones.

The aim is to modify the mechanism underlying the production of Hawking quanta. In the usual picture, outgoing modes are created with ingoing modes: in pairs. In the icezone picture, this is modified and alongside particle pairs are created low energy quanta \cite{Stojkovic2013}. It is then expected that these low energy particle will collide with the already emitted radiation and arrange the necessary order of entanglements so that a particle emitted by an old black hole is entangled both with the interior and the early radiation \cite{Stojkovic2013}.

From another perspective, an icezone performs entanglement transfer: when a particle is emitted by an old black hole it is almost maximally entangled with the interior mode, however as it moves further away from the black hole it becomes almost maximally entangled with the rest of the early radiation \cite{Stojkovic2013}. A similar mechanism of entanglement transfer is proposed conceptually by Page in \cite{page-ecc} though he commented that this mechanism ``would almost certainly involve violations of local effective field theory.'' Icezones do not violate locality \cite{DejanPrivate} however without an explicit theory it would be hard to tell if icezones can be realized.

On the other hand the equivalence principle remains intact in presence of icezones, which is basically because of the low energy of the icezone quanta \cite{DejanPrivate}. So it is not violated, as it has not been by the presence of quanta of energy $\mc O(1/M)$ near the horizon of black holes due to effects of curved spacetime.

Samir Mathur\ip{Mathur, Samir} gave the conditions for unitary black hole evaporation in his paper \cite{mathur-info} which we cited before in Section~\ref{sec:fuzzball-comp}: either black hole horizon is not in vacuum state and/or there are $\mc O(1)$ modifications to the Hawking radiation. Although it is hard to tell, in the absence of an explicit icezone theory, the general attitude of the paper \cite{Stojkovic2013} indicates that black hole horizon is not at the vacuum state and there may be non-perturbative effects in action.

All in all, if an explicit mechanism could be provided for icezones in the future, it has the potential of being the most favored approach.


\section{Shape Dynamics}\ic{shape dynamics|textbf}

Perhaps if the classical theory of gravitation, which currently is general relativity, conjectured that event horizons were \emph{special places}, firewalls would not be disliked as much. There are various alternatives to general relativity that exist in the literature. For example, $f(R)$ theories \cite{DeFelice2010} and massive gravity \cite{DeRham2014} would be only two examples from a rather vast list. Shape dynamics is one such theory, however, it is at the moment at its initial stages of developement. It is a theory of gravitation that is ``a novel formulation of Einstein's equations in which refoliation invariance is replaced with local spatial conformal (Weyl) invariance'' \cite{KoslowskiGomesFreq}. For more about shape dynamics and its connection to general relativity one may find references \cite{KoslowskiGomesFreq, KoslowskiGomesLink, SD-Extra-1, SD-Extra-2, SD-Extra-3, SD-Extra-4} useful.

Henrique Gomes\ip{Gomes, Henrique} and Gabriel Herczeg\ip{Herczeg, Gabriel} in their article \cite{GomesHerczeg} noted that event horizons of shape dynamic black holes\ic{black hole!shape dynamic} can be detected by local measurements, whereas these black holes are indistinguishable from general relativistic ones for distant observers. In shape dynamics the equivalence principle is not fundamental but rather an emergent property, and it is not exhibited at the event horizon \cite{GomesHerczeg}. Shape dynamic black holes should be further investigated as regards the experience of infalling observers, especially by taking into account the quantum fields. The question of how different the experience of an infalling observer from that in general relativity should be answered.

A remark would however be that because the interior of this type of black holes is another mirror universe \cite{GomesHerczeg}, they seem to violate postulate P3 of black hole complementarity which states that the number of interior black hole states is $\exp(A/4)$ where $A$ is the area of the horizon.


\section{The Backreaction of Hawking Radiation}\ic{Hawking radiation!backreaction|textbf}

The back reaction of Hawking quanta onto the black hole has been of concern in discussions regarding quantum evaporation of black holes. The work that Laura Mersini\ip{Mersini, Laura} published \cite{Mersini-Houghton2014-I} deals exactly with this problem. The results are intriguing.

Under the assumptions of pressureless dust, homogeneity and spherical symmetry gravitational collapse of a star is considered including the backreaction of the Hawking radiation \cite{Mersini-Houghton2014-I} where the initial vacuum state is the Hartle-Hawking vacuum\ic{Hartle-Hawking vacuum}.

The key idea that makes the calculation of backreaction possible is
``[k]nowing that particle creation occurs during the collapse stage,
means that we can include the backreaction of Hawking radiation onto
the collapse dynamics of the star to find out if a singularity forms
at the end of the collapse" \cite{Mersini-Houghton2014-I}.

What is found in \cite{Mersini-Houghton2014-I} is quite interesting: no black hole ever forms, and the star begins to expand after reaching a certain size which is outside its Schwarzschild radius. This result is obtained mainly because \cite{Mersini-Houghton2014-I} takes into acocunt the negative energy Hawking radiation that travels inwards and this violates the conditions of the Penrose-Hawking singularity theorem \cite{HawkingPenroseSingularityThm}. Depending on the details of the collapse, a temporary trapped surface might form around the star only in the collapse phase \cite{Mersini-Houghton2014-I}.

The result is approximate, however. Because gravitational collapse and backreaction of the Hawking radiation are not considered at the same time. Rather, what is done is to put a mirror to produce the expected Hawking radiation when there was no backreaction, and then calculating the effect of \emph{this} radiation on the evolution of the star. However, the main result of \cite{Mersini-Houghton2014-I}, that black holes do not form, is shown to be the case in a number of numerical calculations Mersini did with Harald P. Pfeiffer\ip{Pfeiffer, Harald P.} in \cite{Mersini-Houghton2014-II}.

In these numerical calculations the homogeneity assumption is relaxed, the initial vacuum is taken to be the Unruh vacuum, and the obtained results showed that gravitationally collapsing stars reverse their implosion and begin to explode at radii larger than their Schwarzschild radii \cite{Mersini-Houghton2014-II}. The collapse dynamics are found \cite{Mersini-Houghton2014-II} to follow the Openheimer-Snyder model \cite{PhysRev.56.455}\footnote{This model describes the classical implosion of a spherically symmetric star consisting of dust (so there is no pressure) into a black hole. For a quick review about the importance of this paper, readers may like to see this APS Focus article \cite{PhysRevFocus.13.23}.} until a radius near the Schwarzschild radius is approached, at which time the situation change drastically and star begins to bounce, to explode due to backreaction of Hawking radiation.

As a self-criticism \cite{Mersini-Houghton2014-I} mentiones that the effect of perturbations on the stability of solutions needs to be investigated. Reference \cite{Mersini-Houghton2014-II}, on the other hand, gives promises about future calculations where fluid dynamics will have been included for the star. Backscattering by the centrifugal barrier is not considered in both of the works \cite{Mersini-Houghton2014-I, Mersini-Houghton2014-II}. Future work in this direction that includes the backreaction of the Hawking radiation seems to bring many unexpected results onto the discussion table.

We would like to mention \emph{en passe} that Stephen Hawking\ip{Hawking, Stephen W.} published a note \cite{HawkingWeather} which is related to his talk\footnote{Readers may find Hawking's as well as other participants' talks through the web address given in \cite{FuzzorFire}.} at Kavli Institute for Theoretical Physics. He contemplates on the idea that only apparent horizon but no event horizons form and mentions that this ``is the only resolution of the paradox compatible with CPT'' invariance\ic{CPT invariance} \cite{HawkingWeather}. This idea, of course, requires a modification of classical gravitation. Perhaps, when the backreaction is taken into account, there will be no need to modify gravity at the classical level as suggested by the results presented in this section.

Mersini and Pfeiffer declares the solution of the information loss paradox and therefore implying that of the firewall paradox \cite{Mersini-Houghton2014-II}, because no singularity or event horizon ever formed when they took into account the backreaction at least for the cases that were under consideration. Hence the subtitle of \cite{Mersini-Houghton2014-II}: ``Fireworks instead of firewalls.'' Fireworks, as a description of star explosions; fireworks, as a sign of celebration.

\section{Balanced Holography Interpretation of Black Hole Entropy}\ic{balanced holography|textbf}

In a series of two papers \cite{Verlinde2013a, Verlinde2013b} that appeared on the same day on the arXiv, Erik Verlinde\ip{Verlinde, Erik}  and Hermann Verlinde\ip{Verlinde, Hermann} put forward \cite{Verlinde2013a} a third interpretation for the Bekenstein entropy of a black hole: the \emph{balanced holography} interpretation. The first interpretation of the Bekenstein entropy is that it counts the number of black hole quantum states, whereas the second interpretation is that it merely expresses the cross horizon entanglement of field modes present in the vacuum state \cite{Verlinde2013a}.

Two hypothesis are proposed \cite{Verlinde2013a}, the first one is the following:

\begin{enumerate}
\item[H1] A typical quantum black hole, soon after it is formed, is close
  to maximally entangled with its environment.
\end{enumerate}

This is another expression for saying that the horizon geometry is smooth.

It is supposed that with $N$ many qubits one can describe the state of the
black hole, and then they define ``the entangled environment E of a
young black hole as the 2N dimensional Hilbert space spanned by all
states that are entangled with the black hole interior H''
\cite{Verlinde2013a}. Now the entangled state of the black hole lies
in the product space $\mc H_H \otimes \mc H_E$, which is however
$2^{2N}$ dimensional. If all the states were available, then the
entropy of the black hole would be twice the value of Bekenstein
entropy \cite{Verlinde2013a}. As regards the number of physical vacuum
states, here comes their second hypothesis \cite{Verlinde2013a}:

\begin{enumerate}
\item[H2] The physical Hilbert space of a young black hole and its
  entangled environment E is $e^{S_{\text{BH}}} = 2^N$ dimensional.
\end{enumerate}

So there should be $N$ many conditions that are satisfied by vacuum solutions, whose explicit form is to be determined by the specific theory to be considered \cite{Verlinde2013a}.

Erik and Hermann Verlinde make a distinction between virtual\ic{qubit!virtual} and logical qubits\ic{qubit!logical}: virtual qubits do not carry information and are determined by the vacuum state, whereas logical qubits can carry any information \cite{Verlinde2013a}. An example would be appropriate to clarify the meaning of virtual qubits. Two electons of a neutral helium atom that is in its ground state exist in a specific entangled state: they have anti-parallel spins. However this entanglement does not carry any information. The very fact that the atom is in its ground state requires this specific state. According to \cite{Verlinde2013a}, the key to solution of the firewall paradox lies in this distinction.

Due to maximal entanglement between the black hole and its entangled environment, by supposing an implicit symmetry between the two, it is expressed that the two subsystems consist of more or less equal number of virtual and logical qubits \cite{Verlinde2013a}. The task is to reveal the virtual and logical qubits in the state vector. For that purpose, quantum CNOT\ic{CNOT} gate\footnote{Controlled NOT gate. It is a unitary operation acting on two qubits. It changes the state of the first, provided that the second one reads one. Otherwise it acts as an identity. For example: $\ket{\psi}\otimes\ket{0} \mapsto \ket{\psi}\otimes\ket{0}$ and $\ket{0}\otimes\ket{1} \mapsto \ket{1}\otimes\ket{1}$.} is considered. For example, \cite{Verlinde2013a} applies the $\UCNOT$ operation to the following state (balanced black hole state\ic{balanced black hole state}):

\begin{equation}
  \ket{\Psi} = \alpha_0 \ket{0}_{\text{H}} \otimes \ket{0}_{\text{E}} + \alpha_1 \ket{1}_{\text{H}} \otimes \ket{1}_{\text{E}},
\end{equation}

where the first is a horizon state and the second is a state in the entangled environment. When $\UCNOT$ is applied one gets:

\begin{equation}
  \UCNOT \ket \Psi = \ket{0}_{\text{H}} \otimes (\alpha_0 \ket{0}_{\text{E}} + \alpha_1 \ket{1}_{\text{E}}).
\end{equation}

From the point of view of \cite{Verlinde2013a}, this operation has isolated the virtual and logical qubits. First one in the outer product represents the horizon in its ground state and the terms in the paranthesis is just one qubit that is suitable to carry information about the black hole quantum state. If one would apply $\UCNOT$ to, for example, $\ket \Phi = \beta_0 \ket{1}_{\text{H}} \otimes \ket{0}_{\text{E}} + \beta_1 \ket{0}_{\text{H}} \otimes \ket{1}_{\text{E}}$ one would get:

\begin{equation}
  \UCNOT \ket \Phi = \ket{1}_{\text{H}} \otimes (\beta_0 \ket{0}_{\text{E}} + \beta_1 \ket{1}_{\text{E}}).
\end{equation}

Reference \cite{Verlinde2013a} would interpret this state as a firewall state, since the horizon degree of freedom is seen to be excited.

This revelation process is generalized for larger numbers of qubits
and the name of the operator becomes $\UQT$, where similar to $\UCNOT$, it
``must execute unitary quantum teleportation protocol, or entanglement
swap, that transports all free quantum information from the interior
H to the exterior E\ldots'' \cite{Verlinde2013a}.

The main idea of balanced holography\ic{balanced holography!main idea} is that while a distant observer sees the state $\ket \Psi$, an infalling observer encounters the state $\UQT \ket \Psi$ which has all the information in the outside and in particular possesses a smooth horizon \cite{Verlinde2013a}. This dual view is made possible by black hole complementarity\ic{black hole complementarity}. Nevertheless, in order to bring a solution to the paradox, \cite{Verlinde2013a} deviates from the complementarity by modifying one of its postulates, the one about the interpretation of the black hole entropy.

The arguments presented so far are about \emph{young} black holes. The case of \emph{old} black holes is taken into account in the second article \cite{Verlinde2013b}. The idea is to write it an old black hole state in terms of balanced black hole states. For example, if the old black hole has mass $M$, there are young black hole states of this mass as well. In total, there are $\exp(S_{\text{BH}})$ many states. Because young and old black holes' Hilbert spaces are of the same dimension, it is then inferred that young black hole states are sufficient to describe the state of the old black hole \cite{Verlinde2013b}. Because for old black holes, all of the degrees of freedom are entangled with the early radiation, the quantum state of the system is given to be \cite{Verlinde2013b}:

\begin{equation}
  \ket \Psi _{\text{old}} = \sum_i \ket{i,i}_{\text{HE}} \otimes \ket{\Phi_i}_{\text{R}} \label{eq:58}
\end{equation}

where $\ket{i,i}_{\text{HE}}$ denotes a balanced black hole state, and $\ket{\Phi_i}_{\text{R}}$ is an early radiation state. 

When $\UQT$ is applied to (\ref{eq:58}) it will yield someting of the form $\ket{\text{vac}}_{\text{H}} \otimes \sum_i \ket{i}_{\text{E}} \otimes \ket{\Phi_i}_{\text{R}}$. The infalling observer does not encounter a singularity at the horizon. Reference \cite{Verlinde2013b} observes that balanced black hole states are similar to topological qubits\ic{qubit!topological} (which store quantum information non-locally \cite{Verlinde2013b}) mathematically, then puts forward as a hypothesis that ``Black hole information is protected from local sources of decoherence. As a result, it can not be measured, altered, or mined by local probes inside the zone'' \cite{Verlinde2013b}. From this idea, it is inferred that the interaction Hamiltonian that gives rise to time evolution does not alter the form of the (\ref{eq:58}) and hence preserve the balanced black hole states \cite{Verlinde2013b}.

Summing up the main arguments of \cite{Verlinde2013a, Verlinde2013b}, it is put forward that by adopting a third interpretation of the Bekenstein entropy one opens a room for non-local information storage which will, in the end, cause black holes to evaporate unitarily and not possess firewalls. The main idea is the transfer of entanglement relative to the particular observer under consideration, which is backed by complementarity ideas.


\chapter{Conclusion}
\label{chap:conclusion}

We have discussed a few introductory concepts in quantum gravity such as Hawking radiation\ic{Hawking radiation} and black hole complementarity\ic{black hole complementarity} in order to gather the necessary tools to understand the firewall paradox. Then, we introduced the paradox. It states that black hole complementarity (the most popular solution to the information paradox) is in contradiction with Einstein's equivalence principle (the building principle behind the general theory of relativity). The creators of the paradox suggested \cite{Almheiri2012} the existence of a highly excited region around the black hole's event horizon in violation of Einstein's equivalence principle as the ``most conservative'' solution of the tenison. Hence the name \emph{firewall}.

In another article named ``An apologia for firewalls'' \cite{Almheiri2013}, with the addition of Douglas Stanford\ip{Stanford, Douglas} as the fifth author, the AMPSS defended the existence of firewalls against the critics. As Sean Carroll\ip{Carroll, Sean} mentioned in his blog-post \cite{carroll-firewall-blog} the word \emph{apologia} ``means ``defense,'' not ``apology.'' '' As of now, the paradox stands still.

We tried to cover various paths of research on this subject. The perspectives we included are good to exemplify the multitude of approaches that exist in the literature towards the resolution of the paradox.

There is a dislike for firewall among some physicists. This is not a secret. One of the reasons for this attitude is that if one holds onto the black hole complementarity, then the equivalence principle can break down at arbitrarily low curvature regions. It can happen mainly because an old black hole can be of any mass, so the curvature around its horizon can be arbitrarily low. This is really \emph{tragic}. One is reminded of the following passage from Nietzsche\ip{Nietzsche, Friedrich Wilhelm} \cite{NietzscheHumantooHuman} (emphasis in the original):

\begin{quote}
  The laws of numbers were invented on the basis of the initially prevailing error that there are various identical things (but actually there is nothing identical) or at least there are things (but there is no ``thing''). The assumption of multiplicity always presumes that there is \emph{something}, which occurs repeatedly. But this is just where error rules; even here, we invent entities, unities, that do not exist.

Our feelings of space and time are false, for if they are tested rigorously, they lead to logical contradictions. Whenever we establish something scientifically, we are inevitably always reckoning with some incorrect quantities; but because these quantities are at least \emph{constant} (as is, for example, our feeling of time and space), the results of science do acquire a perfect strictness and certainty in their relationship to each other. One can continue to build upon them--up to that final analysis, where the mistaken basic assumptions, those constant errors, come into contradiction with the results, for example, in atomic theory.
\end{quote}

Tragedy is an essential constituent of a meaningful life. If one wants to climb the mount Everest, then one must accept the possibility of falling down off a cliff.\footnote{The idea of \emph{insurance} takes away all the joy, sorrow and meaning from one's life.} What happened in quantum gravity is that we are about to fall off a cliff, though not yet, if we cannot hold onto a principle. Almost everybody thinks about which principle of nature to give up: 1) unitary black hole evaporation, 2) semiclassical approximation outside stretched horizon, 3) the number of black hole quantum states is given by $\exp(A/4)$, 4) equivalence principle, 5) proximity postulate. The proximity postulate\ic{black hole complementarity!proximity postulate} is the idea that ``the interior of a black hole must be constructed from degrees of freedom that are physically near the black hole'' \cite{SusskindHHConj}. Perhaps we need to think about what idea, which principle to hold on to. Julian Barbour's\ip{Barbour, Julian} interpretation of the Mach's principle\ic{Mach's principle} \cite{Barbour2010} would be a good starting point.

There are many ideas in the air. One gets the impression that any one of them, in itself, is not sufficient to solve the paradox; however there may be little pieces towards the resolution that are scattered across the multitude of ideas. On the other hand, the ties between physics and philosophy has weakened considerably\footnote{Readers may find Pablo Echenique-Robba's\ip{Echenique-Robba, Pablo} work \cite{Echenique-Robba} useful as a counter-attack to the dogmatic ``Would you please be quiet and do your calculation?'' approach.}. It seems that this paradox will not be solved if philosophy is contuniously disregarded by physicists. Occasionally, the current situation is likened to the one in the beginning of the 20th century where major paradigm shifts occurred. However what we miss is philosophy. The situations are not quite the same. For example, the \emph{fact} that there is an information \emph{paradox} for black holes is because of the \emph{belief} that one should be able to reverse the \emph{evolution} of a physical system. This \emph{belief} is non-scientific, it is philosophical. Perhaps, in the end, there are dry mathematical rules that govern the cosmos; however one \emph{should} be able to fall in love. Philosophy is a necessity.

\bibliography{./content/refs}

\begin{thebibliography}{10}

\bibitem{solarMassNasa}
http://nssdc.gsfc.nasa.gov/planetary/factsheet/sunfact.html, 1 July 2013.

\bibitem{bhAstroGhez}
http://www.astro.ucla.edu/\~{}ghezgroup/gc/images\_science.html, 19 November
  2014.

\bibitem{Almheiri2013}
Ahmed Almheiri, Donald Marolf, Joseph Polchinski, Douglas Stanford, and James
  Sully.
\newblock An apologia for firewalls.
\newblock {\em Journal of High Energy Physics}, 2013(9), September 2013.
\newblock arXiv:1304.6483.

\bibitem{Almheiri2012}
Ahmed Almheiri, Donald Marolf, Joseph Polchinski, and James Sully.
\newblock Black holes: complementarity or firewalls?
\newblock {\em Journal of High Energy Physics}, 2013(2), February 2013.
\newblock arXiv:1207.3123.

\bibitem{Barbour2010}
Julian Barbour.
\newblock The definition of {M}ach's principle.
\newblock {\em Foundations of Physics}, 40:1263--1284, October 2010.
\newblock arXiv:1007.3368.

\bibitem{BergmannCollapse}
Peter~G. Bergmann.
\newblock Gravitational collapse.
\newblock {\em Phys. Rev. Lett.}, 12:139--140, Feb 1964.

\bibitem{BirrellDavies}
N.~D. Birrell and P.~C.~W. Davies.
\newblock {\em Quantum fields in curved space}.
\newblock Cambridge University Press, 1982.

\bibitem{Bohr1935}
N.~Bohr.
\newblock {Can quantum-mechanical description of physical reality be considered
  complete?}
\newblock {\em Physical Review}, 48(8):696--702, October 1935.

\bibitem{BoussoHolPrinc}
Raphael Bousso.
\newblock The holographic principle.
\newblock {\em Rev. Mod. Phys.}, 74:825--874, August 2002.
\newblock arXiv:hep-th/0203101.

\bibitem{bousso-comp-not-enough}
Raphael Bousso.
\newblock Complementarity is not enough.
\newblock {\em Physical Review D}, 87:124023, June 2013.
\newblock arXiv:1207.5192.

\bibitem{FuzzorFire}
Raphael Bousso, Samir Mathur, Rob Myers, Joe Polchinski, Leonard Susskind, and
  Don Marolf.
\newblock {KITP Rapid Response Workshop: Black Holes: Complementarity, Fuzz, or
  Fire?}
\newblock August 2013.
\newblock http://online.kitp.ucsb.edu/online/fuzzorfire\_m13/.

\bibitem{Braunstein2013}
Samuel Braunstein, Stefano Pirandola, and Karol Życzkowski.
\newblock Better late than never: Information retrieval from black holes.
\newblock {\em Physical Review Letters}, 110(10):101301, March 2013.
\newblock arXiv:0907.1190v3.

\bibitem{Braunstein2009}
Samuel~L. Braunstein.
\newblock {Black hole entropy as entropy of entanglement, or it's curtains for
  the equivalence principle}.
\newblock arXiv:0907.1190v1.

\bibitem{carroll-gr-book}
Sean~M. Carroll.
\newblock {\em An introduction to general relativity: spacetime and geometry}.
\newblock Addison-Wesley, 2004.

\bibitem{carroll-firewall-blog}
Sean~M. Carroll.
\newblock Firewalls, burning brightly.
\newblock 5 June 2013.
\newblock
  http://www.preposterousuniverse.com/blog/2013/06/05/firewalls-burning-brightly.

\bibitem{DeFelice2010}
Antonio {De Felice} and Shinji Tsujikawa.
\newblock {f(R) Theories}.
\newblock {\em Living Reviews in Relativity}, 13, 2010.

\bibitem{DeRham2014}
Claudia de~Rham.
\newblock {Massive Gravity}.
\newblock {\em Living Reviews in Relativity}, 17, 2014.

\bibitem{Echenique-Robba}
Pablo Echenique-Robba.
\newblock Shut up and let me think. or why you should work on the foundations
  of quantum mechanics as much as you please.
\newblock arXiv:1308.5619.

\bibitem{Einstein1935}
A.~Einstein, B.~Podolsky, and N.~Rosen.
\newblock {Can quantum-mechanical description of physical reality be considered
  complete?}
\newblock {\em Physical Review}, 47(10):777--780, May 1935.

\bibitem{Fabbri2005}
Alessandro Fabbri and Jos\'{e} Navarro-Salas.
\newblock {\em {Modeling black hole evaporation}}.
\newblock Imperial College Press, 2005.

\bibitem{FrolovNovikovBHPhysics}
Valeri~P. Frolov and Igor~D. Novikov.
\newblock {\em Black hole physics: {B}asic concepts and new developments}.
\newblock Kluwer Academic Publishers, 1998.

\bibitem{SD-Extra-3}
Henrique Gomes.
\newblock Poincar\'e invariance and asymptotic flatness in shape dynamics.
\newblock {\em Phys.Rev.}, D88:024047, 2013.
\newblock arXiv:1212.1755.

\bibitem{SD-Extra-4}
Henrique Gomes, Sean Gryb, and Tim Koslowski.
\newblock Einstein gravity as a {3D} conformally invariant theory.
\newblock {\em Class.Quant.Grav.}, 28:045005, 2011.

\bibitem{GomesHerczeg}
Henrique Gomes and Gabriel Herczeg.
\newblock A rotating black hole solution for shape dynamics.
\newblock arXiv:1310.6095.

\bibitem{SD-Extra-2}
Henrique Gomes and Tim Koslowski.
\newblock Shape dynamics and gauge-gravity duality.
\newblock arXiv:1301.7688.

\bibitem{KoslowskiGomesLink}
Henrique Gomes and Tim Koslowski.
\newblock The link between general relativity and shape dynamics.
\newblock {\em Class.Quant.Grav.}, 29:075009, March 2012.
\newblock arXiv:1101.5974.

\bibitem{KoslowskiGomesFreq}
Henrique Gomes and Tim Koslowski.
\newblock Frequently asked questions about shape dynamics.
\newblock {\em Foundations of Physics}, 43(12):1428--1458, December 2013.
\newblock arXiv:1211.5878.

\bibitem{HarlowPrivate}
Daniel Harlow.
\newblock Private communication.

\bibitem{HarlowOnPapaRaju}
Daniel Harlow.
\newblock Aspects of the {P}apadodimas-{R}aju proposal for the black hole
  interior.
\newblock arXiv:1405.1995.

\bibitem{harlow-pirsa}
Daniel Harlow.
\newblock Quantum computation vs. firewalls, February 2013.
\newblock http://pirsa.org/13020138.

\bibitem{harlow-hayden}
Daniel Harlow and Patrick Hayden.
\newblock {Quantum computation vs. firewalls}.
\newblock {\em {Journal of High Energy Physics}}, ({6}), June {2013}.
\newblock arXiv:1301.4504.

\bibitem{hartleGR}
James~B. Hartle.
\newblock {\em Gravity: An introduction to {E}instein's general relativity}.
\newblock Addison-Wesley, 2003.

\bibitem{HawkingPenroseSingularityThm}
S.~W. Hawking and R.~Penrose.
\newblock The singularities of gravitational collapse and cosmology.
\newblock {\em Proceedings of the Royal Society A: Mathematical, Physical and
  Engineering Sciences}, 314(1519):529--548, January 1970.

\bibitem{HawkingWeather}
Stephen~W. Hawking.
\newblock Information preservation and weather forecasting for black holes.
\newblock arXiv:1401.5761.

\bibitem{hawking1975}
Stephen~W. Hawking.
\newblock Particle creation by black holes.
\newblock {\em Communications in Mathematical Physics}, 43:199--220, 1975.

\bibitem{hawking-ellis}
Stephen~W. Hawking and G.~F.~R. Ellis.
\newblock {\em The large scale structure of space-time}.
\newblock Cambridge University Press, 1973.

\bibitem{SD-Extra-1}
Gabriel Herczeg and Vasudev Shyam.
\newblock Towards black hole entropy in shape dynamics.
\newblock arXiv:1410.4248.

\bibitem{hossenfelderCommentOnFirewall}
Sabine Hossenfelder.
\newblock Comment on the black hole firewall.
\newblock arXiv:1210.5317.

\bibitem{HossenfelderHawkingRad}
Sabine Hossenfelder.
\newblock Disentangling the black hole vacuum.
\newblock arXiv:1401.0288.

\bibitem{Stojkovic2013}
John Hutchinson and Dejan Stojkovic.
\newblock Icezones instead of firewalls: extended entanglement beyond the event
  horizon and unitary evaporation of a black hole.
\newblock arXiv:1307.5861.

\bibitem{ilgin-yang-inside-old}
{\.{I}}rfan Ilgın and I-Sheng Yang.
\newblock Causal patch complementarity: The inside story for old black holes.
\newblock {\em Physical Review D}, 89:044007, February 2014.
\newblock arXiv:1311.1219.

\bibitem{jensen-karch}
Kristan Jensen and Andreas Karch.
\newblock Holographic dual of an einstein-podolsky-rosen pair has a wormhole.
\newblock {\em Physical Review Letters}, 111:211602, November 2013.
\newblock arXiv:1307.1132.

\bibitem{ssaLiebRuskai}
Elliott~H. Lieb and Mary~Beth Ruskai.
\newblock Proof of the strong subadditivity of quantum-mechanical entropy.
\newblock {\em Journal of Mathematical Physics}, 14:1938--1941, 1973.

\bibitem{PhysRevFocus.13.23}
David Lindley.
\newblock Landmarks.
\newblock {\em Phys. Rev. Focus}, 13:23, May 2004.
\newblock http://link.aps.org/doi/10.1103/PhysRevFocus.13.23.

\bibitem{origin-of-bh-concept-false-bh}
A.~Loinger.
\newblock {On Michell-Laplace dark body}.
\newblock October 2003.
\newblock arXiv:physics/0310058.

\bibitem{LoweThorBHInfoProb}
L\'{a}rus Lowe, David A.;~Thorlacius.
\newblock Remarks on the black hole information problem.
\newblock {\em Physical Review D}, 73:104027, 2006.

\bibitem{maggiore-qft}
Michele Maggiore.
\newblock {\em A modern introduction to quantum field theory}.
\newblock Oxford University Press, 2005.

\bibitem{ErEpr}
Juan Maldacena and Leonard Susskind.
\newblock Cool horizons for black holes.
\newblock arXiv:1306.0533v2.

\bibitem{MathurSSaboutBH}
Samir~D. Mathur.
\newblock What does strong subadditivity tell us about black holes?
\newblock arXiv:1309.6583.

\bibitem{mathur-info}
Samir~D Mathur.
\newblock The information paradox: a pedagogical introduction.
\newblock {\em Classical and Quantum Gravity}, 26(22):224001, October 2009.
\newblock arXiv:0909.1038.

\bibitem{MathurTurtonFlaw}
Samir~D. Mathur and David Turton.
\newblock The flaw in the firewall argument.
\newblock {\em Nuclear Physics B}, 884(0):566 -- 611, July 2014.
\newblock arXiv:1306.5488.

\bibitem{qft-demystified}
David McMahon.
\newblock {\em Quantum field theory demystified}, pages 38--39.
\newblock McGraw-Hill, 2008.

\bibitem{Mersini-Houghton2014-I}
Laura Mersini-Houghton.
\newblock Backreaction of hawking radiation on a gravitationally collapsing
  star i: Black holes?
\newblock arXiv:1406.1525.

\bibitem{Mersini-Houghton2014-II}
Laura Mersini-Houghton and Harald~P. Pfeiffer.
\newblock Back-reaction of the hawking radiation flux on a gravitationally
  collapsing star ii: Fireworks instead of firewalls.
\newblock arXiv:1409.1837.

\bibitem{origin-of-bh-concept}
C.~{Montgomery}, W.~{Orchiston}, and I.~{Whittingham}.
\newblock {Michell, Laplace and the origin of the black hole concept}.
\newblock {\em Journal of Astronomical History and Heritage}, 12:90--96, July
  2009.

\bibitem{mukhanov}
V.~F. Mukhanov and S.~Winitzki.
\newblock {\em Introduction to quantum effects in gravity}.
\newblock Cambridge University Press, 2007.

\bibitem{NastaseAdSCFT}
Horatiu Nastase.
\newblock Introduction to {AdS-CFT}.
\newblock arXiv:0712.0689.

\bibitem{NietzscheHumantooHuman}
Friedrich~Wilhelm Nietzsche.
\newblock {\em Human, All Too Human}, pages 26--27.
\newblock Penguin Classics.
\newblock translated by Marion Faber and Stephen Lehmann.

\bibitem{PhysRev.56.455}
J.~R. Oppenheimer and H.~Snyder.
\newblock On continued gravitational contraction.
\newblock {\em Phys. Rev.}, 56:455--459, Sep 1939.

\bibitem{pageAvgEntropy}
Don~N. Page.
\newblock Average entropy of a subsystem.
\newblock {\em Physical Review Letters}, 71:1291--1294.

\bibitem{page-ecc}
Don~N. Page.
\newblock Excluding black hole firewalls with extreme cosmic censorship.
\newblock arXiv:1306.0562.

\bibitem{Page1976}
Don~N. Page.
\newblock {Particle emission rates from a black hole: Massless particles from
  an uncharged, nonrotating hole}.
\newblock {\em Physical Review D}, 13(2):198--206, January 1976.

\bibitem{PageBHInfo}
Don~N. Page.
\newblock Information in black hole radiation.
\newblock {\em Physical Review Letters}, 71:3743, December 1993.
\newblock arXiv:hep-th/9306083.

\bibitem{ipala}
{\.{I}}skender Pala.
\newblock {\em Kronolojik Divan Şiiri Antolojisi}.
\newblock Kapı yayınları, 2004.

\bibitem{PapaRajuInFall}
Kyriakos Papadodimas and Suvrat Raju.
\newblock An infalling observer in ads/cft.
\newblock {\em Journal of High Energy Physics}, 2013(10), October 2013.
\newblock arXiv:1211.6767.

\bibitem{PapaRajuPRL}
Kyriakos Papadodimas and Suvrat Raju.
\newblock Black hole interior in the holographic correspondence and the
  information paradox.
\newblock {\em Phys. Rev. Lett.}, 112:051301, February 2014.
\newblock arXiv:1310.6334.

\bibitem{PapaRajuPRD}
Kyriakos Papadodimas and Suvrat Raju.
\newblock State-dependent bulk-boundary maps and black hole complementarity.
\newblock {\em Phys. Rev. D}, 89:086010, April 2014.
\newblock arXiv:1310.6335.

\bibitem{parker-toms}
Leonard Parker and David Toms.
\newblock {\em Quantum field theory in curved spacetime}.
\newblock Cambridge University Press, 2009.

\bibitem{PetersenAdSCFT}
Jens~Lyng Petersen.
\newblock Introduction to the {M}aldacena conjecture on {AdS/CFT}.
\newblock {\em International Journal of Modern Physics {A}},
  {14}({23}):{3597--3672}, September {1999}.
\newblock arXiv:hep-th/9902131.

\bibitem{Plotnitsky2012}
Arkady Plotnitsky.
\newblock {\em {Niels Bohr and complementarity}}.
\newblock Springer Briefs in Physics. Springer New York, New York, NY, 2012.

\bibitem{preskill-ln}
John Preskill.
\newblock Lecture notes on quantum information.
\newblock http://www.theory.caltech.edu/people/preskill/ph219/\#lecture.

\bibitem{RamalloAdSCFT}
Alfonso~V. Ramallo.
\newblock Introduction to the {AdS/CFT} correspondence.
\newblock arXiv:1310.4319.

\bibitem{Rindler1966}
W.~Rindler.
\newblock {Kruskal space and the uniformly accelerated frame}.
\newblock {\em American Journal of Physics}, 34(12):1174, July 1966.

\bibitem{schwarzschild1916EngTrans}
K.~Schwarzschild.
\newblock "{G}olden {O}ldie": On the gravitational field of a mass point
  according to {E}instein's theory.
\newblock {\em General Relativity and Gravitation}, 35(5):951--959, May 2003.

\bibitem{schwarzschild1916}
Karl Schwarzschild.
\newblock {Über das Gravitationsfeld eines Massenpunktes nach der
  Einsteinschen Theorie}.
\newblock {\em Sitzungsberichte der K{\"o}niglich Preussischen Akademie der
  Wissenschaften zu Berlin, Phys.-Math. Klasse}, pages 189--196, 1916.

\bibitem{sonner}
Julian Sonner.
\newblock Holographic {S}chwinger effect and the geometry of entanglement.
\newblock {\em Physical Review Letters}, 111:211603, November 2013.
\newblock arXiv:1307.6850.

\bibitem{DejanPrivate}
Dejan Stojkovic.
\newblock Private communication.

\bibitem{SusskindHHConj}
Daniele~C. Struppa and Jeffrey~M. Tollaksen, editors.
\newblock {\em Quantum Theory: A Two-Time Success Story}, chapter The Limits of
  Black Hole Complementarity.
\newblock Springer, 2014.
\newblock arXiv:1301.4505.

\bibitem{susskind-addendum}
Leonard Susskind.
\newblock Addendum to computational complexity and black hole horizons.
\newblock arXiv:1403.5695.

\bibitem{susskind-butterfly}
Leonard Susskind.
\newblock Butterflies on the stretched horizon.
\newblock arXiv:1311.7379.

\bibitem{Susskind2014}
Leonard Susskind.
\newblock {Computational complexity and black hole horizons}.
\newblock arXiv:1402.5674.

\bibitem{susskind-new-old-bh}
Leonard Susskind.
\newblock New concepts for old black holes.
\newblock arXiv:1311.3335.

\bibitem{Susskind2012}
Leonard Susskind.
\newblock The transfer of entanglement: The case for firewalls.
\newblock arXiv:1210.2098.

\bibitem{Susskind2005}
Leonard Susskind and James Lindesay.
\newblock {\em An introduction to black holes, information and the string
  theory revolution: the holographic universe}.
\newblock World Scientific, 2005.

\bibitem{complementarityGedanken}
Leonard Susskind and L\'{a}rus Thorlacius.
\newblock Gedanken experiments involving black holes.
\newblock {\em Physical Review D}, 49:966, January 1994.
\newblock arXiv:hep-th/9308100.

\bibitem{Susskind1993}
Leonard Susskind, L\'{a}rus Thorlacius, and John Uglum.
\newblock {The stretched horizon and black hole complementarity}.
\newblock {\em Physical Review D}, 48(8):3743--3761, October 1993.
\newblock arXiv:hep-th/9306069.

\bibitem{swingle}
Brian Swingle.
\newblock Entanglement renormalization and holography.
\newblock {\em Physical Review D}, 86:065007, September 2012.
\newblock arXiv:0905.1317.

\bibitem{tHooftDimReduction}
Gerardus 't~Hooft.
\newblock Dimensional reduction in quantum gravity.
\newblock arXiv:gr-qc/9310026.

\bibitem{UnruhRadiationNotesOnBH}
William~G. Unruh.
\newblock Notes on black-hole evaporation.
\newblock {\em Physical Review D}, 14:870, 1976.

\bibitem{UnruhWaldPRD_25_942}
William~G. Unruh and Robert~M. Wald.
\newblock Acceleration radiation and the generalized second law of
  thermodynamics.
\newblock {\em Physical Review D}, 25:942.

\bibitem{UnruhWaldPRD_27_2271}
William~G. Unruh and Robert~M. Wald.
\newblock Entropy bounds, acceleration radiation, and the generalized second
  law.
\newblock {\em Physical Review D}, 27:2271.

\bibitem{UnruhWaldEssay}
William~G. Unruh and Robert~M. Wald.
\newblock How to mine energy from a black hole.
\newblock {\em General Relativity and Gravitation}, 15:195--199.

\bibitem{RammsdonkEssay}
Mark van Raamsdonk.
\newblock Building up spacetime with quantum entanglement.
\newblock {\em General Relativity and Gravitation}, 42:2323--2329, October
  2010.
\newblock arXiv:1005.3035.

\bibitem{Verlinde2013b}
Erik Verlinde and Herman Verlinde.
\newblock Black hole information as topological qubits.
\newblock arXiv:1306.0516.

\bibitem{Verlinde2013a}
Erik Verlinde and Herman Verlinde.
\newblock Passing through the firewall.
\newblock arXiv:1306.0515.

\bibitem{WaldParticleBH}
Robert~M. Wald.
\newblock Particle creation by black-holes.
\newblock {\em Communications in Mathematical Physics}, {45}({1}):{9--34},
  {1975}.

\bibitem{waldQFT-BH}
Robert~M. Wald.
\newblock {\em Quantum field theory in curved spacetime and black hole
  thermodynamics}.
\newblock The University of Chicago Press, 1994.

\end{thebibliography}
\bibliographystyle{plain}

\appendix

\chapter{An Analogue of Heisenberg's Uncertainty Principle}\ic{uncertainty principle|textbf}
\label{cha:an-anal-heis}

Heisenberg's uncertainty relation $\Delta x \Delta p \geq \hbar/2$ is a relation difficult to grasp for humans because our intuition evolved in situations where nonrelativistic classical physics is valid. However the following experiment might be an analogue way of understanding it.

We consider a fan that is mounted on a wall, see Figure~\ref{fig:ccomp-2}. The fan has an opaque lid that closes the propeller completely when it is closed. We suppose that the weather outside is moderately windy. The idea is that, when the lid is opened there occurs a current of air that makes the propeller rotate. Of course if the lid is closed, there can occur no sensible amount of air current that is enough to move it, hence the propeller stands still. The more one opens the lid, the more quickly the propeller rotates.

We interpret the situation as follows: the degree of visibility of the fan corresponds to position measurements with increasing precision (decreasing $\Delta x$), hence it begins to rotate faster and faster as one uncovers more of it (increasing $\Delta p$). On the contrary, if one closes the lid more and more (increasing $\Delta x$ because one no longer sees it), the fan begins to slow down (decreasing $\Delta p$) because the air current no longer flows as freely as it used to when the lid was open.

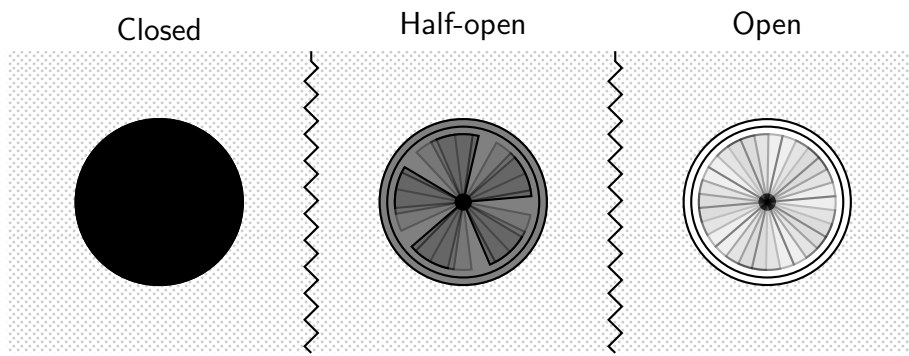
\begin{figure}
  \centering
  \begin{tikzpicture}[scale = 1, thick]
  \pattern [pattern=crosshatch dots, pattern color=black!20] (10,-2) rectangle (-2,2); 
  \filldraw [fill=white] (0,0) circle (1.1); 
  \draw (0,0) circle (1);   
  \filldraw [rotate=5, fill=black!20, draw=black] (0,0) -- (0.9,0) arc (0:36:0.9) -- cycle;
  \filldraw [rotate=77, fill=black!20, draw=black] (0,0) -- (0.9,0) arc (0:36:0.9) -- cycle;
  \filldraw [rotate=149, fill=black!20, draw=black] (0,0) -- (0.9,0) arc (0:36:0.9) -- cycle;
  \filldraw [rotate=221, fill=black!20, draw=black] (0,0) -- (0.9,0) arc (0:36:0.9) -- cycle;
  \filldraw [rotate=293, fill=black!20, draw=black] (0,0) -- (0.9,0) arc (0:36:0.9) -- cycle;
  \filldraw [fill=black, draw=black] (0,0) circle (0.1); 
  \filldraw [fill=black, draw=black, opacity=1] (0,0) circle (1.1); 
  \draw[decorate, decoration=zigzag] (2,-2) -- (2,2);
  \draw[decorate, decoration=zigzag] (6,-2) -- (6,2);
  \draw (0,2) node [anchor=south, above] {Closed};
  \draw (4,2) node [anchor=south, above] {Half-open};
  \draw (8,2) node [anchor=south, above] {Open};
  \begin{scope}[xshift=4cm]
      \filldraw [fill=white] (0,0) circle (1.1); 
      \draw (0,0) circle (1);   
      \begin{scope}
        \filldraw [rotate=5, fill=black!20, draw=black] (0,0) --
        (0.9,0) arc (0:36:0.9) -- cycle; \filldraw [rotate=77,
        fill=black!20, draw=black] (0,0) -- (0.9,0) arc (0:36:0.9) --
        cycle; \filldraw [rotate=149, fill=black!20, draw=black] (0,0)
        -- (0.9,0) arc (0:36:0.9) -- cycle; \filldraw [rotate=221,
        fill=black!20, draw=black] (0,0) -- (0.9,0) arc (0:36:0.9) --
        cycle; \filldraw [rotate=293, fill=black!20, draw=black] (0,0)
        -- (0.9,0) arc (0:36:0.9) -- cycle;
      \end{scope}
      \begin{scope}[opacity=0.5, rotate=5]
        \filldraw [rotate=5, fill=black!20, draw=black] (0,0) --
        (0.9,0) arc (0:36:0.9) -- cycle; \filldraw [rotate=77,
        fill=black!20, draw=black] (0,0) -- (0.9,0) arc (0:36:0.9) --
        cycle; \filldraw [rotate=149, fill=black!20, draw=black] (0,0)
        -- (0.9,0) arc (0:36:0.9) -- cycle; \filldraw [rotate=221,
        fill=black!20, draw=black] (0,0) -- (0.9,0) arc (0:36:0.9) --
        cycle; \filldraw [rotate=293, fill=black!20, draw=black] (0,0)
        -- (0.9,0) arc (0:36:0.9) -- cycle;
      \end{scope}
      \begin{scope}[opacity=0.3, rotate=20]
        \filldraw [rotate=5, fill=black!20, draw=black] (0,0) --
        (0.9,0) arc (0:36:0.9) -- cycle; \filldraw [rotate=77,
        fill=black!20, draw=black] (0,0) -- (0.9,0) arc (0:36:0.9) --
        cycle; \filldraw [rotate=149, fill=black!20, draw=black] (0,0)
        -- (0.9,0) arc (0:36:0.9) -- cycle; \filldraw [rotate=221,
        fill=black!20, draw=black] (0,0) -- (0.9,0) arc (0:36:0.9) --
        cycle; \filldraw [rotate=293, fill=black!20, draw=black] (0,0)
        -- (0.9,0) arc (0:36:0.9) -- cycle;
      \end{scope}
      \filldraw [fill=black, draw=black] (0,0) circle (0.1); 
      \filldraw [fill=black, draw=black, opacity=0.5] (0,0) circle (1.1); 
    \end{scope}
  \begin{scope}[xshift=8cm]
      \filldraw [fill=white] (0,0) circle (1.1); 
      \draw (0,0) circle (1);   
      \begin{scope}[opacity=0.3]
        \filldraw [rotate=5, fill=black!20, draw=black] (0,0) --
        (0.9,0) arc (0:36:0.9) -- cycle; \filldraw [rotate=77,
        fill=black!20, draw=black] (0,0) -- (0.9,0) arc (0:36:0.9) --
        cycle; \filldraw [rotate=149, fill=black!20, draw=black] (0,0)
        -- (0.9,0) arc (0:36:0.9) -- cycle; \filldraw [rotate=221,
        fill=black!20, draw=black] (0,0) -- (0.9,0) arc (0:36:0.9) --
        cycle; \filldraw [rotate=293, fill=black!20, draw=black] (0,0)
        -- (0.9,0) arc (0:36:0.9) -- cycle; \filldraw [fill=black,
        draw=black] (0,0) circle (0.1); 
      \end{scope}
      \begin{scope}[opacity=0.3, rotate=12]
        \filldraw [rotate=5, fill=black!20, draw=black] (0,0) --
        (0.9,0) arc (0:36:0.9) -- cycle; \filldraw [rotate=77,
        fill=black!20, draw=black] (0,0) -- (0.9,0) arc (0:36:0.9) --
        cycle; \filldraw [rotate=149, fill=black!20, draw=black] (0,0)
        -- (0.9,0) arc (0:36:0.9) -- cycle; \filldraw [rotate=221,
        fill=black!20, draw=black] (0,0) -- (0.9,0) arc (0:36:0.9) --
        cycle; \filldraw [rotate=293, fill=black!20, draw=black] (0,0)
        -- (0.9,0) arc (0:36:0.9) -- cycle; \filldraw [fill=black,
        draw=black] (0,0) circle (0.1); 
      \end{scope}
      \begin{scope}[opacity=0.3, rotate=24]
        \filldraw [rotate=5, fill=black!20, draw=black] (0,0) --
        (0.9,0) arc (0:36:0.9) -- cycle; \filldraw [rotate=77,
        fill=black!20, draw=black] (0,0) -- (0.9,0) arc (0:36:0.9) --
        cycle; \filldraw [rotate=149, fill=black!20, draw=black] (0,0)
        -- (0.9,0) arc (0:36:0.9) -- cycle; \filldraw [rotate=221,
        fill=black!20, draw=black] (0,0) -- (0.9,0) arc (0:36:0.9) --
        cycle; \filldraw [rotate=293, fill=black!20, draw=black] (0,0)
        -- (0.9,0) arc (0:36:0.9) -- cycle; \filldraw [fill=black,
        draw=black] (0,0) circle (0.1); 
      \end{scope}
      \begin{scope}[opacity=0.3, rotate=36]
        \filldraw [rotate=5, fill=black!20, draw=black] (0,0) --
        (0.9,0) arc (0:36:0.9) -- cycle; \filldraw [rotate=77,
        fill=black!20, draw=black] (0,0) -- (0.9,0) arc (0:36:0.9) --
        cycle; \filldraw [rotate=149, fill=black!20, draw=black] (0,0)
        -- (0.9,0) arc (0:36:0.9) -- cycle; \filldraw [rotate=221,
        fill=black!20, draw=black] (0,0) -- (0.9,0) arc (0:36:0.9) --
        cycle; \filldraw [rotate=293, fill=black!20, draw=black] (0,0)
        -- (0.9,0) arc (0:36:0.9) -- cycle; \filldraw [fill=black,
        draw=black] (0,0) circle (0.1); 
      \end{scope}
    \end{scope}
\end{tikzpicture}\vspace{1em}
  \caption{An analogue of uncertainty between position and momentum in quantum mechanics using a fan placed on a wall. There is a lid that may close the fan. When the lid is opened, there occurs a rather moderate amount of air current that flows from outside to inside which in the end makes the propeller rotate. We represent lid-opening in the figure by the opacity of the lid: the more translucent it is, the more open it is and vice versa.}
  \label{fig:ccomp-2}
\end{figure}

\chapter{Derivation of the $U(1)$-inner product}
\label{cha:derivation-u1-inner}

We shall derive the $U(1)$-charge\ic{U(1)@$U(1)$-charge} and show that it is constant on spacelike hypersurfaces. The conserved current (which is a functional of one field solution) that gives rise to this charge will then be modified to define another conserved current (which is a functional of two field solutions).

The full Lagrangian is present in equation \eqref{eq:4} and we rewrite it here for convenience:

\begin{equation}
  \label{eq:14}
  \mc L_r = \frac 1 2 \left( \nabla^\mu \phi \nabla_\mu \phi + m^2 \phi^2 + \xi R \phi^2 \right),
\end{equation}

where we denote it as $\mc L_r$ (Subscript $r$ is used because of $\phi$'s being real valued). The Euler-Lagrange equation obtained is $\left( \nabla^2 - m^2 - \xi R \right) \phi = 0$. However, there is another Lagrangian that gives the same equation of motion at the expense of $\phi$'s being complex valued:

\begin{equation}
  \label{eq:15}
  \mc L_c = \nabla^\mu \phi^* \nabla_\mu \phi + m^2 \phi^* \phi + \xi R \phi^* \phi .
\end{equation}

Variation with respect to $\phi^*$ yields the Euler-Lagrange equation mentioned above that is yielded by $\mc{L}_r$.

As we shall see shortly, we can easily find a conserved charge for complex scalar fields. However, since complex and real scalar fields satisfy the same equation of motion, we will have obtained a conserved charge for real scalar fields. Because the modes, or in other words basis vectors, upon which we will expand the real scalar field $\phi$ will be complex valued, we are going to keep the complex conjugation symbol throughout.

There is an \emph{internal symmetry}\ic{internal symmetry} present in $\mc L_c$. An internal symmetry\ic{internal symmetry!definition} concerns the invariance of the action under the variations of fields only. For example Lorentz invariance in usual, \ie flat spacetime, QFT involves the transformation on spacetime points, hence it is not an internal symmetry. However, the invariance of $\mc L_c$ under the mapping $\phi \mapsto e^{i\theta} \phi$ where $\theta \in \reals$ \emph{is} an internal symmetry, which we will use to define the $U(1)$-charge\ic{U(1)@$U(1)$-charge}.

As a side note, the term ``$U(1)$'' is used to denote the charge because it is the name of the \co{group} that corresponds to considered symmetry transformation of fields. It is the group of $1 \times 1$ dimensional unitary matrices. Its elements can be represented on a circle. Please see Figure~\ref{fig:u1-group}.

\begin{figure}
  \centering
  \begin{tikzpicture}[scale = 1, thick]
  \draw[step=0.5, gray, very thin] (-1.9,-1.9) grid (1.9,1.9);
  \draw (0,0) circle (1);
  \draw (-2,0) -- (2,0);
  \draw (0,-2) -- (0,2);
  \draw (-2,2) node [anchor=north west] {$\mathbb C$};
\end{tikzpicture}\vspace{1em}
  \caption{The elements of the group $U(1)$ form a circle. One may see this fact easily by considering where the elements $e^{i\theta}$ lie for $\theta \in \reals$.}
  \label{fig:u1-group}
\end{figure}
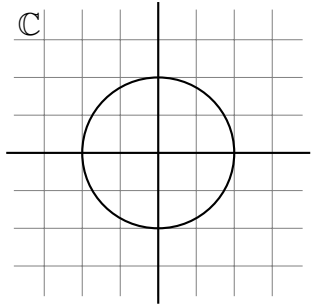

There are two ways to derive the $U(1)$-charge\ic{U(1)@$U(1)$-charge}. The first \cite{qft-demystified} is more direct, it involves the variation of Lagrangian with respect to $\phi$ and its invariance under a \emph{particular}\footnote{Of course $\delta \mc L$ cannot be equal to zero for \emph{all} variations, this is the essence of the variational principle.} type of variation described by the $U(1)$ symmetry. The second involves the invariance of the action, making the global $U(1)$ symmetry local and proceeding afterwards. Because the second method is more indirect and less obvious, we shall not show it. Interested reader may find the details of this approach in \cite{maggiore-qft}, pages 46--49.

\section{Conserved Current}
\label{sec:conserved-current}

We consider the Lagrangian $\mc L_c$. As the reader may have checked already, it is invariant under the mapping $\phi \mapsto e^{i\theta}\phi$. We consider an infinitesimal transformation, $\theta \simeq 0$, which means $\delta \phi = e^{i\theta}\phi - \phi \simeq i\theta\phi \simeq 0$.

We vary $\mc L_c$ with respect to $\phi,\phi^*$ and obtain:

\begin{align}
  \delta \mc L_c &= \pd{\mc L_c}{\phi} \delta \phi + \pd{\mc L_c}{\nabla_\mu \phi} \delta \nabla_\mu \phi + \pd{\mc L_c}{\phi^*} \delta \phi^* + \pd{\mc L_c}{\nabla_\mu \phi^*} \delta \nabla_\mu \phi^* ,\\
\intertext{Noting that $\delta$ and $\nabla_\mu$ commute, we write:}
                &= \pd{\mc L_c}{\phi} \delta \phi + \pd{\mc L_c}{\nabla_\mu \phi} \nabla_\mu \delta \phi + \pd{\mc L_c}{\phi^*} \delta \phi^* + \pd{\mc L_c}{\nabla_\mu \phi^*} \nabla_\mu \delta \phi^* .
\intertext{We can write $\nabla_\mu \delta \phi (\partial\mc L_c / \partial \nabla_\mu \phi)$ as $\nabla_\mu \left[\delta\phi (\partial\mc L_c/\partial \nabla_\mu \phi)\right] - \delta\phi \nabla_\mu (\partial\mc L_c / \partial\nabla_\mu \phi)$.  It is here we require $\phi$ to satisfy the Euler-Lagrange equation, in order to write the second term as $-\delta\phi (\partial\mc L_c / \partial \phi)$. One should do similar calculations for the $\phi^*$ part. After the cancelations we obtain:}
               &= \nabla_\mu \left[\delta \phi \pd{\mc L_c}{\nabla_\mu \phi}\right] + \nabla_\mu \left[\delta \phi^* \pd{\mc L_c}{\nabla_\mu \phi^*}\right],\\
               &= \nabla_\mu \left[\delta \phi \pd{\mc L_c}{\nabla_\mu \phi} + \delta \phi^* \pd{\mc L_c}{\nabla_\mu \phi^*}\right],\label{eq:32}\\
               &= 0.
\end{align}

The last equality is because of invariance of the Lagrangian. In the line before that, the term in square brackets represents a \co{conserved current}. Using $\delta \phi = i\theta\phi$ upto first order in $\theta$ one can calculate the current. It is found to be:

\begin{equation}
  \label{eq:16}
  i\theta (\phi \nabla^\mu \phi^* - \phi^* \nabla^\mu \phi).
\end{equation}

Because a constant times a conserved current is still a conserved current, we neglect the coefficient $i\theta$ in front and define $j^\mu[\phi]$ as:

\begin{equation}
  \label{eq:17}
  j^\mu[\phi] = \phi \nabla^\mu \phi^* - \phi^* \nabla^\mu \phi .
\end{equation}

It should be noticed that the conserved current does not depend on any addition to the Lagrangian that is only a function of the fields, not of their derivatives. For example, an addition of a function of $\phi^* \phi$ to $\mc L_c$ gives no contribution to conserved current, as can be seen via (\ref{eq:32}). Therefore, the conserved current found here can be useful for various other similar Lagrangians. However, of course the parts added should be invariant under the considered symmetry transformation.

Suppose $\phi$ and $\psi$ are two field solutions. If we defined:

\begin{equation}
  \label{eq:18}
  j^\mu[\phi,\psi] = \phi \nabla^\mu \psi^* - \psi^* \nabla^\mu \phi ,
\end{equation}

would this quantity describe a conserved current? The answer is positive, by virtue of field equations. We only need to show that $\nabla_\mu j^\mu[\phi,\psi] = 0$.

\begin{align}
  \nabla_\mu j^\mu[\phi,\psi] &= \nabla_\mu (\phi \nabla^\mu \psi^* - \psi^* \nabla^\mu \phi),\\
    &= \nabla_\mu \phi \nabla^\mu \psi^* + \phi \nabla^2 \psi^* - \nabla_\mu \psi^* \nabla^\mu \phi - \psi^* \nabla^2 \phi ,\\
    &= \phi \nabla^2 \psi^* - \psi^* \nabla^2 \phi .
\end{align}

By the field equations, we can write $(m^2+\xi R)\phi$ in lieu of $\nabla^2 \phi$ and do the same for $\psi^*$. When this is done, two terms cancel each other.

Hence, $j^\mu[\phi,\psi]$ is a conserved current. From now on, we sometimes denote it just as $j^\mu$.

\section{Definition of the Inner Product}
\label{sec:defin-inner-prod}

We consider a region of spacetime ($\mc M$) that has the \co{topology} $\reals \times \Sigma$, \ie it can be foliated by spacelike hypersurfaces\ic{hypersurface!spacelike}. For an illustration please see Figure~\ref{fig:foliation}.

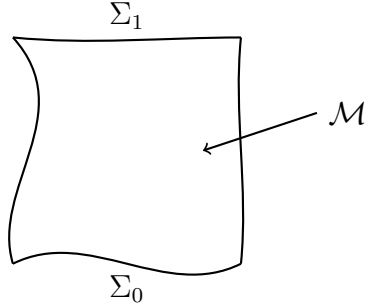
\begin{figure}
  \centering
  \begin{tikzpicture}[scale = 1, thick]
  \draw (0,0) .. controls (1,0.5) and (2,-0.5) .. (3,0) node [midway,below] {$\Sigma_0$};
  \draw (0,3) .. controls (1,2.9) and (2,3) .. (3,3) node [midway,above] {$\Sigma_1$};
  \draw (0,0) .. controls (-0.3,1) and (0.9,2) .. (0,3);
  \draw (3,0) .. controls (3.1,1) and (2.9,2) .. (3,3);
  \draw[->] (4,2) node[anchor=west] {$\mathcal M$} -- (2.5,1.5);
\end{tikzpicture}\vspace{1em}
  \caption{$\Sigma_0$ and $\Sigma_1$ are two spacelike hypersurfaces. We will see that the inner product we will define is constant on any spacelike hypersurface under consideration. This is the reason why it deserves the name ``inner product.''}
  \label{fig:foliation}
\end{figure}

Since $j^\mu$ is a \co{conserved current}, its divergence vanishes. If it is integrated on $\mc M$, the result will be zero of course. However, if we can use the \co{Stoke's theorem}, we can express this integral on the boundary ($\partial\mc M$) of $\mc M$. Part of the boundary that will be of importance to us is the union of two spacelike hypersurfaces, $\Sigma_0$ and $\Sigma_1$. We will suppose that the integral of the flux of $j^\mu$ on the rest of $\partial \mc M$ vanishes. For generality, we assume the spacetime is $(n + 1)$-dimensional.

\begin{align}
  0 &= \int_{\mc M} d^{n+1}x \abs g^{1/2} \nabla_\mu j^\mu ,\\
  \intertext{We use the Stoke's theorem.}
  0 &= \int_{\partial\mc M} d^nx \abs g^{1/2} n_\mu j^\mu .
\end{align}

Here $n^\mu$ is perpendicular to the boundary. We divide the boundary into three disjoint parts\footnote{In the presence of a black hole, $\Sigma_0$ will be the lightlike past. However, because we are foliating the spacetime outside the event horizon, as $t$ approaches $\infty$ the spacelike hypersurface that is appropriate to a stationary distant observer will asymptote to the event horizon and the lightlike future. Because massless particles do not reach the timelike future, $\Sigma_1$ is taken to consist of two pieces: the event horizon and the lightlike future.}: $\Sigma_0$, $\Sigma_1$ and $\partial\mc M \setminus (\Sigma_0 \cup \Sigma_1)$.

\begin{equation}
     0 = \int_{\Sigma_0} d^nx \abs g^{1/2} n_\mu j^\mu + \int_{\Sigma_1} d^nx \abs g^{1/2} n_\mu j^\mu + \int_{\partial\mc M \setminus (\Sigma_0 \cup \Sigma_1)} d^nx \abs g^{1/2} n_\mu j^\mu .
\end{equation}

We suppose that the last integral vanishes and hence require fields to vanish on this part of the boundary as a boundary condition. For example, usually field solutions on spacelike infinity are assumed to vanish.

\begin{equation}
    0 = \int_{\Sigma_0} d^nx \abs g^{1/2} n_\mu j^\mu + \int_{\Sigma_1} d^n \abs g^{1/2} n_\mu j^\mu .
\end{equation}

Here it should be noted that the vectors $n^\mu$ in $\Sigma_0$ and $n^\mu$ in $\Sigma_1$ point in \emph{opposite} directions. We can make both vectors point in the same direction, either forward or backward in time, by multiplying one of them with negative one. In this case we obtain:

\begin{equation}
  \label{eq:19}
  \int_{\Sigma_0} d^nx \abs g^{1/2} n_\mu j^\mu = \int_{\Sigma_1} d^n \abs g^{1/2} n_\mu j^\mu ,
\end{equation}

where both $n^\mu$ vectors point in the \emph{same} direction. As a convention we choose $n^\mu$ to point in the future direction. We therefore see that the quantity

\begin{equation}
  \int_\Sigma d^nx \abs g^{1/2} n_\mu j^\mu[\phi,\psi]
\end{equation}

has the same value in all spacelike hypersurfaces\ic{hypersurface!spacelike}. We would like to define an inner product\ic{inner product} $(\phi,\psi)$ proportional to this integral. However, we require $(\phi,\psi)^* = (\psi,\phi)$. In case of the $j^\mu$ given above, this is not satisfied: because of a factor of negative one. If we define $j^\mu = i j^\mu_{\text{(old)}}$, we get a pretty good inner product. One more subtlety is that we would like complex factors that multiply $\psi$ to pass outside the inner product in the same way, whereas the ones that multiply $\phi$ are to be complex conjugated. This point is not quite important, however it carries on the custom of quantum mechanics where $\phi$ is regarded as a bra and $\psi$ as a ket. For that purpose we map $j^\mu_{\text{(old)}} \mapsto j^{*,\mu}_{\text{(old)}}$ and use $j^\mu = -i j^{*,\mu}_{\text{(old)}}$ as the current that defines the charge. The $U(1)$-inner product\ic{U(1)@$U(1)$-inner product!definition} obtained is as follows:

\begin{equation}
  \label{eq:20}
  (\phi,\psi) = i \int_\Sigma d^nx \abs g^{1/2} n_\mu (\psi \nabla^\mu \phi^* - \phi^* \nabla^\mu \psi).
\end{equation}

Careful readers may have already noticed that we used complex conjugation in defining the inner product, even though we were interested in real scalar fields. The field $\phi$ is indeed required to be real valued, however, we use the symbols above (\ie $\phi$ and $\psi$ in (\ref{eq:20})) simply to denote \emph{any} solution of the field equation. Therefore, if the basis we choose for this function space, that includes all the field solutions, includes some complex valued functions, then the appearance of complex conjugation operation is necessary.

We list basic properties of this inner product:

\begin{enumerate}[(i)]
\item $(\phi,\psi)^* = (\psi,\phi)$
\item $(\psi^*,\phi^*) = - (\phi,\psi)$
\item $\forall a,b \in \comps, (a\phi + b\theta,\psi) = a^*(\phi,\psi) + b^*(\theta,\psi)$
\item $\forall a,b \in \comps, (\phi, a\psi + b\chi) = a(\phi,\psi) + b(\phi,\chi)$
\end{enumerate}

\chapter{Bogoliubov Coefficients for Hawking Radiation}
\label{chap:bog-coef-Hawking}

\section{Continuum Case}
\label{sec:continuum-case}

Bogoliubov coefficients, in general, appear in the following expansion:

\begin{equation}
  \fout_\omega = \int d\omega' \; (\alpha_{\omega\omega'} \fin_{\omega'} + \beta_{\omega\omega'} \fins_{\omega'}).\label{eq:43}
\end{equation}

Let us remember that we have $\fin_\omega = (4\pi\omega)^{-1/2} e^{-i\omega v}/r$. As mentioned in the text, we no longer write the angular dependency explicitly; otherwise $\fin_\omega$ must be multiplied by $Y_{lm}(\theta,\phi)$. If we multiply (\ref{eq:43}) by $e^{\pm i\Omega v}$ (where $\Omega > 0$) and integrate\footnote{Note that $\fout_\omega$ has no support on $\mc I^-$ for $v > v_H$. We integrate over $v \in \reals$ and use $\fout_\omega = 0$ for $v > v_H$.} over $v$, we can isolate $\alpha$ for the plus sign and $\beta$ for the minus sign. Rearranging the both sides of the obtained relations and using the form of $\fout$ near $\mc I^-$ as given in (\ref{eq:11}) we obtain:

\begin{align}
  \alpha_{\omega\omega'} &= \frac{1}{2\pi} \left(\frac{\omega'}{\omega}\right)^{1/2} \int_{-\infty}^{v_H} dv\; e^{i\omega'v} e^{-i\omega \uout(v)},\\
  \beta_{\omega\omega'}  &= \frac{1}{2\pi} \left(\frac{\omega'}{\omega}\right)^{1/2} \int_{-\infty}^{v_H} dv\; e^{-i\omega'v} e^{-i\omega \uout(v)}.
\end{align}

The incoming wave from $\mc I^-$ for which $v=v_H$ holds, does not reach $\mc I^+$ but intersects the event horizon. Therefore in the limit $v \to v_H^-$, $\uout$ must approach to $\infty$. Using $v$ in lieu of $\uin$ in equation (\ref{eq:uout-intermsof-uin}) we see that $v_H = v_0 - 4M$ should hold. If we use $v_H$ instead of $v_0$ we have the following:

\begin{equation}
  \uout(v) = v - 4M \ln\left( \frac{v_H-v}{4M} \right).
\end{equation}

Using this form for $\uout(v)$ and performing changes of variables in integrals, we obtain:

\begin{align}
  \alpha_{\omega\omega'} &= \frac{e^{i(\omega'-\omega)v_H}}{2\pi} \left(\frac{\omega'}{\omega}\right)^{1/2} \int_0^{\infty} dv\; e^{-i\sigma\abs{\omega'-\omega}v} \exp[i 4 \omega M \ln(v/4M)],\\
  \beta_{\omega\omega'}  &= \frac{e^{-i(\omega'+\omega)v_H}}{2\pi} \left(\frac{\omega'}{\omega}\right)^{1/2} \int^0_{-\infty} dv\; e^{-i(\omega'+\omega)v} \exp[i 4 \omega M \ln(-v/4M)],
\end{align}

where we have written $\omega'-\omega$ as $\sigma \abs{\omega'-\omega}$ for $\sigma = \sgn(\omega'-\omega)$ because when evaluating this integral after a Wick rotation\ic{Wick rotation}, whether $\sigma$ equals plus or minus one will be important in choosing how to close the contour. Figure~\ref{fig:bog-coef-contour} illustrates appropriate choice of contour in each case.

\begin{figure}
  \centering
  \begin{tikzpicture}[scale = 1, thick]
  \begin{scope}
      \draw[thin] (-2,0) -- (2,0);
      \draw[thin] (0,-2) -- (0,2);
      \draw[ultra thick,postaction={decorate, decoration={markings, mark=between positions 0.5 and 0.5 step 0.5 with {\arrow{stealth}}}}] (0,0) -- (2,0);
      \draw[ultra thick,postaction={decorate, decoration={markings, mark=between positions 0.5 and 0.5 step 0.5 with {\arrow{stealth}}}}] (0,2) -- (0,0);
      \draw[ultra thick,postaction={decorate, decoration={markings, mark=between positions 0.5 and 0.5 step 0.5 with {\arrow{stealth}}}}] (2,0) arc (0:90:2);
      \draw (0,2.2) node [anchor=south west] {$\!\!\!\alpha; \sigma=-$};
  \end{scope}
  
  \begin{scope}[xshift=5cm]
      \draw[thin] (-2,0) -- (2,0);
      \draw[thin] (0,-2) -- (0,2);
      \draw[ultra thick,postaction={decorate, decoration={markings, mark=between positions 0.5 and 0.5 step 0.5 with {\arrow{stealth}}}}] (0,0) -- (2,0);
      \draw[ultra thick,postaction={decorate, decoration={markings, mark=between positions 0.5 and 0.5 step 0.5 with {\arrow{stealth}}}}] (0,-2) -- (0,0);
      \draw[ultra thick,postaction={decorate, decoration={markings, mark=between positions 0.5 and 0.5 step 0.5 with {\arrow{stealth}}}}] (2,0) arc (0:-90:2);
      \draw (0,2.2) node [anchor=south west] {$\!\!\!\alpha; \sigma=+$};
  \end{scope}

  \begin{scope}[xshift=10cm]
      \draw[thin] (-2,0) -- (2,0);
      \draw[thin] (0,-2) -- (0,2);
      \draw[ultra thick,postaction={decorate, decoration={markings, mark=between positions 0.5 and 0.5 step 0.5 with {\arrow{stealth}}}}] (-2,0) -- (0,0);
      \draw[ultra thick,postaction={decorate, decoration={markings, mark=between positions 0.5 and 0.5 step 0.5 with {\arrow{stealth}}}}] (0,0) -- (0,-2);
      \draw[ultra thick,postaction={decorate, decoration={markings, mark=between positions 0.5 and 0.5 step 0.5 with {\arrow{stealth}}}}] (0,-2) arc (-90:-180:2);
      \draw (0,2.2) node [anchor=south west] {$\!\!\!\beta$};
  \end{scope}  
\end{tikzpicture}\vspace{1em}
  \caption{Contours on the complex plane to be considered for Wick rotating the integrals present in expressions for Bogoliubov coefficients. Integrals on the quarter circles vanish.}
  \label{fig:bog-coef-contour}
\end{figure}
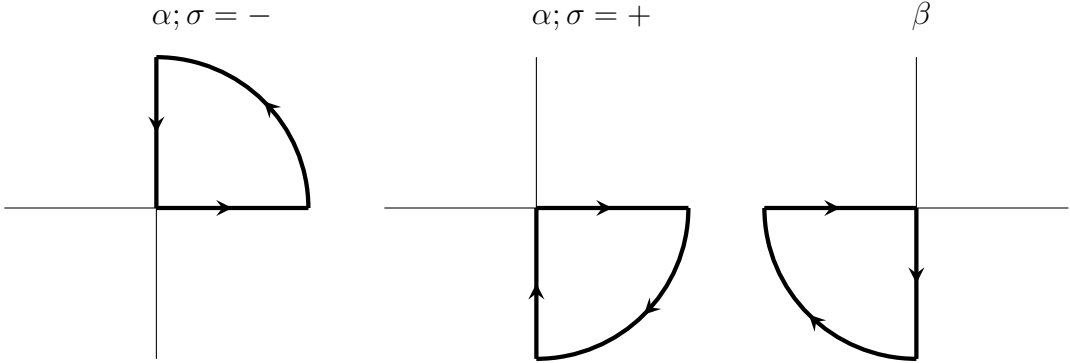

Integrals on the quarter circles vanish, which allows us to Wick rotate the variable $v$. The resulting integrals can be evaluated easily in terms of $\Gamma$ functions. The results are found as follows:\ic{Bogoliubov coefficients!for Hawking radiation}

\begin{align}
  \alpha_{\omega\omega'} &= \frac{-i\sigma}{2\pi} \left(\frac{\omega'}{\omega}\right)^{1/2} e^{2\pi\sigma\omega M} \frac{e^{i(\omega'-\omega)v_H} \Gamma(1+i4\omega M)}{(4M)^{i4\omega M} \abs{\omega'-\omega}^{1+i4\omega M}},\\
  \beta_{\omega\omega'} &= \frac{i}{2\pi} \left(\frac{\omega'}{\omega}\right)^{1/2} e^{-2\pi\omega M} \frac{e^{-i(\omega'+\omega)v_H} \Gamma(1+i4\omega M)}{(4M)^{i4\omega M} (\omega'+\omega)^{1+i4\omega M}}.
\end{align}

We can further simplify these by noting $\Gamma(1+ix) = [\pi x / \sinh(\pi x)]^{1/2}$:

\begin{align}
  \alpha_{\omega\omega'} &= -i\sigma \left( \frac{\omega' M}{\pi} \right)^{1/2} \frac{e^{\sigma 2\pi \omega M}}{\sinh^{1/2}(4\pi \omega M)} \frac{e^{i(\omega' - \omega)v_H}}{(4M)^{i4\omega M} \abs{\omega' - \omega}^{1+i4\omega M}},\\
  \beta_{\omega\omega'} &= i \left( \frac{\omega' M}{\pi} \right)^{1/2} \frac{e^{- 2\pi \omega M}}{\sinh^{1/2}(4\pi \omega M)} \frac{e^{-i(\omega' + \omega)v_H}}{(4M)^{i4\omega M} (\omega' + \omega)^{1+i4\omega M}}.
\end{align}

\section{Semi-discrete Case}
\label{sec:semi-discrete}

We defined wave packets $g_{jn}$ in the text as follows:

\begin{equation}
  g_{jn} \equiv \ve^{-1/2} \int_{j\ve}^{(j+1)\ve} d\omega \; e^{-i2\pi n \omega/\ve} \fout_\omega, \quad j \in \mathbb{Z}^{\geq 0}, n \in \mathbb{Z},
\end{equation}

and expanded these in terms of in-modes:

\begin{equation}
  g_{jn} = \int d\omega \; (\alpha_{jn\omega} \fin_\omega + \beta_{jn\omega} \fins_\omega).
\end{equation}

Of course, we determine semi-discrete Bogoliubov coefficients using the $U(1)$-inner product:

\begin{align}
  \alpha_{jn\omega'} &= (\fin_{\omega'},g_{jn}), & \beta_{jn\omega'} &= -(\fins_{\omega'},g_{jn}).
\end{align}

By using the explicit form of $g_{jn}$, one can relate these semi-discrete coefficients to the old ones:

\begin{align}
  \gamma_{jn\omega'} = \ve^{-1/2} \int_{j\ve}^{(j+1)\ve} d\omega \; e^{-i2\pi n \omega/\ve} \gamma_{\omega\omega'},
\end{align}

where $\gamma$ stands for $\alpha$ or $\beta$. Explicitly:

\begin{align}
  \alpha_{jn\omega'} &= -i\sigma h(\omega';\ve) \int_{j\ve}^{(j+1)\ve} \negmedspace\negmedspace d\omega \;
     \frac{e^{\sigma 2\pi \omega M}}{\sinh^{1/2}(4\pi\omega M)}
     \frac{e^{-i(v_H + 2\pi n /\ve)\omega}}{(4M)^{i4\omega M} \abs{\omega' - \omega}^{1+i4\omega M}},\\
   \beta_{jn\omega'} &= i h^*(\omega';\ve) \int_{j\ve}^{(j+1)\ve} \negmedspace \negmedspace d\omega \; 
     \frac{e^{- 2\pi \omega M}}{\sinh^{1/2}(4\pi\omega M)}
     \frac{e^{-i(v_H + 2\pi n /\ve)\omega}}{(4M)^{i4\omega M} (\omega' + \omega)^{1+i4\omega M}},
\end{align}

where $h(\omega';\ve) = e^{i\omega' v_H} (\omega' M / \pi \ve)^{1/2}$.

\chapter{The Entangled Nature of Vacuum}
\label{chap:entangledvacuum}

In section~\ref{sec:hawking-radiation} about the Hawking radiation, even though we calculated the spectrum of radiation we never attempted to expand the field in the out-modes. In order to fully expand the field, we need outgoing modes that reach $\mc I^+$ as well as the modes that fall through the horizon. Because event horizon together with the lightlike future is a Cauchy surface, we should be able to express any in-state in the Fock space of horizon modes and modes that reach $\mc I^+$.

The in-vacuum when expressed in late times, corresponds to an entangled state. Horizon degrees of freedom are entangled with outgoing Hawking radiation. We will not derive this well known result from scratch. Readers may like to see \cite{Fabbri2005, HossenfelderHawkingRad, FrolovNovikovBHPhysics}.

We, however, give a sketch of an idea of how it can be done. In order to express the in-vacuum in late times one acts on it with the identity operator: $\invac = \one \invac = \sum_\psi \ket{\psi}\braket{\psi \mid \text{in}}$ where the sum is over the Fock basis in late times. The state $\ket \psi$ is proportional to $a_{k_1}^\dagger a_{k_2}^\dagger \cdots a_{k_n}^\dagger \outvac$ for some $k_1,k_2,\ldots,k_n$ where $a_{k}^\dagger$ is a creation operator in late times. For now, we do not distinguish horizon modes from Hawking modes: $a_k^\dagger$ may be related to any one or both of them. By using Bogoliubov coefficients\ic{Bogoliubov coefficients} we can write: $\ain_k = \int dl \; (\alpha_{kl} a_l + \beta_{kl} a_l^\dagger$). Applying both sides to $\invac$ we obtain: $\int dl \; \alpha_{kl} a_l \invac = -\int dl \; \beta_{kl} a_l^\dagger \invac$. Multiplying both sides by\footnote{For the existence of $\alpha^{-1}_{k'k}$ readers may see a footnote in \cite{Fabbri2005} on page 79 which refers to Wald's\ip{Wald, Robert M.} \cite{WaldParticleBH}.} $\alpha_{k'k}^{-1}$, integrating over $k$ and renaming variables we obtain: $a_k \invac = -\int dk' dl \; \alpha_{kk'}^{-1}\beta_{k'l} a_l^\dagger \invac$. If we define $V_{kl} \equiv -\int dk' \; \alpha_{kk'}^{-1} \beta_{k'l}$, we can rewrite the previous expression as:

\begin{equation}
  a_k \invac = \int dl \; V_{kl} a_l^\dagger \invac .
\end{equation}

Hence the product $\braket{\psi \mid \text{in}}$ is proportional to $\int dl \;  V_{k_n l} \bra{\text{out}} a_{k_1} a_{k_2} \cdots a_{l}^\dagger \invac$. The product of operators inside can be replaced by the commutator

\begin{equation}
  [a_{k_1} a_{k_2} \cdots ,a_{l}^\dagger],
\end{equation}

because $a_l^\dagger$ annihilates $\bra{\text{out}}$. This commutator can be calculated\footnote{Readers may find the worm rule of commutators\ic{worm rule of commutators} presented in Appendix~\ref{chap:wormrule} useful.} and can be expressed as a sum of $n-1$ terms each containing an expectation value of $n-2$ operators. Recursively one will in the end reach a summation of integrals of various products of $V_{kl}$, where the total expression will be multiplied by $\braket{\text{out}\mid \text{in}}$.  By the way, this argument explains why an in-vacuum contains particles in \emph{pairs}. If we began with an odd number of operators, in the end only one operator would be left in between $\bra{\text{out}} \cdot \invac$ and it would vanish. Hence, the innerproduct between states in late times that contain odd number of particles and in-vacuum is zero: there is no transition to these states. Afterwards a nice expression of $\invac$ in late Fock basis can be obtained:

\begin{equation}
  \invac = \braket{\text{out} \mid \text{in}} \exp\left( \frac{1}{2} \int dk dl \; V_{kl} a_k^\dagger a_l^\dagger\right) \outvac .
\end{equation}

Reference \cite{Fabbri2005} may be seen for more discussion. Fabbri\ip{Fabbri, Alessandro} and Navarro-Salas\ip{Navarro-Salas, José} \cite{Fabbri2005} use out-modes that are the same with ours apart from a factor of $(4\pi)^{1/2}$ and define infalling horizon modes using the form of $\fout_\omega$ near $\mc I^-$: infalling horizon modes and outgoing Hawking modes are separate from each other and do not mix. We quote their result, while making clear that their summation is integration and explicitly writing angular dependencies, which is \cite{Fabbri2005}:

\begin{equation}
  \invac = \braket{\text{out} \mid \text{in}} \exp\left( \sum_{lm} \int_0^\infty d\omega \; e^{-4\pi M \omega} \aintd_{\omega lm} \aoutd_{\omega lm}\right) \outvac  , \label{eq:55}
\end{equation}

where a superscript `int' is used to denote horizon modes, the indices $l$ and $m$ denote angular properties, $a_k^{\text{out}}$ is the same as before and $\outvac$ is now the vacuum that is annihilated by both $a_k^{\text{int}}$ and $a_k^{\text{out}}$.

For a stationary observer\footnote{I would like to thank Sabine Hossenfelder\ip{Hossenfelder, Sabine} who turned my attention to this point and made myself aware that Hawking radiation and Unruh radiation has the same entanglement structure.} located at fixed Schwarzschild $r$ coordinate, the appropriate vacuum state is $\outvac$. This is because his proper time is proportional to Schwarzschild time that is used to define positive frequency modes that is appropriate for distant observers. The only difference is that the detected Hawking quanta will be blue shifted by an amount $(1-2M/r)^{-1/2}$. Stationary observers will observe increased temperatures as the stretched horizon is approached.

On the other hand, a freely falling observer will rather see a local vacuum while passing through the event horizon. Let $p$ be the event that the infalling geodesic intersects with the horizon. Expressing the local spacetime around $p$ in Riemann normal coordinates\ic{Riemann normal coordinates} one can see that it is locally flat, which is required by the manifold structure of spacetime. In this local region, the worldline of the observer is almost a straight line. Please see Figure~\ref{fig:localinfalling}.

\begin{figure}
  \centering
  \begin{tikzpicture}[scale = 2] 
  \draw[dashed] (-2,-2) -- (2, 2) node [near end, above, sloped] {Event horizon};
  \draw[dashed] (-2, 2) -- (2,-2);
  \draw[->] (-2,0) -- (2,0) node [anchor=west] {$\rho$};
  \draw[->] (0,-2) -- (0,2) node [anchor=south] {$\tau$};
  
  \draw[thin, domain=-0.5:1.5] plot (\x, {1-2*\x});
  
  \draw[thin, domain=-1.31695789692:1.31695789692] plot ({1.0*cosh(\x)}, {1.0*sinh(\x)});
\end{tikzpicture}\vspace{1em}
  \caption{Illustration of the locally flat spacetime region that is seen by an infalling observer. Here $\tau$ and $\rho$ are the locally flat temporal and radial coordinates, respectively. The straight line describes the worldline of the infalling observer, whereas the hyperbola does that of the stationary observer. Event horizon is juxtaposed on one of Rindler horizons.}
  \label{fig:localinfalling}
\end{figure}
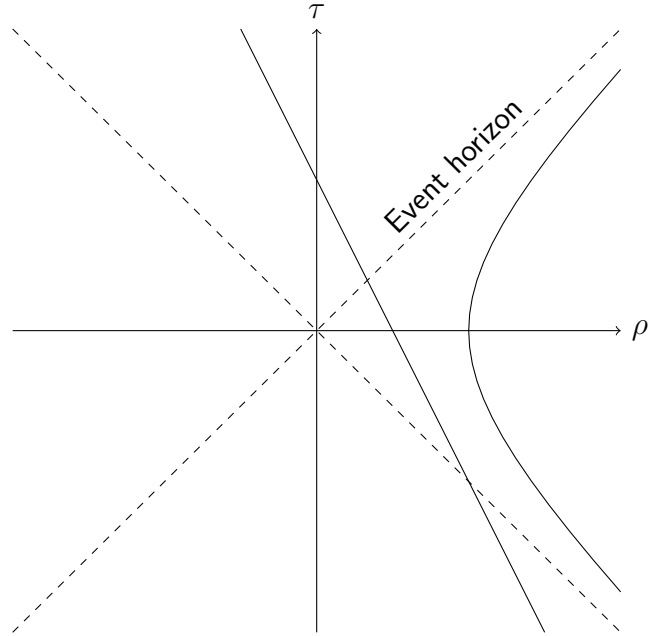

We described the motion of infalling observers, let us now turn our atention to stationary observers. A stationary observer who is located at $r \approx 2M$ does not follow a geodesic: his motion is accelerated. What is more, his worldline will turn out to match that of a Rindler observer.

We begin with the Schwarzschild metric:

\begin{equation}
  ds^2 = -\left(1- \frac{2M}{r}\right) dt^2 + \left(1- \frac{2M}{r}\right)^{-1} dr^2 + r^2 d\Omega^2 \label{eq:54} .
\end{equation}

We are interested in regions close to event horizon, for that purpose, we let $r = 2M (1+x)$ where $x \ll 1$. In this approximation, (\ref{eq:54}) becomes:

\begin{equation}
  ds^2 \approx -x dt^2 + \frac 1 x dr^2 + r^2 d\Omega^2 .
\end{equation}

Defining $x = \xi^2 / 16 M^2$, and $\eta = t/4M$ we obtain:

\begin{equation}
  ds^2 \approx -\xi^2 d\eta^2 + d\xi^2 + r^2 d\Omega^2 .
\end{equation}

The first two terms constitute the Rindler metric. Reference \cite{BergmannCollapse} exhibits similar calculations. In order to obtain the form we used in section~\ref{sec:unruh-radiation} one may let $\xi = e^{a\xi'} / \xi'$, however this is not necessary. What we observe is that constant $r$ worldlines corresponds to constant $\xi$ worldlines. These are the worldlines of Rindler observers.

The Hawking radiation detected by the stationary observer is indeed\footnote{The question of why a distant observer encounters Hawking radiation, when he is not accelerated at all, arises. The answer is the existence of \emph{event horizon}.} the radiation seen by a Rindler observer: Unruh radiation\ic{Unruh radiation}.

Here is the crux of the entanglement conflict. According to semiclassical calculation, the state of the black hole vapor is expressed by (\ref{eq:55}):

\begin{equation}
  \invac = \braket{\text{out} \mid \text{in}} \exp\left( \sum_{lm} \int_0^\infty d\omega \; e^{-4\pi M \omega} \aintd_{\omega lm} \aoutd_{\omega lm}\right) \outvac . \tag{\ref{eq:55}}
\end{equation}

However if the evaporation is unitary, black vapor must begin being purified after the Page time\ic{Page time} \cite{PageBHInfo}: it cannot be in the above state. If it is not in the above state at least for the modes that are emitted when the black hole is old as argued by AMPS \cite{Almheiri2012} and they are almost maximally entangled with early radiation, then they cannot be entangled with interior modes to yield the Minkowski vacuum when traced backwards in time to the near horizon region: because these modes are Rindler modes around the horizon.

\chapter{Two Proofs of the No-cloning Theorem}
\label{chap:proof-no-cloning}

Here we provide two proofs of the no-cloning theorem. First one shows that cloning of arbitrary states violate the linearity of quantum mechanics, whereas the second one obtains a contradiction using a more general approach known as the generalized measurement formalism.

\section{Linearity of Quantum Mechanics}
\label{sec:no-clone-linear-qm}

Here we provide a proof that reach its conclusion by showing that an operation that clones arbitrary quantum states is non-linear. Because operators in quantum mechanics are linear, it is concluded that cloning is not allowed in QM.

The cloning operation on states is of the following form:

\begin{equation}
  \label{eq:30}
  \ket \psi \otimes \ket \phi \mapsto \ket \psi \otimes \ket \psi ,
\end{equation}

where $\ket \psi$ is the state to be copied and $\ket \phi$ is some standard state, for example the state in which the storage unit is. Because quantum mechanics is linear, there must be a linear operator $L$ that performs the transformation of $\ket \psi \otimes \ket \phi$ into $\ket \psi \otimes \ket \psi$:

\begin{equation}
  \label{eq:31}
  L \ket \psi \otimes \ket \phi = \ket \psi \otimes \ket \psi .
\end{equation}

Because superposed states must be cloned by this operation as well, a contradiction can be obtained as in reference \cite{Susskind2005} whose argument is as follows. First of all, a cloning operation must perform the following:

\begin{align}
  \ket{+} \otimes \ket \phi &\rightarrow \ket + \otimes \ket + ,\\
  \ket{-} \otimes \ket \phi &\rightarrow \ket - \otimes \ket - ,
\end{align}

where we use $\pm$ notation instead of $\uparrow\downarrow$ notation for spins. On the other hand, the following two expressions must be true as well, because $L$ is a linear operator:

\begin{align}
  \frac 1 {\sqrt 2} (\ket + + \ket -) \otimes \ket \phi &\rightarrow \frac 1 {\sqrt 2} (\ket + \otimes \ket + + \ket - \otimes \ket -).\\
  \frac 1 {\sqrt 2} (\ket + + \ket -) \otimes \ket \phi &\rightarrow \frac 1 {\sqrt 2} (\ket + + \ket -) \otimes \frac 1 {\sqrt 2} (\ket + + \ket -),\\
  &= \frac 14 \ket + \otimes \ket + + \frac 14 \ket + \otimes \ket - ,\\
  &\quad + \frac 14 \ket - \otimes \ket + + \frac 14 \ket - \otimes \ket - .
\end{align}

Because the states $\ket +$ and $\ket -$ are linearly independent, the contradiction follows.

\section{Generalized Measurement Approach}
\label{sec:no-clone-gen-measure}

The proof by generalized measurement approach is done in one of Sadi Turgut's\ip{Turgut, Sadi} lectures on quantum information\footnote{I would like to thank Ümit Alkuş\ip{Alkuş, Ümit} for sharing his lecture notes about the proof of the theorem. What follows below is based on these notes.}.

We suppose there exist generalized measurement operators $\{M_i\}_i$ such that:

\begin{equation}\label{eq:35}
  \sum_i M_i^\dagger M_i = \one .
\end{equation}

We would like these operators to clone some given state $\ket \psi$:

\begin{equation}\label{eq:33}
  M_i \ket \psi \otimes \ket 0 = c_i \ket \psi \otimes \ket \psi ,
\end{equation}

where $c_i$'s depend on $\ket \psi$. On the other hand, (\ref{eq:35}) causes $c_i$'s to satisfy:

\begin{equation}
  \sum_i \abs{c_i}^2 = 1.  
\end{equation}

If this operation is able to copy the state $\ket \psi$, it should also do the same for another state $\ket{\psi'}$:

\begin{equation}\label{eq:34}
  M_i \ket{\psi'} \otimes \ket 0 = c_i' \ket{\psi'} \otimes \ket{\psi'}.
\end{equation}

Because $c$'s depend on the state to be cloned, $c_i'$ may differ from $c_i$: they are not necessarily the same. We apply $\dagger$-operation to (\ref{eq:33}) and apply it to (\ref{eq:34}). Then, we sum over $i$:

\begin{align}
  \sum_i \bra{\psi} \otimes \bra 0 M_i^\dagger M_i \ket{\psi'} \otimes \ket 0 &= \sum_i c_i^* c_i' \bra{\psi} \otimes \bra{\psi} \cdot \ket{\psi'} \otimes \ket{\psi'},
  \intertext{We simplify the left hand side by using (\ref{eq:35}) and evaluate the inner-products:}
  \braket{\psi|\psi'} &= \sum_i c_i^* c_i' \braket{\psi|\psi'}^2,
  \intertext{Let us assume that $\ket{\psi'}$ and $\ket \psi$ are non-orthogonal: $\braket{\psi|\psi'} \neq 0$:}
  \braket{\psi|\psi'}^{-1} &= \sum_i c_i^* c_i' .\label{eq:36}
\end{align}

By the Cauchy-Schwarz inequality, the following holds:

\begin{equation}
  \label{eq:cici'-ineq}
  \abs{\sum_i c_i^* c_i'} \leq \sum_i \abs{c_i}^2 \cdot \sum_i \abs{c_i'}^2.
\end{equation}

Because $\sum_i \abs{c_i}^2 = \sum_i \abs{c_i'}^2 = 1$, we have:

\begin{equation}
  \abs{\sum_i c_i^* c_i'} \leq 1.
\end{equation}

Using this inequality in (\ref{eq:36}), we obtain:

\begin{equation}
  \abs{\braket{\psi|\psi'}} \geq 1 .
\end{equation}

Because we have chosen $\ket{\psi'}$ different from $\ket\psi$, we obtain $\abs{\braket{\psi|\psi'}} > 1$: This is a contradiction. Therefore, cloning is not allowed in quantum mechanics.

\chapter{The Worm Rule of Commutators}\ic{worm rule of commutators|textbf}
\label{chap:wormrule}

\begin{figure}[h]
  \centering
  \includegraphics[width=10cm]{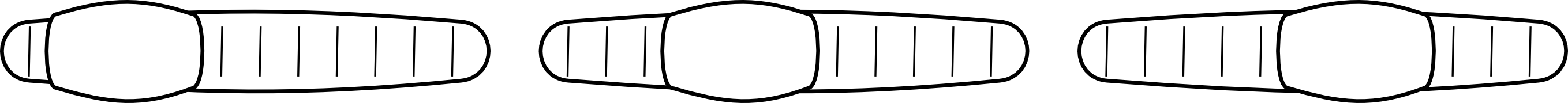}\vspace{1em}
  \caption{A simplistic illustration of how an earthworm moves.}
  \label{fig:earthworm}
\end{figure}

This is a very nice property of operations $[\cdot,\cdot]$ that satisfy $[a,bc] = [a,b]c + b[a,c]$. Commutators in quantum mechanics and Poisson brackets\ic{Poisson bracket} of classical mechanics are two examples. We express this as a theorem:

\begin{thm}
  Let $[\cdot,\cdot]$ be an operation such that for all $a,b,c$ the equality $[a,bc] = [a,b]c + b[a,c]$ holds. Then the following is true:

  \begin{equation}
    [a,bcde\ldots] = [a,b]cde\ldots + b[a,c]de\ldots + bc[a,d]e\ldots + bcd[a,e]\ldots + \cdots\label{eq:52}
  \end{equation}

\end{thm}

  The niceness of the pattern is that the \emph{order} of the sequence $bcde\ldots$ remains the same while $[a,\cdot]$ is applied to each term. Since the sequence of terms that are summed in (\ref{eq:52}) resembles how an earthworm moves, we call this relation \emph{the worm rule of commutators}\ic{worm rule of commutators!origin of name}. This property is not very hard to notice, therefore it might have appeared elsewhere; however we would like to call it that way.

If, as it is for a quantum mechanical commutator, $[a,b] = -[b,a]$, then the reverse relation ($[bcd,a] = [b,a]cd + b[c,a]d + bc[d,a]$) can be shown to hold as well, which we will not do it explicitly.

\begin{proof}
  Proof is by recursion. Let $n$ denote the number of terms on the right argument of the commutator. Theorem holds for $n=2$ by definition. Now suppose $n \geq 2$. We have an expression of the form $[a,b_1 \ldots b_n b_{n+1}]$. We consider $b_1 \ldots b_n$ as one term and use the theorem:

  \begin{equation}
    [a,b_1 \ldots b_n  b_{n+1}] = [a,b_1 \ldots b_n]  b_{n+1} + b_1 \ldots b_n [a,b_{n+1}]
  \end{equation}

The first term can be expanded to give a sum of $n$ terms where $[a,\cdot]$ is applied from $b_1$ to $b_n$. On the other hand the second term in the equation is the term that would appear when $[a,\cdot]$ is applied to $b_{n+1}$. In total, in each one of $n+1$ terms the order of $b_1 \ldots b_n b_{n+1}$ is preserved and $[a,\cdot]$ is applied to each one of them. This completes the proof.
\end{proof}

\newpage
\phantomsection
\addcontentsline{toc}{chapter}{Index}

\printindex[people][Index of People]
\printindex[concepts][Index of Concepts]

\end{document}